\documentclass[10pt]{article}
\usepackage[applemac]{inputenc}
\usepackage{latexsym,amsmath,amsthm,amscd,amssymb,framed,graphics,color,mathrsfs}
\usepackage{enumerate}
\usepackage{geometry}
\usepackage[colorlinks]{hyperref}
\usepackage[T1]{fontenc}
\usepackage{graphicx}
\usepackage{empheq}
\definecolor{sand}{rgb}{0.76, 0.7, 0.5}
\definecolor{taupegray}{rgb}{0.55, 0.52, 0.54}
\usepackage[colorlinks]{hyperref}
\usepackage{subcaption}
\hypersetup{
pdfstartview={FitH},
urlcolor=sand,
citecolor=blue,
linkcolor=taupegray,
}

\usepackage{soul}

\usepackage{array}
\usepackage{float}
\usepackage{multicol}
\usepackage{fancyhdr}
\usepackage[arrow, matrix, curve]{xy}
\usepackage{url}
\usepackage{colordvi}
\usepackage{framed}
\usepackage{color}
\usepackage{algorithmic}
\usepackage{algorithm}

\usepackage{enumitem}

\newcommand{\mathsym}[1]{{}}

\newcommand{\todo}[1]{\vspace{5 mm}\par \noindent
\marginpar{\qquad\textsc{ToDo}} 
\framebox{\begin{minipage}[c]{0.98 \textwidth}
\tt #1 \end{minipage}}\vspace{5 mm}\par}

\newtheorem{theorem}{Theorem}[section]
\newtheorem{definition}[theorem]{Definition}
\newtheorem{lemma}[theorem]{Lemma}
\newtheorem{remark}[theorem]{Remark}
\newtheorem{proposition}[theorem]{Proposition}

\begin{document}
\title{Discrete Dirac structures and discrete Lagrange--Dirac dynamical systems in mechanics}
\author{\hspace{-1cm}
\begin{tabular}{cc}
 Linyu Peng & Hiroaki Yoshimura
\\[2mm] \normalsize   Faculty of Science and Technology & \normalsize   School of Fundamental Science and Engineering
\\ \normalsize   Keio University & \normalsize  Waseda University 
\\  \normalsize  3-14-1 Hiyoshi, Kohoku, Yokohama, Japan\quad & \quad \normalsize  3-4-1 Okubo, Shinjuku, Tokyo, Japan
\\ \normalsize  l.peng@mech.keio.ac.jp  & \normalsize  yoshimura@waseda.jp\\
\end{tabular}\\\\
}

\maketitle

\date{}
\abstract{In this paper, we propose the concept of $(\pm)$-discrete Dirac structures over a manifold, where we define $(\pm)$-discrete two-forms on the manifold and incorporate discrete constraints using $(\pm)$-finite difference maps. Specifically, we develop $(\pm)$-discrete induced Dirac structures as discrete analogues of the induced Dirac structure on the cotangent bundle over a configuration manifold, as described by Yoshimura and Marsden \cite{YoMa2006a}. We demonstrate that $(\pm)$-discrete Lagrange--Dirac systems can be naturally formulated in conjunction with the $(\pm)$-induced Dirac structure on the cotangent bundle. Furthermore, we show that the resulting equations of motion are equivalent to the $(\pm)$-discrete Lagrange--d'Alembert equations proposed in \cite{CoMa2001, McPe2006}. We also clarify the variational structures of the discrete Lagrange--Dirac dynamical systems within the framework of the  $(\pm)$-discrete Lagrange--d'Alembert--Pontryagin principle. Finally, we validate the proposed discrete Lagrange--Dirac systems with some illustrative examples of nonholonomic systems through numerical tests.
}

\tableofcontents

\section{Introduction}

\paragraph{Variational integrators and nonholonomic integrators.}
It is well known that discrete Hamilton's variational principle is one of the structure-preserving numerical integrators, providing a long-time numerically stable scheme for conservative Lagrangian systems. This principle is based on the fact that the discrete Euler--Lagrange equations preserve a discrete symplectic structure, called the discrete Lagrangian two-form, along the discrete Lagrangian map, which is the discrete analogue of the flow of the Lagrangian vector field in the continuous setting.

On the other hand, when nonholonomic constraints are present, as in nonholonomic mechanical systems, the associated equations of motion can be derived using the Lagrange--d'Alembert principle. In the discrete setting, discrete Hamilton's principle is replaced by the discrete Lagrange--d'Alembert principle, from which the discrete Lagrange--d'Alembert equations can be obtained, as discussed in \cite{CoMa2001}. These algorithms, which generalize variational integrators for unconstrained Lagrangian systems, exhibit geometric properties similar to those of continuous nonholonomic systems. From a slightly different perspective, \cite{McPe2006} considered nonholonomic systems that admit reversing symmetries and developed integrators for such systems that preserve an analogous reversing symmetry. In these cases, the numerical integrator, referred to as a "nonholonomic integrator," no longer preserves the discrete symplectic structure. In fact, it is not yet fully understood what structure, if any, is preserved by such a nonholonomic integrator. Nevertheless, it remains a long-term numerically stable scheme for conservative nonholonomic mechanical systems.

\paragraph{Dirac structures in mechanics.} The notion of Dirac structures generalizes presymplectic and almost Poisson structures and was developed by \cite{Do1987, CoWe1988, Co1990a, Co1990b}. Dirac structures can directly incorporate constraint distributions into the framework of Hamiltonian systems. Initially, they were applied to constrained Hamiltonian systems, sometimes referred to as implicit Hamiltonian systems \cite{Co1990a, VaMa1995, BlCr1997, CeVaBa2003, vdSc2006}. Later, \cite{YoMa2006a} explored Dirac dynamical system on the Lagrangian side, focusing on Lagrangian systems with nonholonomic constraints. Specifically, they examined bundle structures over a configuration manifold $Q$, including Tulczyjew's triple $TT^{\ast}Q, T^{\ast}TQ, T^{\ast}T^{\ast}Q$, and introduced the Dirac differential $\mathbf{d}_{D}L: TQ \to T^{\ast}T^{\ast}Q$ for a given Lagrangian on $TQ$, possibly degenerate. This led to the definition of Lagrange--Dirac dynamical systems, which serves as the Lagrangian version of Dirac dynamical systems. The work included some illustrative examples of nonholonomic mechanical systems and electric circuits. Over the past years, various physical systems---such as continuum systems, and nonequilibrium thermodynamic systems as well as nonholonomic mechanical systems and electric circuits---have been formulated within the context of Dirac dynamical systems \cite{YoMa2006a, GBYo2015, GBYo2018a, GBYo2020}.

\paragraph{Variational structures of Dirac dynamical systems.} 
The variational structure of Lagrange--Dirac dynamical systems was further clarified by \cite{YoMa2006b} in the context of the Lagrange--d'Alembert--Pontryagin principle. As previously mentioned, for nonholonomic Lagrangian systems, the equations of motion can be derived from the Lagrange--d'Alembert principle. Specifically, this involves taking variations of the action integral for the Lagrangian, where one must impose variational constraints by choosing variations of curves in the configuration manifold so that the variations lie within the distribution and also satisfy the nonholonomic constraints on the motion of the system. The Lagrange--d'Alembert--Pontryagin principle, an extension of the Lagrange--d'Alembert principle to the Pontryagin bundle $TQ \oplus T^{\ast}Q$, recovers the equations of motion for Lagrange--Dirac dynamical systems. Additionally, the Hamiltonian analogue of Lagrange--Dirac dynamical systems was developed through the induced Dirac structures \cite{YoMa2006b}.

\paragraph{Discretization of Dirac structures.} 
Several approaches to discretizing Dirac structures have been proposed, such as those by \cite{LeOh2011} and \cite{CFTZ2022}. The former approach considered the discretization of induced Dirac structures on vector spaces together with their associated Lagrange--Dirac systems. This work aimed to develop the discrete analogue of the induced Dirac structures and the associated Lagrange--Dirac dynamical systems proposed in \cite{YoMa2006a}. By comparing with the continuous Tulczyjew's triple (see, e.g., \cite{Tu1977}), \cite{LeOh2011} defined a discrete triple, which was then used to introduce their notion of discrete Dirac structures. However, \cite{CFTZ2022}  later noted  that the discrete structures proposed by \cite{LeOh2011} do not constitute true Dirac structures. In contrast, \cite{CFTZ2022} proposed an alternative approach, constructing discrete Dirac structures on a discrete Pontryagin bundle over a configuration manifold, rather than on the cotangent bundle. This alternative approach aims to develop a general Dirac system instead of constructing the implicit Lagrangian and Hamiltonian systems. 

\paragraph{Our goals and the organization of this Paper.} 
The main goal of this paper is to propose an alternative discrete version of the induced Dirac structure on the cotangent bundle over a configuration manifold, enabling the construction of a discrete version of the Lagrange--Dirac dynamical system, consistent with the approach outlined in \cite{YoMa2006a}. We begin by presenting a discretization of continuous Dirac structures on a manifold $M$, defined by a two-form $\Theta_M$ and a distribution $\Delta_M$. This is achieved by introducing finite forward difference and backward difference maps. Next, we extend the concept of continuous Dirac structures to the induced Dirac structure on the cotangent bundle $M=T^\ast Q$. In this context, we introduce $(\pm)$-discrete symplectic forms and $(\pm)$-discrete constraint spaces, which are discrete analogues of the canonical symplectic forms and distributions on $Q$. 

To construct the discrete Lagrange--Dirac systems (or the discrete implicit Lagrangian systems) within the framework of the proposed discrete Dirac structures, we consider the bundle structures between $T^{\ast}T^{\ast}Q$, $T^{\ast}Q \times T^{\ast}Q$ and $T^{\ast}(Q \times Q)$. These structures serve as the discrete analogues of Tulczyjew's triple between $T^{\ast}T^{\ast}Q$, $TT^{\ast}Q$ and $T^{\ast}TQ$ in the continuous setting. We define the discrete flat maps $(\Omega^{d\pm}_{T^\ast Q})^\flat: T^\ast Q \times T^\ast Q \to T^{\ast}T^{\ast}Q$, which are skew-symmetric bundle maps over $T^{\ast}Q$ naturally induced by the discrete canonical symplectic structures $\Omega^{d\pm}_{T^\ast Q}$. We also introduce the discrete analogues of the Dirac differential of the Lagrangian, defined as the maps $d^{\pm}_DL: Q\times Q \to T^{\ast}T^{\ast}Q$. Furthermore, we present the $(\pm)$-associated discrete Lagrange--d'Alembert--Pontryagin principles, which yield the $(\pm)$-discrete equations of motion for the discrete Lagrange--Dirac dynamical systems. Finally, we show that the resulting $(\pm)$-discrete system equations can recover the discrete Lagrange--d'Alembert equations given in \cite{CoMa2001,McPe2006}, clarifying that the resulting nonholonomic integrators preserve the discrete induced Dirac structures. We then validate the proposed approach with numerical tests on two illustrative examples of nonholonomic systems. 
\medskip

This paper is organized as follows:
\begin{itemize}
\item \textrm{\S 2 and \S3}: We provide a brief review of variational formulations in mechanics, covering both continuous and discrete settings. 
\item \S4: We describe the continuous setting of Dirac structures and the associated Lagrange--Dirac systems, in which a Dirac structure is defined by a given distribution on a configuration manifold and the canonical symplectic structure on the cotangent bundle.
\item \S 5: We introduce the concept of $(\pm)$-discrete Dirac structures on a manifold. This includes defining discrete two-forms and  discrete constraint spaces over the manifold in conjunction with $(\pm)$-finite difference maps. In particular, we illustrate the discrete theory by presenting the discrete induced Dirac structures on the cotangent bundle of a configuration manifold.
\item \S6: We develop the notion of $(\pm)$-discrete Lagrange--Dirac dynamical systems within the framework of the $(\pm)$-discrete induced Dirac structures on the cotangent bundle. These formulations are also derived from the $(\pm)$-discrete Lagrange--d'Alembert--Pontryagin principles. We show that the resultant $(\pm)$-discrete Lagrange--d'Alembert--Pontryagin equations are equivalent to the system equations obtained from the $(\pm)$-discrete Lagrange--Dirac dynamical systems. 
\item \S7: We demonstrate that the discrete  Lagrange--d'Alembert equations developed by \cite{CoMa2001,McPe2006} can be recovered from the $(\pm)$-discrete Lagrange--d'Alembert--Pontryagin equations. This implies that the resulting $(\pm)$-nonholonomic integrators preserve the $(\pm)$-induced discrete Dirac structures.
\item \S8: We show the validity of our discrete theory by illustrative examples of a vertical rolling disk on a plane and a classical Heisenberg system through numerical tests.
\end{itemize}

\section{Variational formulation in Lagrangian mechanics}
In this section, we shall make a brief review on the variational formulations in Lagrangian mechanics, namely, the variational principle of Hamilton for unconstrained mechanical systems and the Lagrange--d'Alembert principle for nonholonomic mechanical systems. For further details, for instance, see \cite{MaRa1999,Bl2003,GoPoSa2000,MadeLeDadeDi1996}.

\subsection{Hamilton's variational principle}
\paragraph{The path space and variations of curves.}
Consider a mechanical system with a Lagrangian $L: TQ \to \mathbb{R}$, where $TQ$ is the tangent bundle of an $n$-dimensional  configuration manifold $Q$ with local coordinates $q^i, \, i=1,...,n$ for $q \in Q$. Consider a path space
\begin{equation}\label{PathSpace}
\mathcal{C} (Q) =\{ q:I=[0,T] \to Q \mid q \  \textrm{is a $C^{2}$ curve on $Q$ such that}\  q(0)=q_1\;\textrm{and}\;\ q(T)=q_2 \},
\end{equation}
where $I=[0,T] \subset \mathbb{R}^+$ is the space of time.
\medskip

A point $q$ in the manifold $\mathcal{C}(Q)$ is a curve on $Q$, namely, $q=q(t)$. The deformation of $q=q(t)\in \mathcal{C}(Q)$ is given by $q(t,\epsilon)=q_{\epsilon}(t)$ such that $q_{0}(t)=q(t,0)=q(t)$. Then, the variation of the curve $q(t)$ is defined by
\[
\delta{q}(t)= \frac{d}{d \epsilon}\biggr \arrowvert_{\epsilon =0} q_{\epsilon}(t),  
\]
which is the tangent vector to a curve $q(t)$. Let $\tau_{Q}: TQ \to Q; (q,\delta{q}) \mapsto q$ be the canonical projection and we get $\tau_Q \circ \delta{q}=q$. The restrictions $q_{\epsilon}(0)=q_1$ and $q_{\epsilon}(T)=q_2$ lead to $\delta{q}(0)=0$ and $\delta{q}(T)=0$ respectively.

\begin{figure}[htbp]
\begin{center}
\includegraphics[scale=.85, clip]{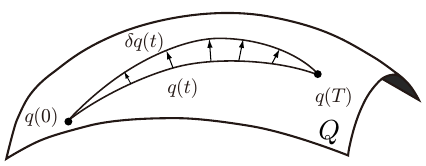}
\caption{Variations $\delta q(t)$ of a curve $q(t)$.}
\end{center}
\end{figure}

\paragraph{Hamilton's variational principle.}
Define the action functional $\mathfrak{S}: \mathcal{C} (Q)  \to \mathbb{R}$ by
\[
\mathfrak{S}(q)= \int ^{T}_{0} L \left( q(t), \dot{q}(t) \right) dt,
\]
where $\dot{q}(t)$ denotes the time derivative of $q(t)$.
If a curve $q=q(t) \in \mathcal{C} (Q)$ is a critical point of $\mathfrak{S}: \mathcal{C} (Q) \to \mathbb{R}$, namely,
$
\delta \mathfrak{S}(q)=0,
$
with the fixed endpoint conditions $\delta{q}(0)=\delta{q}(T)=0$,  the direct computation using local coordinates $q^i, \, i=1,...,n$ for $q \in Q$ yields
\begin{equation*}
\begin{split}
\delta \mathfrak{S}(q)&=\mathbf{d}\mathfrak{S}(q) \cdot \delta{q}= \frac{d}{d \epsilon}\biggr \arrowvert_{\epsilon =0} \mathfrak{S}(q_{\epsilon}(t) )\\
&= \int^T_0  \left(\frac{\partial L}{\partial q^i}-\frac{d}{dt}\frac{\partial L}{\partial \dot{q}^i}\right)\delta{q}^idt +  \frac{\partial L}{\partial \dot{q}^i} \;\delta{q}^i \bigg|_0^T\\
&= \int^T_0 \left(\frac{\partial L}{\partial q^i}-\frac{d}{dt}\frac{\partial L}{\partial \dot{q}^i}\right)\delta{q}^idt =0,
\end{split}
\end{equation*}
for all $\delta{q}^{i}$. In the above, the Einstein summation convention is employed; that is, a repeated index implies summation over that index. We shall use this convention throughout the paper unless stated otherwise. Thus we get the {\it Euler--Lagrange equations}:
\begin{equation}\label{EulerLagEqn}
\frac{d}{dt} \left( \frac{\partial L}{\partial \dot{q}^i} \right)- \frac{\partial L}{\partial q^i}=0,\quad i=1,...,n.
\end{equation}

\paragraph{The second-order vector fields.} 
Suppose the Lagrangian $L: TQ \to \mathbb{R}$ is {\it hyperregular}, namely, for every point $\dot{q} \in T_qQ$,
\begin{equation}\label{JacobLag}
\mathrm{det}\left[\frac{\partial^2{L}}{\partial{\dot{q}^i}\partial{\dot{q}^j}}\right] \neq 0.
\end{equation}

From the Euler--Lagrange equations \eqref{EulerLagEqn}, we get
\[
\ddot{q}^{j}=\left( \frac{\partial^{2} L}{\partial \dot{q}^i\partial \dot{q}^j}\right)^{-1}\left( \frac{\partial L}{\partial q^i} - \frac{\partial^{2} L}{\partial \dot{q}^i\partial q^k}\dot{q}^k\right),\quad j=1,...,n.
\]
In fact the above equations ensure that there exists a second-order vector field, called the Lagrangian vector field, denoted by $X_L: TQ \to \ddot{Q} \subset TTQ$, 
in which $\ddot{Q}$ is the {\it second-order submanifold} defined by
\[
\ddot{Q} := \left\{w \in TTQ \biggr{\arrowvert} T\tau_{Q}(w)=\tau_{TQ}(w) \right\},
\]
where $T\tau_{Q}: TTQ \to TQ; (q,\dot{q},\delta{q},\delta{\dot{q}}) \mapsto (q,\delta{q})$ and $\tau_{TQ}: TTQ \to TQ; (q,\dot{q},\delta{q},\delta{\dot{q}}) \mapsto (q,\dot{q})$. Hence, $T\tau_{Q}(w)=\tau_{TQ}(w)$ yields $\delta{q}=\dot{q}$ and therefore an element $w$ in the second-order submanifold $\ddot{Q}$ has the components $(q,\dot{q}, \ddot{q})$.
\paragraph{The Legendre transform.} Associated with $L$, recall from \cite{MaRa1999} that the Legendre transform $\mathbb{F}L: TQ \to T^{\ast}Q$ is a map, called the fiber derivative of $L$ at $v \in T_{q}Q$ along $w\in T_{q}Q$ that is defined by
\[
\mathbb{F} L(v) \cdot w = \left. \frac{d}{d \epsilon}\right \arrowvert_{\epsilon=0} L(v + \epsilon \,w).
\]
It is given in local coordinates by
\[
\mathbb{F} L(q^i,v^i) = \left(q^i, p_i=\frac{\partial L}{\partial v^i} \right).
\]
When $L$ is hyperregular, the Legendre transform $\mathbb{F}L: TQ \to T^{\ast}Q$ is globally diffeomorphic.
\paragraph{Energy conservation.}
Let $E_L: TQ \to \mathbb{R}$ be a Lagrangian energy defined by $E_L(q,\dot{q})=\mathbb{F}L(q,\dot{q}) \cdot \dot{q}-L(q,\dot{q})$ for $(q,\dot{q}) \in TQ$.
Along the solution curve $q(t) \in Q$ of the Euler--Lagrange equations \eqref{EulerLagEqn}, the energy $E_{L}$ is conserved as
\[
\begin{split}
\frac{dE_L}{dt}(q,\dot{q})&=\left< \frac{d}{dt} \left( \frac{\partial L}{\partial \dot{q}} \right)- \frac{\partial L}{\partial q}, \dot{q} \right>
=0.
\end{split}
\]
\paragraph{Canonical forms on the cotangent bundle.} 
On the other hand, the cotangent bundle $T^{\ast}Q$ of $Q$ is naturally equipped with the {\it canonical one-form} $\Theta_{T^{\ast}Q}\in \Lambda^1(T^\ast Q)$, which is defined by, for each $p_{q} \in T^{\ast}Q$,
\begin{equation}\label{CanOneForm_Cot}
\Theta_{T^{\ast}Q}(p_{q}) \cdot w_{p_{q}} = \langle p_{q}, \, T \pi_{Q} \cdot w_{p_{q} } \rangle, \;\; \textrm{for any $w_{p_{q} } \in T_{p_{q}}(T^{\ast}Q)$},
\end{equation}
where $\pi_{Q}: T^{\ast}Q \to Q; p_q=(q,p) \mapsto q$ is the cotangent bundle projection, and $\Theta_{T^{\ast}Q}$ is denoted in local coordinates by
$
\Theta_{T^{\ast}Q}=p_i\,dq^i.
$
Taking the exterior derivative yields the canonical two-form, namely, the {\it canonical symplectic structure} as
$
\Omega_{T^{\ast}Q}= -\mathbf{d}\Theta_{T^{\ast}Q} \in \Lambda^2(T^\ast Q)
$
with the coordinate expression
$
\Omega_{T^{\ast}Q}=dq^i \wedge dp_i.
$

\paragraph{Lagrangian forms on the tangent bundle.} 
The {\it Lagrangian one-form} and the {\it Lagrangian two-form (induced symplectic structure)} on $TQ$ can be introduced by using the Legendre transform as 
\begin{equation}\label{LagrangianForms}
\Theta_L = (\mathbb{F} L)^{\ast} \Theta_{T^{\ast}Q} \quad \text{and} \quad \Omega_L = (\mathbb{F} L)^{\ast} \Omega_{T^{\ast}Q}.
\end{equation}
Since the exterior derivative and the pull-back commute, it follows
$
\Omega_L=-\mathbf{d}\Theta_L.
$
The coordinate expressions of $\Theta_L$ and $\Omega_L$ are given by
\[
\Theta_L (q,v)= \frac{\partial L}{\partial v^{i}}dq^i \qquad \text{and}\qquad 
\Omega_L(q,v) = \frac{\partial^2 L}{\partial v^{i} \partial q^j}dq^i \wedge dq^j
+ \frac{\partial^2 L}{\partial v^{i} \partial  v^{j}}dq^i \wedge d v^j.
\]

\paragraph{Preservation of the Lagrangian two-form.}
Recall from \cite{MaRa1999} that the Euler--Lagrange equations are equivalently expressed by 
$\mathbf{i}_{X_{L}}\Omega_{L}=\mathbf{d}E_{L}$ and hence the Lie derivative $\pounds_{X_{L}}\Omega_{L}=\mathbf{i}_{X_{L}}\mathbf{d}\Omega_{L}+\mathbf{d}(\mathbf{i}_{X_{K}}\Omega_{L})=\mathbf{d}(\mathbf{i}_{X_{K}}\Omega_{L})=\mathbf{d}\mathbf{d}E_{L}=0$, where $\mathbf{i}$ denotes the interior product, and $\mathbf{d}\Omega_{L}=0$ since $\Omega_{L}$ is closed. Then, letting $\varphi_t: TQ \to TQ$ be the {\it flow} associated with the Lagrangian vector field $X_L$, we get the {\it preservation of the Lagrangian symplectic structure} along the {\it Lagrangian flow map:}
\[
\varphi^{\ast}_t \; \Omega_L = \Omega_L.
\]
As to the details, see \cite{MaRa1999}.

\subsection{The Lagrange--d'Alembert principle}\label{sect:LDAP}
\paragraph{Nonholonomic constraints.} Next we consider a mechanical system on an $n$-dimensional configuration manifold $Q$ with Lagrangian $L: TQ \to \mathbb{R}$, in which 
there exists a nonintegrable constraint distribution $\Delta_Q \subset TQ$ on $Q$.  In this paper, we assume that every distribution is {\it regular}, namely, it has constant rank at each point $q \in Q$ and is smooth unless otherwise stated. 

Choosing local coordinates $ q ^i,\;i=1,...,n $ for $q\in Q$ so that the configuration manifold $Q$ is locally represented by an open subset $U$ of $\mathbb{R}^n$. 
Now we suppose that the distribution $\Delta_Q \subset TQ$ is given by, for each $q \in Q$,
\begin{equation}\label{ConstraintDistribution}
\Delta_Q(q)=\left\{ \dot{q} \in T_{q}Q\,\mid \left< \omega^{r}(q),\dot{q}\right>=0,\,r=1,...,m <n\right\},
\end{equation}
where $\omega^{r}=\omega_{i}^{r}(q)dq^{i},\; r=1,...,m<n$ are some given $m$ independent one-forms on $Q$. We note that, in general, $\omega^{r}$ are not {\it completely integrable} in the sense of Frobenius; in other words, the constraints are nonholonomic.

\paragraph{The Lagrange--d'Alembert principle.}
For the case in which the nonholonomic constraints exist, Hamilton's principle is modified into the following variational formulation, called the Lagrange--d'Alembert principle. 
 A curve $q(t),\;t\in [0,T]$ is a critical point of the action integral $\mathfrak{S}: \mathcal{C} (Q)  \to \mathbb{R}$, namely,
\begin{equation*}  
\begin{split}
\delta\mathfrak{S}(q)= \delta \int_{0}^{T} L(q(t),\dot{q}(t))dt=0,
\end{split}
\end{equation*}
subject to the kinematic constraint $\dot{q}(t) \in \Delta_Q(q(t))$ as well as to the variational constraint
$\delta{q}(t) \in \Delta_Q(q(t))$, together with the fixed endpoint conditions. 

Then, the Lagrange--d'Alembert principle reads
\begin{equation*}  
\int^T_0  \left<\frac{\partial L}{\partial q}-\frac{d}{dt}\frac{\partial L}{\partial \dot{q}}, \,\delta{q} \right>dt=0,
\end{equation*}
for the chosen variations $\delta{q}(t) \in \Delta_Q(q(t))$. Then, the curve $q(t),\;t\in [0,T]$ satisfies  the {\it intrinsic Lagrange--d'Alembert equations:}
\begin{equation}\label{LagDAlembertEqn}
\frac{d}{dt} \left( \frac{\partial L}{\partial \dot{q}} \right)- \frac{\partial L}{\partial q} \in \Delta_Q^{\circ}(q), \quad \dot{q} \in \Delta_Q(q),
\end{equation}
where the annihilator $\Delta_Q^{\circ}\subset T^{\ast}Q$ of the distribution $\Delta_Q \subset TQ$ is defined by, for each $q \in Q$,
\[
\Delta_Q^{\circ}(q)=\left\{ \alpha \in T_{q}^{\ast}Q \mid \left< \alpha, v \right>=0, \;\textrm{for all}\;v \in \Delta_{Q}(q) \right\}.
\]
For an element $\alpha=\alpha_i dq^i$ in $\Delta_Q^{\circ}(q)\subset T_q^{\ast}Q$,  we have the local coordinate expression $\alpha_i=\mu_r \omega_i^r(q)$ by introducing Lagrange multipliers $\mu_r,\, r=1,...,m$. Therefore, we get the local coordinate expressions for the Lagrange--d'Alembert equations in \eqref{LagDAlembertEqn} as
\begin{equation}  \label{eq: local_LDA}
 \frac{d}{dt} \left( \frac{\partial L}{\partial \dot{q}^{i}} \right)- \frac{\partial L}{\partial q^{i}} =\mu_{r}\omega^{r}_{i}(q), \;\; \omega_{i}^{r}(q)\dot{q}^{i}=0,\quad i=1,...,n;\; r=1,...,m <n.
\end{equation}
For the unconstrained case in which $\Delta_{Q}=TQ$, the equations \eqref{eq: local_LDA} recover the Euler--Lagrange equations \eqref{EulerLagEqn}.

\paragraph{Energy conservation.}
Along the solution curve $q(t) \in Q$ of the Lagrange--d'Alembert equations in \eqref{LagDAlembertEqn}, the energy $E_{L}$ is conserved as
\[
\begin{split}
\frac{dE_L}{dt}(q,\dot{q})&=\left< \frac{d}{dt} \left( \frac{\partial L}{\partial \dot{q}} \right)- \frac{\partial L}{\partial q}, \dot{q} \right>
=0,
\end{split}
\]
since $\dot{q} \in \Delta_{Q}(q)$.

\paragraph{Failure of the preservation of symplecticity.}
The symplectic property of the  Lagrangian two-form {\it does not} hold anymore in the case that there exists nonholonomic constraints $\Delta_{Q} \subset TQ$. 
In fact, the Lagrange--d'Alembert equations are intrinsically expressed by $\mathbf{i}_{X_{L}}\Omega_{L}-\mathbf{d}E_{L} =\boldsymbol{\beta}$ together with ${X_{L}}_{\mid \Delta_{Q}} \in \Delta_{TQ}$ and $\boldsymbol{\beta} \in \Delta_{TQ}^{\circ}$, where $\Delta_{TQ}=(T\tau_{Q})^{-1}(\Delta_{Q})$ denotes the lifted distribution by using $\tau_{Q}: TQ \to Q$ and $\Delta_{TQ}^{\circ}$ its annihilator. Then it follows
\begin{equation}\label{LDA_Str}
\pounds_{X_{L}}\Omega_{L}=\mathbf{d}\boldsymbol{\beta},
\end{equation}
which implies that the Lagrangian symplectic two-form does not preserve; see also \cite{CoMa2001}.

\section{Variational discretization of Lagrangian systems}

The discretization of Hamilton's principle for Lagrangian systems was first shown by \cite{MoVe1991} and later the detailed applications to Lagrangian mechanics were extensively studied in \cite{MaWe2001}. Here, we  review briefly the variational discretization of Lagrangian systems, focusing on both unconstrained and nonholonomically constrained systems.

\subsection{Discrete Hamilton's principle}\label{Subsec:DisHamPrin}
First, we start with the geometric setting for discrete Hamilton's principle in order to formulate discrete Lagrangian systems by following \cite{MaWe2001}.
\paragraph{Geometric and variational setting for discrete mechanics.}
Let $Q$ be a configuration manifold and let $(q^1,...,q^n)$ be local coordinates for each point $q \in Q$. Consider a continuous curve 
$q: I \to Q$, where $I=\{t \in \mathbb{R} \mid 0 \le t \le T\}$ is the space of time interval, and then we define the increasing sequence of discrete points $I_{d}=\{t_{k}\}_{k=0}^{N}$ associated with $I$ by
\begin{equation*}
I_{d}:=\{ t_{k}=kh \in \mathbb{R} \mid k=0,...,N \in \mathbb{N} \},
\end{equation*}
where $h=t_{k+1}-t_{k} \in \mathbb{R}^{+}$ is a constant time step. 
Then, we introduce the \textit{discrete path space} by
\[
\mathcal{C}_d(Q)=\left\{q_{d}: \{t_{k}\}_{k=0}^{N} \to Q \right\}, 
\]
which is the discrete analogue of the continuous path space given in equation \eqref{PathSpace}.
Since the discrete path space $\mathcal{C}_d(Q)$ is isomorphic to $Q \times \cdots \times Q$ ($N +1$ copies), it
can be given by a smooth product manifold structure and hence we identify a \textit{discrete path} $q_d \in \mathcal{C}_d(Q)$ with its image 
\begin{equation*}
q_{d}=\{ q_k \}_{k=0}^N,
\end{equation*}
where $q_k:=q_d(t_k) \in Q$.

\begin{figure}[htbp]
\begin{center}
\includegraphics[scale=.7, clip]{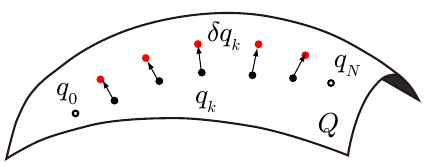}
\caption{Discrete sequences.}
\end{center}
\end{figure}

\paragraph{Finite difference maps.}
Consider a Lagrangian $L=L(q,\dot{q})$, which is a function on the tangent bundle $TQ$ of the configuration manifold $Q$, where the fiber component of $(q,\dot{q})$ is the time derivative of $q=q(t)$ at each time $t$. Therefore, in the discrete setting, the time derivative $\dot{q}(t)$ at $t=t_k$ can be approximated, for instance, by employing the finite difference between $q_k=q_d(t_k)$ and $q_{k+1}=q_d(t_{k+1})$ as 
\[
\dot{q}(t_k) \approx \frac{q_{k+1}-q_{k}}{h}.
\]
Here we introduce the \textit{forward finite difference map} defined on a neighborhood of the diagonal of $Q \times Q$ by
\begin{equation*}
\begin{split}
\Psi_{Q}: Q \times Q \to TQ; \quad 
(q_{k},q_{k+1}) \mapsto \Psi_{Q}(q_{k},q_{k+1})=\left(q_k, \frac{q_{k+1}-q_{k}}{h}\right) \in T_{q_{k}}Q,
\end{split}
\end{equation*}
where each pair $(q_{k} , q_{k+1}), k \in \{0, . . . , N-1\}$ denotes an evaluation of the curve $q(t)$ at $t=kh$ and $t=(k + 1)h$; hence $\Psi_{Q}(q_{k},q_{k+1})$ becomes an approximation of $\dot{q}(t_k)$.

Note that there is another type of  approximation of $\dot{q}(t_k)$ due to the {\it backward finite difference map}, namely, $\Psi_{Q}(q_{k-1},q_{k})=\left(q_{k}, \frac{q_{k}-q_{k-1}}{h}\right) \in T_{q_{k}}Q$. Strictly speaking, we need to further specify the base point to get the approximations. In the above cases, we choose the base point $q_k$ for the finite difference maps. In \S \ref{Sect:DiscDiracStr_LagSys}, we will return to the issues on the finite difference maps in the construction of discrete Dirac structures.

\paragraph{Discrete Lagrangians.}
We can introduce a discrete Lagrangian $L_d: Q \times Q \rightarrow \mathbb{R} $, which is an approximation of the time integral of the continuous Lagrangian $L: TQ \to \mathbb{R}$ between two consecutive configurations $q_k$ and $q_{k+1}$ as follows.

\begin{definition}\rm For a Lagrangian $L$ on $TQ$, the discrete Lagrangian $L_d$ on $Q \times Q$ is defined by
\begin{equation}\label{L_d}
L_d( q_k , q_{k+1}):=hL\circ \Psi_Q(q_{k},q_{k+1})\approx
\int_{t _k }^{t_{k+1}}L( q(t), \dot q(t)) dt.
\end{equation} 
Associated with the action functional $\mathfrak{S}:\mathcal{C} (Q)  \to \mathbb{R}$, define the {\it discrete action functional} $\mathfrak{S}_{d}: \mathcal{C}_d(Q) \to \mathbb{R}$ by
\[
\mathfrak{S}_{d}(q_{d})=\sum_{k=0}^{N-1}L_d(q _k , q_{k+1}).
\]
Since the discrete path $q_{d}\in \mathcal{C}_{d}(Q)$ is identified with its image $q_{d}=\{q_{k}\}_{k=0}^{N}$ and $\mathcal{C}_{d}$ is isomorphic to $Q\times \dots \times Q$ ($N+1$ copies), it is a smooth product manifold structure and hence $\mathfrak{S}_{d}$ inherits the smoothness of the discrete Lagrangian $L_{d}$ given in \eqref{L_d}.

\end{definition}

Corresponding to $TTQ$, the discrete object is the space $(Q \times Q) \times (Q \times Q)$. Let us introduce the projection operator $\pi$ and the translation operator $\sigma$ to be
\begin{equation*}
\begin{split}
&\pi: ((q_{0},q_{1}),(q_{0}^{\prime},q_{1}^{\prime})) \mapsto (q_{0},q_{1}),\\
&\sigma: ((q_{0},q_{1}),(q_{0}^{\prime},q_{1}^{\prime})) \mapsto (q_{0}^{\prime},q_{1}^{\prime}).
\end{split}
\end{equation*}
Let  $\pi_{i}: Q \times Q \to Q$ ($i=1,2$) be the usual projections of the first and second factors onto $Q$.
Then, the discrete second-order submanifold of $(Q \times Q) \times (Q \times Q)$ is defined by
 \[
 \ddot{Q}_{d}=\{w_{d} \in (Q \times Q) \times (Q \times Q) \mid \pi_{1} \circ \sigma(w_{d})=\pi_{2} \circ \pi(w_{d})\},
 \]
 which follows that for $w_{d}=((q_{0},q_{1}),(q_{0}^{\prime},q_{1}^{\prime})) \in \ddot{Q}_{d}$, we get the condition $q_{1}=q_{0}^{\prime}$.

\paragraph{Discrete Euler--Lagrange map.}
For a $C^k$ discrete Lagrangian $L_d$ $(k \ge 1)$, there exists a unique $C^{k-1}$ mapping $D_{DEL}L_{d}: \ddot{Q} \to T^{\ast}Q$ and unique $C^{k-1}$ one-forms $\Theta^{\pm}_{L_{d}}$ on $Q \times Q$ such that, for all variations $\delta q_{d} \in T_{q_{d}}\mathcal{C}(Q)$ of $q_{d}$, one has
\begin{equation}\label{dS_d}
\begin{split}
\delta \mathfrak{S}_{d}(q_{d})&=\mathbf{d}\mathfrak{S}_{d}(q_{d}) \cdot \delta{q}_{d}\\
&= \sum_{k=1}^{N-1} D_{DEL}L_{d}((q_{k-1},q_{k}),(q_{k},q_{k+1}))\cdot \delta{q}_{k}\\
&\qquad +\Theta^{+}_{L_{d}}(q_{N-1},q_{N})\cdot 
(\delta q_{N-1},\delta q_{N})-\Theta^{-}_{L_{d}}(q_{0},q_{1})\cdot (\delta q_{0},\delta q_{1}),
\end{split}
\end{equation}
where the map $D_{DEL}L_{d}: \ddot{Q} \to T^{\ast}Q$, called the {\it discrete Euler--Lagrange map}, is given by
\[
\begin{split}
D_{DEL}L_{d}((q_{k-1},q_{k}),(q_{k},q_{k+1}))=\left(D_2L_d(q_{k-1}, q_k)+D_1L_d(q_k, q_{k+1})\right)dq_{k}
\end{split}
\]
and the discrete Lagrangian one-forms are denoted by
\[
\begin{split}
&\Theta^{+}_{L_{d}}(q_{k},q_{k+1})=D_2L_d(q_{k}, q_{k+1})dq_{k+1},\\
&\Theta^{-}_{L_{d}}(q_{k},q_{k+1})=-D_1L_d(q_{k}, q_{k+1})dq_{k}.
\end{split}
\]
\paragraph{Discrete Hamilton's principle.}
Now, the discrete analogue of Hamilton's principle is given as follows.
\medskip

If a discrete path $q_{d}\in \mathcal{C}_{d}(Q)$ is critical for the discrete action functional, i.e.,
\[
\delta \mathfrak{S}_{d}(q_{d})=\delta \sum_{k=0}^{N-1}L_d(q _k , q_{k+1})=0,
\]
for all variations $\delta q_{d} \in T_{q_{d}}\mathcal{C}(Q)$ of $q_{d}$ with the fixed endpoints $\delta{q}_{0}=\delta{q}_{N}=0$, then the discrete path  satisfies the discrete Euler--Lagrange equations
\begin{equation}\label{DisEulLagEqn}
D_{DEL}L_{d}((q_{k-1},q_{k}),(q_{k},q_{k+1}))=0,\quad    k=1, \ldots,N-1,
\end{equation} 
namely,
\begin{equation}\label{DiscreteEulerLagEq} 
\begin{split}
D_2L_d(q_{k-1}, q_k)=-D_1L_d(q_k, q_{k+1}),
\;\;
\textrm{or},
\;\; \frac{\partial L_{d}}{\partial q_{k}}(q_{k-1}, q_k)+\frac{\partial L_{d}}{\partial q_{k}}(q_k, q_{k+1})=0,
\end{split}
\end{equation} 
for all $k=1,...,N-1$, where $D _i $ denotes the partial derivative with respect to the $i^{th}$ variable. 
Regarding the proof, refer to \cite{MaWe2001}. Note that the discrete Euler--Lagrange equations \eqref{DiscreteEulerLagEq} are the condition that the three consecutive configuration variables $q_{k-1}, q_k , q_{k+1}$ have to satisfy and hence it provides an integration scheme such that $q_{k+1}$ can be explicitly solved, under appropriate conditions, in terms of the two previous configuration variables $q_{k-1}$ and $q_k$.

\paragraph{Discrete Lagrangian evolution operator and mappings.}
Corresponding to the vector field $X: TQ \to TTQ$ in the continuous setting, we consider a {\it discrete evolution operator} $X_{d}: Q \times Q \to (Q \times Q) \times (Q \times Q): (q_{0},q_{1}) \mapsto ((q_{0},q_{1}),(q_{0}^{\prime},q_{1}^{\prime}))$ satisfying $\pi\circ X_{d}=\mathrm{Id}_{Q\times Q}$. Furthermore, the corresponding discrete object to the flow $\varphi_{t}: TQ \to TQ$ is a discrete map $\varphi_{d}: Q \times Q \to Q \times Q$, which is given by $\varphi_d:=\sigma \circ X_d;\;(q_{0},q_{1}) \mapsto (q_{0}^{\prime},q_{1}^{\prime})$. Now for the special case of a regular discrete Lagrangian system, we can define the {\it discrete Lagrangian evolution operator} 
\begin{equation*}
\begin{split} 
X_{L_{d}}: Q \times Q \to  \ddot{Q}_{d};\quad  
(q_{0},q_{1}) \mapsto (q_{0},q_{1},q_{1},q_{2}),
\end{split}
\end{equation*} 
which is a second-order discrete evolution operator that satisfies 
\[
D_{DEL}L_{d} \circ X_{L_{d}}=0,
\]
and the {\it discrete Lagrangian map} $\varphi_{L_{d}}: Q \times Q \to Q \times Q$ is defined as
\[
\varphi_{L_{d}}:= \sigma \circ X_{L_{d}};\; (q_{0},q_{1}) \mapsto (q_{1},q_{2}).
\]

\paragraph{Discrete Legendre transforms.}
One can define the \textit{discrete Legendre transforms} or the discrete fiber derivatives $\mathbb{F}^{\pm}L_d: Q \times Q \rightarrow T^*Q$ by
\begin{equation*}
\begin{split} 
&\left<\mathbb{F} ^+L_d( q _0 , q _{1} ), \delta q_{1}\right> = \left<D_2L_d( q _0 , q _{1} ), \delta q_{1}\right>, \\
&\left<\mathbb{F} ^-L_d( q _0 , q _{1} ), \delta q_{0} \right> = \left<-D_1L_d( q _0 , q _{1} ), \delta q_0\right>,
\end{split}
\end{equation*} 
each of which can be respectively denoted by
\begin{equation}\label{eq:disLeg}
\begin{split} 
&\mathbb{F} ^+L_d:  (q _0 , q _{1}) \mapsto (q_{{1}}, p_{{1}})= (q_{1}, D_2L_d( q _0 , q _{1} )) \in T^*_{q _{1} }Q,\\
&\mathbb{F} ^-L_d: (q _0 , q _{1}) \mapsto (q_{0}, p_{0})= (q_0, -D_1L_d( q _0 , q _{1} )) \in T^*_{q _0}Q.
\end{split}
\end{equation} 
%
%
Note that
\[
\Theta^{\pm}_{L_{d}}=(\mathbb{F} ^{\pm}L_d)^{\ast}\Theta_L.
\]
%
%
In particular, the discrete Lagrangian $L_d$ is called \textit{regular} when the discrete Legendre transforms are local diffeomorphisms. This turns out to be equivalent to the invertibility of the matrix $D_1D_2L_d( q _0 , q _{1} )$ for all $ q _0, q _{1} $, which exactly corresponds to the regularity condition of the Lagrangian given in \eqref{JacobLag} in the continuous setting. If the discrete fiber derivatives $\mathbb{F}^{\pm}L_d$ are global isomorphisms, then $L_{d}$ is called {\it hyperregular}.

Under the regularity hypothesis, the integration scheme \eqref{DiscreteEulerLagEq} becomes
\[
\mathbb{F} ^+L_d(q_{k-1}, q_k)=\mathbb{F} ^-L_d(q_k, q_{k+1}),
\]
which yields the well-defined discrete Lagrangian map $\varphi_{L_d} : Q\times Q \to Q \times Q$ by
\begin{equation*} 
\begin{split}
\varphi_{L_{d}}=(\mathbb{F}^{-}L_{d})^{-1} \circ \mathbb{F}^{+}L_{d}:
(q_{k-1},q_k) \mapsto(q_k, q_{k+1}).
\end{split}
\end{equation*} 
Further, we develop 
\begin{equation*}
\varphi_{L_{d}}^{k}:=\underbrace{\varphi_{L_{d}} \circ \cdots \circ \varphi_{L_{d}}}_{\textrm{$k$-copies}},\quad k=1,...,N-1,
\end{equation*}
where $\varphi_{L_{d}}^{1}=\varphi_{L_{d}}$. Corresponding to the continuous path $t \mapsto \varphi_t(v_q), \; t \in [0,T]$, we can define the discrete path by $k \mapsto \varphi_{L_{d}}^{k}(q_{0},q_{1}),\;k\in \{1,...,N-1\}$.

\paragraph{Structure preserving property.} 
Denote the set of solution paths of the discrete Euler--Lagrange equations in \eqref{DiscreteEulerLagEq} by 
\[
\mathcal{C}_{L_{d}}(Q) \subset \mathcal{C}_{d}(Q).
\]
Letting $(q_{0},q_{1}):=(q_{d}(t_{0}),q_{d}(t_{1})) \in Q \times Q$ be an initial condition, a solution path $q_{d} \in \mathcal{C}_{L_{d}}(Q) $ can be uniquely determined by the discrete Lagrange map $\varphi_{L_{d}}: Q \times Q \to Q \times Q$. Then, we identify the space of solution paths $\mathcal{C}_{L_{d}}(Q)$ with the space $Q \times Q$, and restrict the discrete action functional $\mathfrak{S}_{d}(q_{d})$ to $v_{d}=(q_{0},q_{1}) \in Q \times Q$ to get
\begin{equation*}
\hat{\mathfrak{S}}_{d}(v_{d})=\mathfrak{S}_{d}(q_{d}).
\end{equation*}
From equation \eqref{dS_d}, it follows that for any $w_{v_{d}} \in T_{v_{d}}(Q \times Q)$,
\begin{equation*}
\begin{split}
\mathbf{d}\hat{\mathfrak{S}}_{d}(v_{d}) \cdot w_{v_{d}}&=\Theta^{+}_{L_{d}}(\varphi_{L_{d}}^{N-1}(v_{d}))(\varphi_{L_{d}}^{N-1})_{\ast}(w_{v_{d}}))-\Theta^{-}_{L_{d}}(v_{d})(w_{v_{d}})\\[3mm]
&=(\varphi_{L_{d}}^{N-1})^{\ast}\Theta^{+}_{L_{d}}(v_{d})(w_{v_{d}})-\Theta^{-}_{L_{d}}(v_{d})(w_{v_{d}}),
\end{split}
\end{equation*}
where the first term on the right-hand side of \eqref{dS_d} vanishes because of \eqref{DisEulLagEqn}. Thus, we get
\begin{equation}\label{dHatSd}
\mathbf{d}\hat{\mathfrak{S}}_{d}(v_{d})=(\varphi_{L_{d}}^{N-1})^{\ast}\Theta^{+}_{L_{d}}(v_{d})-\Theta^{-}_{L_{d}}(v_{d}).
\end{equation}
Taking the exterior derivative of \eqref{dHatSd} leads to, in view of $\mathbf{d}^{2}\hat{\mathfrak{S}}_{d}=0$, the structure preserving property of the discrete Lagrangian two-form:
\begin{equation}\label{SymplecticProperty_DisLagSys}
(\varphi_{L_{d}}^{N-1})^{\ast}\Omega_{L_{d}}=\Omega_{L_{d}}
\end{equation}
or, simply,
\begin{equation*}
\varphi_{L_d} ^\ast \Omega _{L_d}= \Omega _{L_d}.
\end{equation*} 
In the above, the Lagrangian two-form $\Omega_{L_{d}}$ is defined with respect to either $\mathbb{F}  ^+L_d$ or $ \mathbb{F}  L_d ^-$ as
\[
\Omega _{L_d}=\mathbf{d}\Theta^{+}_{L_{d}}=\mathbf{d}\Theta^{-}_{L_{d}}
\]
with the coordinate expression
\[
\Omega _{L_d}(q_{0},q_{1})=\frac{\partial^{2}L_{d} }{\partial q^{i}_{0} \partial q^{j}_{1}} dq^{i}_{0}\wedge dq^{j}_{1}.
\]
This is  the discrete analogue of the Lagrangian two-form in the continuous setting given in \eqref{LagrangianForms} and hence $\Omega_{L_{d}}$ is the discrete symplectic structure on $Q \times Q$ that is induced by the discrete Legendre transforms $\Theta^{\pm}_{L_{d}}$ as $ \Omega _{L_d}:= (\mathbb{F} ^\pm L_d) ^\ast \Omega_{T^{\ast}Q}=(\mathbb{F} ^\pm L_d) ^\ast (-\mathbf{d}\Theta_{T^{\ast}Q})$.

\paragraph{Momentum matching.}
Associated with the discrete Legendre transform, let us introduce the notations for the momentum at the two endpoints of each interval $[t_k, t_{k+1}]$ as
\begin{equation*}
\begin{split}
&p^{+}_{k,k+1}=p^{+}(q_{k},q_{k+1})=\mathbb{F}^{+}L_{d}(q_{k},q_{k+1}),\\
&p^{-}_{k,k+1}=p^{-}(q_{k},q_{k+1})=\mathbb{F}^{-}L_{d}(q_{k},q_{k+1}).
\end{split}
\end{equation*}
From the discrete Euler--Lagrange equations \eqref{DiscreteEulerLagEq}, one observes
\begin{equation*}
\mathbb{F}^{+}L_{d}(q_{k-1},q_{k})=\mathbb{F}^{-}L_{d}(q_{k},q_{k+1}), \;\; \textrm{or},\;\; p^{+}_{k-1,k}=p^{-}_{k,k+1},
\end{equation*} 
which implies that the discrete Euler--Lagrange equations simply enforce the condition that the momentum at time $k$ must be the same when evaluated over the lower interval $[k-1, k]$ and the upper interval $[k, k+1]$. Then, along a solution curve, there is a unique momentum at each time $k$, which is denoted by
\[
p_{k}\equiv p^{+}_{k-1,k}=p^{-}_{k,k+1}.
\]
Thus, the discrete trajectory $\{q_{k}\}_{k=0}^{N}$ in $Q$ can be lifted as a trajectory $\left\{(q_{k}, q_{k+1})\right\}_{k=0}^{N-1}$ in $Q \times Q$ or, equivalently, as a trajectory $\left\{(q_{k}, p_{k})\right\}_{k=0}^{N}$ in $T^{\ast}Q$.

\paragraph{Discrete Hamiltonian maps.}
Associated with the discrete Lagrange map $\varphi_{L_{d}}: Q \times Q \to Q \times Q$, 
we can develop the discrete Hamiltonian map $\tilde{\varphi}_{L_{d}}: T^{\ast}Q \to T^{\ast}Q$ by
\begin{equation*}
\begin{split}
\tilde{\varphi}_{L_{d}}=\mathbb{F}^{+}L_{d}\circ (\mathbb{F}^{-}L_{d})^{-1}:
(q_{0},p_{0}) \mapsto (q_{1},p_{1}), 
\end{split}
\end{equation*}
where
\begin{equation*}
\begin{split}
p_{0}=-D_1L_d( q _{0} , q _{1} ),\quad p_{1}=D_2L_d( q _{0} , q _{1} ).
\end{split}
\end{equation*}
Note that the discrete Hamiltonian maps are equivalent to the following definitions: 
\begin{equation*}
\begin{split}
\tilde{\varphi}_{L_{d}}&=\mathbb{F}^{+}L_{d} \circ \varphi_{L_{d}} \circ (\mathbb{F}^{+}L_{d})^{-1},\\[2mm]
\tilde{\varphi}_{L_{d}}&=\mathbb{F}^{-}L_{d} \circ \varphi_{L_{d}} \circ (\mathbb{F}^{-}L_{d})^{-1},\\[2mm]
\tilde{\varphi}_{L_{d}}&=\mathbb{F}^{+}L_{d}  \circ (\mathbb{F}^{-}L_{d})^{-1}.
\end{split}
\end{equation*}
Then, we get the following commutative diagram.
\[
\xymatrix@R=60pt{
&(q_{0},q_{1})\ar[ld]_{\mathbb{F}^{-}L_{d}}\ar[rr]^-{\varphi_{L_{d}}}\ar[dr]^{\mathbb{F}^{+}L_{d}} & &(q_{1},q_{2})\ar[dl]_{\mathbb{F}^{-}L_{d}}\ar[dr]^{\mathbb{F}^{+}L_{d}}\\
(q_{0},p_{0})\ar[rr]^-{\tilde{\varphi}_{L_{d}}} & & (q_{1},p_{1})\ar[rr]^-{\tilde{\varphi}_{L_{d}}}&&(q_{2},p_{2})
}
\]

\subsection{Nonholonomic integrators} 
Here let us make a short review on the variational discretization of Lagrangian systems with nonholonomic constraints, in which we introduce the discrete Lagrange--d'Alembert principle that is a discrete analogue of the Lagrange--d'Alembert principle shown in \S\ref{sect:LDAP}. 
\medskip

\paragraph{The discrete Lagrange--d'Alembert principle.}
Following \cite{CoMa2001,McPe2006}, the constraint distribution $\Delta_Q \subset TQ$ that is given in \eqref{ConstraintDistribution} can be discretized as a discrete constraint space $\Delta_Q^{d} \subset Q \times Q$ such that
\begin{equation}\label{DiscConstDist}
\begin{split}
\Delta_{Q}^{d}:=\left\{ (q_{0},q_{1}) \in Q \times Q \mid  \omega^{r}_d(q_0,q_1)=0,\,r=1,...,m <n \right\},
\end{split}
\end{equation}
where $ \omega^{r}_d: Q \times Q \to \mathbb{R},\; r=1,...,m$ are functions which span the annihilator of $\Delta_Q^{d}$ and are defined by $\omega^{r}_d(q_0,q_1)=\left<\omega^{r}(q_0), \Psi_Q(q_0,q_1)\right>$. Notice that the base point $q_0$ is chosen from $(q_0,q_1) \in Q \times Q$ in \cite{CoMa2001,McPe2006}, so that $\omega^{r}(q)$ is evaluated at $q_0$. Later, we will also discuss the case in which the base point $q_1$ is chosen from $(q_0,q_1) \in Q \times Q$ so that $\omega^{r}(q)$ is evaluated at $q_1$ in \S \ref{Sect:DiscDiracStr_LagSys}.
\medskip

Recall that $q_{d}=\{ q_0, q_1, ....,q_{N} \}$ is the discrete sequences on $Q$ and also that for the discrete Lagrangian $L_d: Q\times Q \to \mathbb{R}$, the stationarity condition of the discrete action sum $\mathfrak{S}_d$ is given by equation \eqref{dS_d}. For the case in which there exist the discrete constraints \eqref{DiscConstDist} that are derived from the nonholonomic constraint \eqref{ConstraintDistribution}, the discrete Lagrange--d'Alembert principle is given by
\begin{equation}\label{dS_d_LDAP}
\begin{split}
\delta \mathfrak{S}_{d}(q_{d})=\sum_{k=1}^{N-1} \left(D_2L_d(q_{k-1}, q_k)+D_1L_d(q_k, q_{k+1})\right)\delta q_{k}=0,
\end{split}
\end{equation}
for the chosen variation $\delta{q}_k \in \Delta_Q({q}_k)$ with the fixed boundary conditions $\delta{q}_0=\delta{q}_N=0$, together with the constraints $\omega^{r}_d(q_k,q_{k+1})=0$.

Thus we get the {\it discrete Lagrange--d'Alembert equations} as,
\begin{equation}  \label{eq: local_discrete_LDA}
\left\{
\begin{aligned}
&D_2L_d(q_{k-1}, q_k) +D_1L_d(q_k, q_{k+1})\in \Delta^\circ_Q({q}_k),\quad k=1,\ldots, N-1, \\[2mm]
&(q_k,q_{k+1}) \in \Delta_Q^d,\quad k=0,\ldots, N-1,
\end{aligned}
\right.
\end{equation}
which are described in coordinates by
\begin{equation}  \label{eq: local_discrete_LDA_mu}
\left\{
\begin{aligned}
&\frac{\partial L_{d}}{\partial q_{k}}(q_{k-1}, q_k)+\frac{\partial L_{d}}{\partial q_{k}}(q_k, q_{k+1})=\mu_{r}\omega^{r}_{i}(q_k), \quad i=1,...,n, \\[1mm]
& \omega^{r}_d(q_k,q_{k+1})=0,\quad r=1,...,m <n.
\end{aligned}
\right.
\end{equation}
\paragraph{Discrete flow map.} Under appropriate regularity assumption, by the implicit function theorem, if the following Jacobian 
matrix
\begin{equation}\label{regularity_criteria} 
\left[
\begin{array}{ll}
D_2 D_1 L_d(q_k,q_{k+1}) & \omega^{r}(q^{k}) \\
D_2\omega^{r}(q_k,q_{k+1}) & 0 
\end{array}
\right]
\end{equation} 
is invertible for each $(q_k,q_{k+1})$ in a neighborhood of the diagonal of $Q\times Q$, then it is guaranteed that there exists the second-order discrete flow map
$$
\varphi_{L_{d}}:(q_{k-1},q_{k})\mapsto(q_k,q_{k+1}),
$$
so that $q_{k+1}$ satisfies the discrete Lagrange--d'Alembert equations in \eqref{eq: local_discrete_LDA_mu} provided that $(q_{k-1},q_{k}) \in \Delta_Q^d$.

\paragraph{Structural property of the discrete nonholonomic systems.} 
In the case of nonholonomic mechanical systems, we show that the preservation of the discrete Lagrangian two-form is broken. 
Let $\tau_{Q_1}: Q \times Q \to Q; (q_{0},q_{1}) \mapsto q_{0}$ be the natural projection onto the first slot. 

It follows from \eqref{dS_d_LDAP} that the variation \eqref{dS_d} with nonholonomic constraints reads
 that, for an initial condition $v_d=(q_0,q_1)$ and for any $w_{v_{d}}=(q_0, q_1, \delta{q}_0,\delta{q}_1)$,
\begin{equation*}
\begin{split}
\mathbf{d}\hat{\mathfrak{S}}_{d}(v_{d}) \cdot w_{v_{d}}&=
\beta(\tau_{Q_1}(v_{d})) \cdot T\tau_{Q_1}(w_{v_{d}})+
\Theta^{+}_{L_{d}}(\varphi_{L_{d}}^{N-1}(v_{d}))(\varphi_{L_{d}}^{N-1})_{\ast}(w_{v_{d}}))-\Theta^{-}_{L_{d}}(v_{d})(w_{v_{d}})\\[3mm]
&=\tau_{Q_1}^{\ast}\beta(v_{d})\cdot w_{v_{d}}+
(\varphi_{L_{d}}^{N-1})^{\ast}\Theta^{+}_{L_{d}}(v_{d})(w_{v_{d}})-\Theta^{-}_{L_{d}}(v_{d})(w_{v_{d}}),
\end{split}
\end{equation*}
where $\beta=\mu_{r}\omega^{r}_{i}(q)$ is the one-form on $Q$ that implies constraint forces. For the nonholonomic mechanical systems, note that the first term on the right-hand side of \eqref{dS_d} is replaced by the constraint force because of \eqref{eq: local_discrete_LDA} or \eqref{eq: local_discrete_LDA_mu}. Therefore, we get
\begin{equation*}
\mathbf{d}\hat{\mathfrak{S}}_{d}(v_{d})=\tau_{Q_1}^{\ast}\beta(v_{d})+(\varphi_{L_{d}}^{N-1})^{\ast}\Theta^{+}_{L_{d}}(v_{d})-\Theta^{-}_{L_{d}}(v_{d}).
\end{equation*}
Setting 
$\boldsymbol{\beta}_{d}=\tau_{Q_1}^{\ast}\beta$, it follows
\[
\varphi^{\ast}_{L_{d}}\Omega_{L_{d}}(v_{d})=\Omega_{L_{d}}(v_{d})+\mathbf{d}\boldsymbol{\beta}_{d}(v_{d}).
\]
Since $v_d$ is arbitrary, we get the structural property of the discrete Lagrange--d'Alembert equations as
\begin{equation*}
\varphi_{L_{d}}^{\ast}\Omega_{L_d}=\Omega_{L_d}+\mathbf{d}\boldsymbol{\beta}_{d}.
\end{equation*}
The above structural relation is the discrete analogue of the structure-preserving property in the continuous setting that is given in \eqref{LDA_Str} and it follows that 
the structure preserving property of the discrete Lagrangian two-form is broken for the case of nonholonomic mechanical systems.

\section{Dirac structures in Lagrangian mechanics}
In the previous section, we examined the variational discretization of a Lagrangian system with nonholonomic constraints, where the structure-preserving property is broken in the discrete setting, similar to the continuous case. To clarify the underlying structure of such nonholonomic Lagrangian systems, we will provide a brief review of Dirac structures and their associated Lagrange--Dirac systems (or implicit Lagrangian systems). 

\subsection{Review on continuous Dirac structures in mechanics}
We begin by describing the concept of Dirac structures on manifolds, focusing on their application to induced Dirac structures on cotangent bundles. These structures play a crucial role in the construction of Dirac dynamical systems in Lagrangian mechanics. For further details, see \cite{CoWe1988, Co1990a, YoMa2006a}.

\paragraph{Dirac structures on vector spaces.} Before going into details on Dirac structures on a manifold, we start with the definition of a linear Dirac structure on a vector space $V$. Let $V$ be an $n$-dimensional vector space, and $V^{\ast}$ its dual vector space. We consider a symmetric paring
$\langle \! \langle\cdot,\cdot \rangle \!  \rangle$ on $V \oplus V^{\ast}$ that is defined by, for $(v,\alpha), (\bar{v},\bar{\alpha}) \in V \oplus V^{\ast}$,
\begin{equation*}
\langle \! \langle\, (v,\alpha),
(\bar{v},\bar{\alpha}) \,\rangle \!  \rangle
=\langle \alpha, \bar{v} \rangle
+\langle \bar{\alpha}, v \rangle,
\end{equation*}
where $\langle\cdot \, , \cdot\rangle$ denotes the dual paring between $V^{\ast}$ and $V$.
\medskip

A {\it linear Dirac structure} on $V$ is defined as a subspace $D \subset V \oplus
V^{\ast}$ such that
$D=D^{\perp}$, where $D^{\perp}$ is the orthogonal subspace relative to the pairing
$\langle \! \langle \cdot,\cdot \rangle \!  \rangle$, given by
\begin{equation*}
D^{\perp}:=\left\{(v,\alpha)\in V\oplus V^*\mid \langle \bar{\alpha}, v\rangle+\langle \alpha, \bar{v}\rangle=0\text{ holds for any }(\bar{v},\bar{\alpha})\in D\right\}.
\end{equation*}

\paragraph{Dirac structures on manifolds.} 
Let $M$ be an $n$-dimensional smooth manifold and we denote by $TM$ and $T^*M$ the tangent bundle and cotangent bundle over $M$ respectively. 
Then an (almost) Dirac structure on $M$ is defined as a subbundle $D_{M} \subset TM \oplus T^{\ast}M$ such that the subspace $D_{M}(x)\subset T_{x}M\times T^*_{x}M$ is a Dirac structure in the sense of vector spaces at each point $x\in M$. An example of such a Dirac structure on a manifold $M$ is constructed using a two-form and a distribution on $M$, as shown below.
\medskip

Denote by $\Omega_{M}$ a two-form on $M$ and by $\Delta_{M}$ a distribution on $M$. A Dirac structure $D_{M}$ on $M$ is given by, for each  $ x \in M$, 
\begin{equation}\label{DiracManifold}
\begin{split}
D_{M}(x)=\{ (v_{x}, \alpha_{x}) \in T_{x}M \times T^{\ast}_{x}M
  \; \mid \; & v_{x} \in \Delta_{M}(x), \; \mbox{and} \\ 
  & \alpha_{x}(w_{x})=\Omega_{\Delta_{M}}(x)(v_{x},w_{x}) \; \;
\mbox{for all} \; \; w_{x} \in \Delta_{M}(x) \},
\end{split}
\end{equation}
where $\Omega_{\Delta_{M}}$ is the restriction of $\Omega_{M}$ to $\Delta_{M}$. 
\medskip

In this paper, we sometimes use the following representation which is equivalent with the Dirac structure in \eqref{DiracManifold}:
\begin{equation*}
\begin{aligned}
D_M(x):=&\left\{(X_x,\alpha_x)\in T_xM\oplus T^*_xM\mid X_x\in\Delta_M(x)\text{ and } \alpha_x-\mathbf{i}_{X_x}\Omega_{\Delta_M}(x)\in \Delta_M^{\circ}(x)\right\}.
\end{aligned}
\end{equation*}
Furthermore, $D_{M}$ may be equivalently written as 
\begin{equation}\label{ContinuousDiracStructure}
\begin{split}
D_M(x):=&\left\{(X_x,\alpha_x)\in T_xM\oplus T^*_xM\mid X_x\in\Delta_M(x)\text{ and } \alpha_x-\Omega^{\flat}_{M}(x)\cdot X_{x}\in \Delta_M^{\circ}(x)\right\},
\end{split}
\end{equation}
where $\Omega_{M}^{\flat}: TM \to T^{\ast}M$ is the bundle map defined by, for each $x \in M$, 
\[
\Omega^{\flat}_{M}(x)(X_{x})(Y_{x})=\mathbf{i}_{X_x}\Omega_{\Delta_M}(x)(Y_{x})=\Omega_{M}(x)(X_{x},Y_{x}), \;\;\textrm{for all}\;\; X, Y \in \mathfrak{X}(M).
\]
\paragraph{Integrability.}
The Dirac structure simultaneously generalizes both two-forms and Poisson structures. Its integrability condition requires either the closedness of the two-form or the satisfaction of Jacobi's identity by the Poisson tensor. In particular, if  the condition
\begin{equation*}\label{ClosedCond}
\langle \pounds_{X_1} \alpha_2, X_3 \rangle
+\langle \pounds_{X_2} \alpha_3, X_1 \rangle+\langle \pounds_{X_3}
\alpha_1, X_2 \rangle=0
\end{equation*}
is satisfied for all pairs of vector fields and one-forms $(X_1, \alpha_1)$,
$(X_2,\alpha_2)$, $(X_3,\alpha_3)$ that take values in $D_{M}$, then $D_{M}$ is said to be an {\it integrable} Dirac structure.

In mechanics, we often encounter mechanical systems with nonholonomic constraints, where the Dirac structure fails to satisfy the integrability condition, sometimes termed as an `almost' Dirac structure. In this paper, we will refer to a Dirac structure even in such cases, unless otherwise stated.

\paragraph{The induced Dirac structure on the cotangent bundle.} 
We have already shown the construction of the Dirac structure as in \eqref{DiracManifold} or \eqref{ContinuousDiracStructure}.
Here we consider the standard example of Dirac structures in mechanics, i.e., an induced Dirac structure on the cotangent bundle of a configuration manifold by applying the construction. 

Let $Q$ be an $n$-dimensional configuration manifold.  We consider a Dirac structure on the cotangent bundle $T^{\ast}Q$ that is induced from a given distribution $\Delta_{Q}$ on $Q$.
The distribution $\Delta_{T^{\ast}Q}$ on $T^{\ast}Q$ is defined by lifting $\Delta_{Q}$ as
\begin{equation*}
\Delta_{T^{\ast}Q}
:=( T\pi_{Q})^{-1} \, (\Delta_{Q}) \subset T(T^{\ast}Q),
\end{equation*}
where $T\pi_{Q}:TT^{\ast}Q \to TQ$  is the tangent map of $\pi_{Q}:T^{\ast}Q \to Q$.
\medskip

We can define an {\it induced Dirac structure} $D_{\Delta_{Q}}$ on $T^{\ast}Q$ by the subbundle of $T  T^{\ast}Q \oplus T ^{\ast} T^{\ast}Q$, whose fiber is given by, for each $z \in
T^{\ast}Q$, 
\begin{align}\label{IndDiracStru}
D_{\Delta_{Q}}(z)
& =\left\{ (X_{z}, \alpha_{z}) \in T_{z}T^{\ast}Q \times
T^{\ast}_{z}T^{\ast}Q \mid \; \right. \nonumber\\
&\left.\hspace{2cm}
 X_{z} \in
\Delta_{T^{\ast}Q}(z),  \; \mbox{and}\;
\alpha_{z}- \Omega_{T^{\ast}Q}^{\flat}(z) \cdot X_{z} \in \Delta_{T^{\ast}Q}^{\circ}\right\},
\end{align}
where $\Omega^{\flat}_{T^{\ast}Q}: TT^{\ast}Q \to T^{\ast}T^{\ast}Q$ is the bundle map associated with the canonical symplectic form $\Omega_{T^{\ast}Q}$.
Note that the induced Dirac structure is integrable if and only if the constraint distribution $ \Delta _Q$ is holonomic.

\paragraph{Local expressions for the induced Dirac structure.}
We consider a local expression of the induced Dirac structure, and we choose local coordinates $q
^1,...,q^n$ for $q \in Q$ so that $Q$ is locally denoted by an open set $U \subset
\mathbb{R}^n$. The distribution $\Delta_Q$ defines a constraint subspace $\Delta_Q(q)$ of $T_qQ$ at each point $q \in Q$ and 
if the dimension of the constraint space is $n-m$, then we can choose a basis $e _{m+1}(q), e _{m+2}(q),\ldots, e _n (q)$ of
$\Delta(q)$.

Since the cotangent bundle projection $\pi_Q: T^{\ast}Q\rightarrow Q $ is denoted  
as $z=(q,p) \mapsto q$, its tangent map is locally given by
$T\pi_Q : TT^{\ast}Q \rightarrow TQ; (q, p, \dot{q}, \dot{p}) \mapsto (q, \dot{q})$.
Then the induced distribution is locally given by
\begin{equation*}
\Delta_{T^{\ast}Q}(z)=\{(q,p,\dot{q},\dot{p})\mid  \dot{q}\in\Delta_Q(q)\}
\end{equation*}
and its annihilator is locally denoted by
\begin{equation*}
\Delta_{T^{\ast}Q}^{\circ}(z)=\{(q,p,\eta,\xi)\mid \eta\in \Delta^{\circ}_Q(q) \text{ and } \xi=0\}.
\end{equation*}
The flat bundle map $\Omega^{\flat}_{T^{\ast}Q}(z): T_{z}T^{\ast}Q \to T_{z}^{\ast}T^{\ast}Q$ is locally written as $(q,p,\dot{q}, \dot{p}) \mapsto (q,p,-\dot{p},\dot{q})$, and hence the condition $\alpha_{z}- \Omega_{T^{\ast}Q}^{\flat}(z) \cdot X_{z} \in \Delta_{T^{\ast}Q}^{\circ}$ in equation \eqref{IndDiracStru} reads
\begin{equation*}
\eta+\dot{p}\in\Delta^{\circ}_Q(q) \text{ and } \xi-\dot{q}=0.
\end{equation*}
Thus we get the local expression of the induced Dirac structure as
\begin{equation}\label{LocIndDiracStr}
D_{\Delta_Q}(z)=\{((q,p,\dot{q},\dot{p}),(q,p,\eta,\xi))\mid\dot{q}\in\Delta_Q(q),\; \xi=\dot{q} \;\text{ and } \eta+\dot{p}\in\Delta^{\circ}_Q(q)\}.
\end{equation}

\subsection{Lagrange--Dirac dynamical systems}\label{Sect:LagDiracDySys}
Here, we consider a Lagrangian mechanical system with nonholonomic constraints within the framework of the induced Dirac structure on $T^{\ast}Q$.  This system is referred to as a Lagrange--Dirac dynamical system or an implicit Lagrangian system.

\paragraph{Diffeomorphisms between iterated tangent and cotangent bundles.} To consider the Dirac structure in mechanics, we employ the canonical diffeomorphisms between $T^{\ast}TQ$, $TT^{\ast}Q$ and $T^{\ast}T^{\ast}Q$. 
\medskip

In a local trivialization, $Q$ is represented by an open set $U$ in a linear space $W$, so that
$TT^{\ast}Q$ is represented by $(U \times W^{\ast}) \times (W \times W^{\ast})$, $T^{\ast}TQ$ is locally given by $(U \times W) \times (W^{\ast} \times W^{\ast})$, and $T^{\ast}T^{\ast}Q$ is locally given by $(U \times W^{\ast}) \times (W^{\ast} \times W)$. In this local trivialization, let us denote by $(q,p)$ the local coordinates of $T^{\ast}Q$ and also by $(q,p,\delta{q},\delta{p})$ the corresponding
coordinates of $TT^{\ast}Q$, and $(q, \delta q, \delta p, p)$ are the local coordinates of $T^{\ast}TQ$, while $(q, p, -\delta p, \delta q)$ are the local coordinates of $T^{\ast}T^{\ast}Q$.

\begin{figure}[htbp]
  \begin{center}
    \includegraphics[scale=0.6]{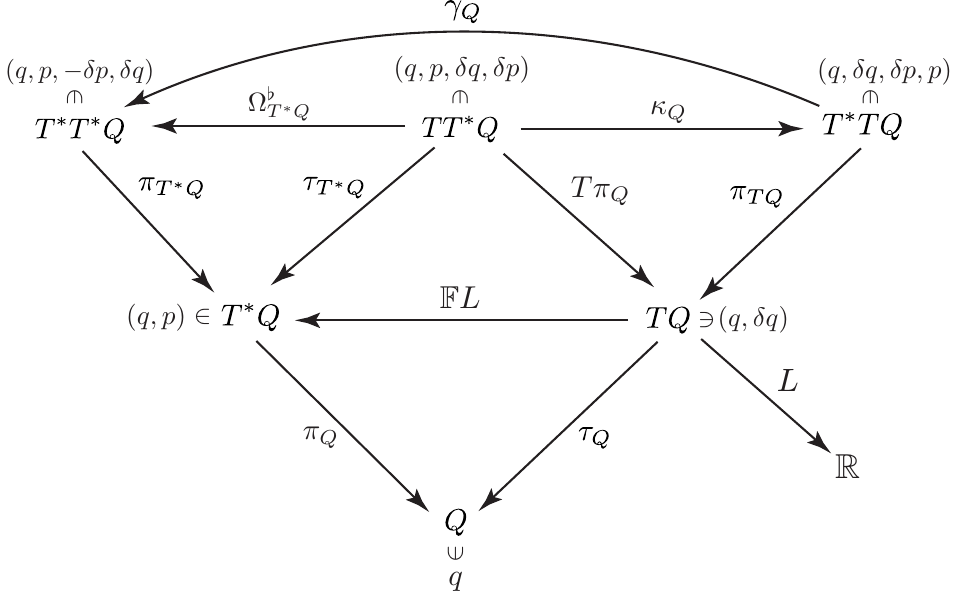}
  \end{center}
  \caption{The bundle picture of the iterated tangent and cotangent bundles.}
  \label{fig:BundPic}
\end{figure}

As in Fig. \ref{fig:BundPic}, there exist three canonical diffeomorphisms: 
\begin{equation}\label{eq:candiff}
\begin{split}
&\Omega_{T^{\ast}Q}^{\flat}: TT^{\ast}Q \to T^{\ast}T^{\ast}Q; \hspace{1.6cm}\, (q,p,\delta{q},\delta{p}) \mapsto (q,p,-\delta{p}, \delta{q}), \\
&\kappa_{Q}: TT^{\ast}Q \to T^{\ast}TQ;\hspace{2.2cm} (q, p, \delta q, \delta p) \mapsto (q, \delta q, \delta p, p),\\
&\gamma_{Q}:=\Omega^{\flat} \circ \kappa_{Q}^{-1}: T^{\ast}TQ \to T^{\ast}T^{\ast}Q;\; \quad (q,\delta{q},\delta{p},p) \mapsto (q,p,-\delta{p}, \delta{q}).
\end{split}
\end{equation}
In the above, the first diffeomorphism $\Omega^{\flat}_{T^{\ast}Q}:TT^{\ast}Q \to T^{\ast}T^{\ast}Q$ is obviously the bundle map associated with the canonical symplectic structure $\Omega_{T^{\ast}Q}$. The second diffeomorphism  $\kappa_{Q}: TT^{\ast}Q \to T^{\ast}TQ$ was originally introduced by \cite{Tu1977} in the context of the generalized Legendre transform, and $\kappa_{Q}$ is the unique map that intertwines two sets of maps given
as follows.
\medskip

The first commutation condition that is used to define $\kappa_Q$ is given by the set of maps $T{\pi_Q}: TT^{\ast}Q \to TQ$ and $\pi_{TQ}: T^{\ast}TQ \to TQ$ as
\begin{equation*}
\pi_{TQ} \circ \kappa_Q = T \pi_Q.
\end{equation*}
The second commutation condition is given by the set of maps $\tau_{T^{\ast}Q}: TT^{\ast}Q \to T^{\ast}Q$, and $\varrho: T^{\ast}TQ \to T^{\ast}Q$ as
\begin{equation*}
\varrho \circ \kappa_Q=\tau_{T^{\ast}Q}.
\end{equation*}
In the above, the map $\varrho: T^{\ast}TQ \to T^{\ast}Q$ is defined by, for $\alpha_{v_q} \in T^{\ast}_{v_q}TQ$ and $u_q \in T_qQ$,
\[
\langle \varrho(\alpha_{v_q}), u_q \rangle=\langle \alpha_{v_q}, 
\textrm{ver}(u_q,v_q) \rangle,
\]
where
\[
\mathrm{ver}(u_q,v_q)=\frac{d}{dt}\bigg|_{t=0}(v_q+tu_q) \in T_{v_q}TQ
\]
is the vertical lift of $u_q$ along $v_q$.
\medskip

In the local trivialization, two sets of maps are readily checked to be given by
\[
\begin{split}
T\pi_Q(q,p,\delta q,\delta p)&=(q,\delta q),\\
\pi_{TQ}(q,\delta q, \delta p, p)&=(q, \delta q),\\
\tau_{T^{\ast}Q}(q,p,\delta q,\delta p)&=(q,p),\\
\varrho(q,\delta q,\delta p,p)&=(q,p),
\end{split}
\]
from which it is straightforward to check that the commutation conditions are 
satisfied and it is clear that this uniquely characterizes the map 
$\kappa_Q$.

\paragraph{The Lagrange--Dirac dynamical systems.} Let $L:TQ \to \mathbb{R}$ be a Lagrangian, possibly degenerate. The differential $ \mathbf{d} L:TQ \rightarrow T^{\ast}TQ$ of $L$ is a one-form on $TQ$, which is locally given by
\[
\mathbf{d} L(q,v)= \left( q,v, \frac{\partial L}{\partial q}, \frac{\partial L}{\partial v}\right) .  
\]
Using the  diffeomorphism $ \gamma_{Q}:T^{\ast}TQ \rightarrow T^{\ast}T^{\ast}Q$, we introduce the {\it Dirac differential} of $L$ as 
\[
\mathbf{d}_{D} L:= \gamma_{Q} \circ \mathbf{d} L : TQ \to T^{\ast}T^{\ast}Q,
\]
which is locally given by
\begin{equation}\label{DiracDiff}
\mathbf{d}_{D} L(q,v)= \left(q,\frac{\partial L}{\partial v}, - \frac{\partial L}{\partial q},  v \right).
\end{equation}

\begin{remark}\rm
The Dirac differential form \eqref{DiracDiff} is closed in the sense that 
\begin{equation*}
\begin{aligned}
\mathbf{d}\circ\mathbf{d}_{D}L(q,v)&=\mathbf{d}\left(-\frac{\partial L}{\partial q^{i}}\mathbf{d}q^{i}+v^{i}\mathbf{d}\left(\frac{\partial L}{\partial v^{i}}\right)\right)\\
&=-\frac{\partial^2 L}{\partial v^{j}\partial q^{i}}\mathbf{d}v^{j}\wedge \mathbf{d}q^{i}-\frac{\partial^2 L}{\partial q^{j}\partial q^{i}}\mathbf{d}q^{j}\wedge \mathbf{d}q^{i}+ \mathbf{d}v^{i}\wedge \left(\frac{\partial^{2} L}{\partial q^{j}\partial v^{i}}\mathbf{d}q^{j}+\frac{\partial^{2} L}{\partial v^{j}\partial v^{i}}\mathbf{d}v^{j}\right)\\
&=0.
\end{aligned}
\end{equation*}
\end{remark}

\begin{definition}\rm\label{ContLagDiracDynSys}
Let $ \Delta _Q \subset TQ$ be a distribution on $Q$ and consider the induced Dirac structure $D_{\Delta_{Q}}$ on $T^{\ast}Q$. The {\it Lagrange--Dirac dynamical system} is given by 
\begin{equation}\label{LDirac_system} 
\left( (q(t),p(t),\dot{q}(t),\dot{p}(t)), \mathbf{d}_{D} L(q(t),v(t))\right)  \in D_{ \Delta _Q }(q(t),p(t)).
\end{equation} 
Any curve $(q(t),v(t),p(t)) \in TQ \oplus T^{\ast}Q,\; t\in [0,T]$ satisfying \eqref{LDirac_system} is called a {\it solution curve} of the Lagrange--Dirac dynamical system.
It immediately follows from \eqref{LocIndDiracStr}  that such a solution curve satisfies the equations of motion for the Lagrange--Dirac dynamical system
\begin{equation}\label{LDAPeqn}
p =\frac{\partial L}{\partial v }, \quad \dot{q} =v \in \Delta_Q(q), \quad  \dot{p} - \frac{\partial L}{\partial q}
\in \Delta_Q^{\circ} (q).
\end{equation}
\end{definition}
Notice that the equations of motion in \eqref{LDAPeqn} include the Lagrange--d'Alembert equations $ \dot{p}-{\partial L}/{\partial q} \in \Delta_{Q}^{\circ}(q)$, the Legendre transformation  $p={\partial L}/{\partial v}$ and the second-order condition $\dot{q}=v \in \Delta_{Q}(q)$. For the unconstrained case $\Delta_{Q}=TQ$, we recover the {\it implicit Euler--Lagrange equations}.

\section{Discrete Dirac structures in Lagrangian mechanics}\label{Sect:DiscDiracStr_LagSys}
So far, we have described the continuous setting of the induced Dirac structure and its associated Lagrange--Dirac dynamical systems. In this section, we will explore the discretization of the induced Dirac structure and its associated discrete Lagrange--Dirac systems.

\subsection{Discrete Dirac structures on manifolds}\label{Sec:DiscDiracStr}
First, we begin with the discretization of a continuous Dirac structure on a manifold, introducing $(\pm)$-finite difference maps, $(\pm)$-discrete dual parings, $(\pm)$-discrete two-forms, and $(\pm)$-discrete constraint spaces to construct $(\pm)$-discrete Dirac structures.

\paragraph{Discrete symplectic structures.} Let $M$ be an $n$-dimensional smooth manifold equipped with a two-form $\Omega_M$ and let us first consider the discretization of the two-form $\Omega_M$ by introducing $(\pm)$-finite difference maps.
Let $x^1,...,x^n$ be local coordinates for each point $x \in M$. Let $I=\{t \in \mathbb{R} \mid 0 \le t \le T\}$ be the space of time interval and recall $I_{d}=\{t_{k}\}_{k=0}^{N}$ is the associated increasing sequence of discrete points  with $I$, which is given by
$I_{d}:=\{ t_{k}=kh \in \mathbb{R} \mid k=0,...,N \in \mathbb{N} \}$ where $h=t_{k+1}-t_{k} \in \mathbb{R}^{+}$ is a constant time step. 
Associated with the continuous path space $\mathcal{C}(M)$,
recall the \textit{discrete path space} is given by
$
\mathcal{C}_d(M)=\left\{x_{d}: \{t_{k}\}_{k=0}^{N} \to M \right\}, 
$
which is isomorphic to $M \times \cdots \times M$ ($N +1$ copies). Then, corresponding to a continuous curve  $x: I \to M$ together with its image $x(t) \in M$, we identify a \textit{discrete path} $x_d \in \mathcal{C}_d(M)$ with its image 
$
x_{d}=\{ x_k \}_{k=0}^N,
$
where $x_k:=x_d(t_k)$.
\begin{definition}\rm \label{Def:DisMap}
Consider $(\pm)$-\textit{finite difference  maps} $\Psi_{M_{i}}^{\pm}: M \times M \to TM,\;i=1,2$ that are given by
\begin{equation*}
\begin{split}
&\hat{x}^{+}_{1}=\Psi^{+}_{M_{2}}(x_{0},x_{1}) =\left(x_1, \frac{x_{2}-x_{1}}{h}\right) \in T_{x_{1}}M,\\
&\hat{x}^{-}_{0}=\Psi^{-}_{M_{1}}(x_{0},x_{1}) =\left(x_0, \frac{x_{0}-x_{-1}}{h}\right) \in T_{x_{0}}M,
\end{split}
\end{equation*}
where  $\hat{x}^{+}_{1} \in T_{x_{1}}M$ and $\hat{x}^{-}_{0} \in T_{x_{0}}M$  denote the forward difference and backward difference approximations of the velocities  $\dot{x}_{1}=\dot{x}(t_{k+1})$ and $\dot{x}_{0}=\dot{x}(t_{k})$, respectively. 
\end{definition}

\begin{remark}\rm
The notations $\Psi_{M_{i}}^{\pm}$ ($i=1,2$) are adopted in this paper, whose subindex indicates the base point of the $i$-th manifold of $M\times M$, and $(+)$ (or $(-)$) means a forward difference approximation  (or backward difference approximation) at the specified base point. The  $(\pm)$-finite difference maps $\Psi_{M_{i}}^{\pm}$ ($i=1,2$) must be properly chosen so that they cannot violate the property of the discrete Dirac structures as shown below. In fact, for the case of $M=T^{\ast}Q$, there exist some twisted structures as in Def. \ref{def:DiscMap_Cot}; see also Remark \ref{rem:DiscMap_Cot} in \S\ref{Sec:DisDiracStuCotBun}.

\end{remark}

\paragraph{Discrete dual parings.} Let us introduce discrete analogues of the dual paring $\left<\cdot,\cdot\right>: T^{\ast}M \times TM \to \mathbb{R}$, i.e., $(\pm)$-discrete dual paring $\left<\cdot,\cdot\right>_{d\pm}: T^{\ast}M \times (M \times M) \to \mathbb{R}$ as follows.
\begin{definition}\label{def:pairing}
\rm
Define the $(+)$-{\it discrete dual paring} $\left<\cdot,\cdot\right>_{d+}$ between $\alpha_{x_1}=(x_{1},\alpha) \in T^{\ast}_{x_{1}}M$ and $(x_{0},x_{1}) \in M\times M$ as
\begin{equation*}
\begin{split}
\left\langle (x_{1},\alpha), (x_{0},x_{1}) \right\rangle_{d+}&=\left\langle \alpha_{x_1},\hat{x}^{+}_{1}\right\rangle=\left\langle \alpha_{x_1},\Psi_{M_{2}}^{+}(x_{0},x_{1})\right\rangle,
\end{split}
\end{equation*}
where $\hat{x}^{+}_{1}=\Psi^{+}_{M_{2}}(x_{0},x_{1})\in T_{x_{1}}M$. Similarly, define the $(-)$-{\it discrete dual paring} $\left<\cdot,\cdot\right>_{d-}$ between $\alpha_{x_0}=(x_{0},\alpha) \in T^{\ast}_{x_{0}}M$ and $(x_{0},x_{1}) \in M\times M$  as
\begin{equation*}
\begin{split}
\left\langle (x_{0},\alpha), (x_{0},x_{1}) \right\rangle_{d-}&=\left\langle \alpha_{x_0},\hat{x}^{-}_0\right\rangle=\left\langle \alpha_{x_0},\Psi_{M_{1}}^{-}(x_0,x_1)\right\rangle,
\end{split}
\end{equation*}
where $\hat{x}^{-}_0=\Psi^{-}_{M_{1}}({x}_{0},x_1)\in T_{x_{0}}M$.
\end{definition}
\begin{definition}\label{DiscTwoForm_M}\rm$\;$ Associated with the two-form $\Omega_M$ on $M$, define the {\it $(\pm)$-discrete two-form} ${\Omega}^{d\pm}_M$ on  $M$ by, for each $(x_0,x_1) \in M \times M$,
\begin{equation*}
\begin{split}
{\Omega}^{d+}_M((x_{0},x_{1}), v)&:=\Omega_M(x_1)(\hat{x}^{+}_{1}, v)=\Omega_M(x_1)(\Psi^{+}_{M_{2}}(x_{0},x_1), v), \quad \forall v \in T_{x_{1}}M,\\[1.5mm]
{\Omega}^{d-}_M((x_0,x_1), v)&:=\Omega_M(x_0)(\hat{x}^{-}_1, v)=\Omega_M(x_0)(\Psi^{-}_{M_{1}}(x_0,x_1), v), \quad \forall v \in T_{x_{0}}M.
\end{split}
\end{equation*}
\end{definition}
\begin{definition}\rm
We can define the $(\pm)$-\textit{discrete flat maps} $({\Omega}_M^{d\pm})^{\flat}: M \times M \to T^{\ast}M$ such that
\begin{equation*}
\begin{split}\hspace{1cm}
(\Omega^{d+}_M)^\flat(x_{0},x_{1})(v)&=\mathbf{i}_{(x_{0},x_{1})}\Omega^{d+}_M(v)\\
&=\mathbf{i}_{\hat{x}^{+}_{1}}\Omega_M(x_{1})(v)=\Omega_M(x_1)(\hat{x}^{+}_{1},v)\\
&=\mathbf{i}_{ \Psi^{+}_{2}(x_{0},x_{1})}\Omega_M(x_{1})(v)=\Omega_M(x_1)(\Psi^{+}_{M_{2}}(x_{0},x_{1}),v),\;\; \forall {v} \in T_{x_{1}}M,\\
(\Omega^{d-}_M)^\flat(x_{0},x_{1})(v)&=\mathbf{i}_{(x_0,x_1)}\Omega^{d-}_M(v)\\
&=\mathbf{i}_{\hat{x}^{-}_0}\Omega_M(x_{0})(v)=\Omega_M(x_0)(\hat{x}^{-}_0,v)\\
&=\mathbf{i}_{ \Psi^{-}_{1}(x_0,x_1)}\Omega_M(x_{0})(v)=\Omega_M(x_0)(\Psi^{-}_{M_{1}}(x_{0},x_{1}),v),\;\; \forall {v} \in T_{x_{0}}M.
\end{split}
\end{equation*}
\end{definition}
\paragraph{Discrete constraint spaces.} 
Next we consider the discretization of constraint sets, which are given by a regular distribution $\Delta_{M}$ on $M$; namely, $\Delta_{M}(x) \in T_{x}M$ has a constant rank for all $x \in M$.
%
\begin{definition}\rm\label{DiscConstraintSpace_M}
Associated with the distribution $\Delta_{M} \subset TM$ on $M$, define the $(\pm)$-\textit{discrete constraint spaces} $\Delta_{M}^{d\pm}  \subset M \times M$ as
\begin{equation*}
\begin{split}
\Delta_{M}^{d+}&=\left\{ (x_{0},x_{1}) \in M \times M \mid \hat{x}^{+}_{1}=\Psi^{+}_{M_{2}}(x_{0},x_{1})\in \Delta_{M}(x_{1})  \right\},
\\[1.5mm]
\Delta_{M}^{d-}&=\left\{ (x_{0},x_{1}) \in M \times M \mid \hat{x}^{-}_{0}=\Psi^{-}_{M_{1}}(x_{0},x_{1})\in \Delta_{M}(x_{0}) \right\}.
\end{split}
\end{equation*}

The \textit{annihilators of the $(\pm)$-discrete constraint spaces} are defined by 
\begin{equation*}
\begin{split}
(\Delta_{M}^{d+})^{\circ}(x_{1})&=\left\{\beta_{x_{1}} \in T^{\ast}_{x_{1}}M \mid \left<(x_{1},\beta),(x_{0},x_{1})\right>_{d+}=\left<\beta, \Psi^{+}_{M_{2}}(x_{0},x_{1})\right> =0, \; \forall (x_{0},x_{1})\in \Delta^{d+}_M\right\},\\[1mm]
(\Delta_{M}^{d-})^{\circ}(x_{0})&=\left\{\beta_{x_{0}}\in T^{\ast}_{x_{0}}M \mid \left<(x_{0},\beta),(x_{0},x_{1})\right>_{d-}=\left<\beta, \Psi^{-}_{M_{1}}(x_{0},x_{1})\right>=0,  \; \forall (x_{0},x_{1})\in \Delta^{d-}_M\right\}.
\end{split}
\end{equation*}
\end{definition}

\paragraph{Discrete Dirac structures.}
Recall from \eqref{ContinuousDiracStructure} that the induced Dirac structure on $M$ can be defined by the distribution $\Delta_M$ and the bundle map $\Omega_{M}^{\flat}: TM \to T^{\ast}M$ associated with the two-form $\Omega_{M}$; namely, for each $x \in M$, 
\begin{equation*}
D_M(x):=\left\{(X_x,\alpha_x)\in T_xM\oplus T^*_xM\mid X_x\in\Delta_M(x)\text{ and } \alpha_x-\Omega^{\flat}_{M}(x)\cdot X_{x}\in \Delta_M^{\circ}(x)\right\}.
\end{equation*}
Now, as a discrete analogue of $D_{M}$ on $M$, we propose $(\pm)$-discrete induced Dirac structures  $D^{d\pm}_{M}$ on $M$ by the following theorem.

\begin{framed}
\begin{theorem}\label{thm:discreteDirac_M}\rm 
Let $\Omega_{M}^{d\pm}$ be the discrete two-form given in Def. \ref{DiscTwoForm_M} and let $\Delta_{M}^{d\pm}$ be the discrete constraint spaces given in Def. \ref{DiscConstraintSpace_M}. Then, the $(+)$-discrete structure $D^{d+}_{M} \subset (M\times M)\times T^{\ast}M$ that is defined by, for each $x_{1} \in M$,
\begin{equation*}
\begin{aligned}
D^{d+}_{M}(x_{1}):&=\left\{((x_{0},x_{1}),\alpha_{x_1})\in(M\times M)\times T^{\ast}_{x_{1}}M \mid  (x_{0},x_{1})\in \Delta_{M}^{d+}\; \right.\\
&\left.\hspace{5cm}
 \text{ and } \;\alpha_{x_1}-(\Omega^{d+}_M)^\flat(x_{0},x_{1})\in (\Delta_{M}^{d+})^{\circ}(x_1) \right\},\\[2mm]
\end{aligned}
\end{equation*}
is a Dirac structure on $M$.
Further,  the $(-)$-discrete structure $D^{d-}_{M} \subset (M\times M)\times T^{\ast}M$ that is defined by, for each $x_{0} \in M$,
\begin{equation*}
\begin{aligned}
D^{d-}_{M}(x_{0}):&=\left\{((x_0,x_1),\alpha_{x_0})\in(M\times M)\times T^{\ast}_{x_{0}}M \mid  (x_{0},x_1)\in \Delta_{M}^{d-}\;
 \right.\\
&\left.\hspace{5cm}
 \text{ and } \;\alpha_{x_0}-({\Omega}^{d-}_M)^{\flat}(x_{0},x_1)\in (\Delta_{M}^{d-})^{\circ}(x_0)\right\}
\end{aligned}
\end{equation*}
is a Dirac structure over $M$.
\end{theorem}
\end{framed}
\begin{proof} Let us prove $D_M^{d\pm}\subset (D_M^{d\pm})^{\perp}$ and $(D_M^{d\pm})^{\perp}\subset D_M^{d\pm}$, where the orthogonal subspaces of $D_M^{d\pm}$ are defined as
\begin{equation}\label{DiscretePerpDirac}
(D_M^{d\pm})^{\perp}=\left\{(u,\beta)\in (M\times M)\oplus T^*M\mid \langle\alpha,u\rangle_{d\pm}+\langle\beta,v\rangle_{d\pm}=0 \text{ for all } (v,\alpha)\in D_M^{d\pm}\right\}.
\end{equation}

First let us check $D_M^{d\pm}\subset (D_M^{d\pm})^{\perp}$. Let $(v,\alpha)$ and $(v',\alpha')$ be two elements of $D_M^{d\pm}$. Recall that any element $(v,\alpha)\in D_M^{d\pm}$ satisfies the following properties:
\begin{equation*}
v\in \Delta_M^{d\pm}\; \text{ and }\; \alpha-\mathbf{i}_v{\Omega}^{d\pm}_M\in (\Delta_M^{d\pm})^{\circ},
\end{equation*}
and also for any element $(v{'},\alpha{'})\in D_M^{d\pm}$, similarly we have
\begin{equation*}
v{'}\in \Delta_M^{d\pm}\; \text{ and }\; \alpha{'}-\mathbf{i}_{v{'}}{\Omega}^{d\pm}_M\in (\Delta_M^{d\pm})^{\circ}.
\end{equation*}
Hence these properties lead to
\begin{equation*}
\begin{aligned}
\langle \alpha,v'\rangle_{d\pm}+\langle\alpha',v\rangle_{d\pm}&=\left\langle\mathbf{i}_v{\Omega}^{d\pm}_M,v'\right\rangle_{d\pm}+\left\langle\mathbf{i}_{v^{\prime}}{\Omega}^{d\pm}_M,v\right\rangle_{d\pm}\\
&=\mathbf{i}_{\Psi_M^{\pm}(v')}\mathbf{i}_{\Psi_M^{\pm}(v)}\Omega_{M}+\mathbf{i}_{\Psi_M^{\pm}(v)}\mathbf{i}_{\Psi_M^{\pm}(v')}\Omega_{M}\\
&=0,
\end{aligned}
\end{equation*}
from which it follows $(v{'},\alpha{'})\in (D_M^{d\pm})^{\perp}$. Thus $D_M^{d\pm}\subset (D_M^{d\pm})^{\perp}$.
\medskip

Next let us check $(D_M^{d\pm})^{\perp}\subset D_M^{d\pm}$. First, let $(u,\beta)\in (D_M^{d\pm})^{\perp}$ and by definition, it follows
\begin{equation}\label{d_orthogonal}
\langle\alpha,u\rangle_{d\pm}+\langle\beta,v\rangle_{d\pm}=0
\end{equation}
for all $(v,\alpha)\in (M\times M)\oplus T^{\ast}M$ such that $v \in \Delta_{M}^{d\pm}$ and $\alpha-\mathbf{i}_{v}\Omega^{d\pm}_M\in (\Delta_{M}^{d\pm})^{\circ}$, namely, $(v,\alpha)\in D_M^{d\pm}$. Now, choose $v=(0,0) \in M \times M$ which is indeed an element of $\Delta_M^{d\pm}$ and hence it  follows $\alpha-\mathbf{i}_v{\Omega}^{d\pm}_{M}=\alpha\in (\Delta_M^{d\pm})^{\circ}$.  Then, by definition of \eqref{DiscretePerpDirac}, for any element $(u,\beta)\in (D_M^{d\pm})^{\perp}$, it follows
\begin{equation*}
\langle\alpha,u\rangle_{d\pm}=0,
\end{equation*}
from which it follows $u\in \Delta_M^{d\pm}$. Second, let $v\in \Delta_M^{d\pm}$ be arbitrary and suppose that $\left<\alpha, w\right>_{d\pm}=\left\langle\mathbf{i}_v{\Omega}^{d\pm}_M,w\right\rangle_{d\pm}$ for all $w\in \Delta_M^{d\pm}$. Since we already know $u\in \Delta_M^{d\pm}$, we may set $\left<\alpha, u\right>_{d\pm}=\left\langle\mathbf{i}_v{\Omega}^{d\pm}_M,u\right\rangle_{d\pm}$ for all $u\in \Delta_M^{d\pm}$, while from \eqref{d_orthogonal}, we get 
\[
\left\langle\mathbf{i}_v{\Omega}^{d\pm}_M,u\right\rangle_{d\pm}+\langle\beta,v\rangle_{d\pm}=0
\]
for all $v \in \Delta_{M}^{d\pm}$. In other words, $\langle\beta,v\rangle_{d\pm}=\left\langle\mathbf{i}_u{\Omega}^{d\pm}_M,v\right\rangle_{d\pm}$ for all $v \in \Delta_{M}^{d\pm}$. Therefore, $(u,\beta)\in D_M^{d\pm}$.
 Thus $(D_M^{d\pm})^{\perp}\subset D_M^{d\pm}$, and the proof completes.
\end{proof}
We call the discrete structures $D^{d\pm}_{M} \subset (M\times M)\oplus T^{\ast}M$ over $M$ the {\it $(\pm)$-discrete Dirac structures} on $M$.

\subsection{Discrete Dirac structures on cotangent bundles}\label{Sec:DisDiracStuCotBun} 
 Now we shall apply the general theory of the discrete Dirac structure on a manifold $M$ described in \S\ref{Sec:DiscDiracStr} to the special case of a Dirac structure on the cotangent bundle $M=T^{\ast}Q$ of a configuration manifold $Q$, which is induced from a given constraint distribution $\Delta_{Q}$ on $Q$.

 \paragraph{The $(\pm)$-finite difference maps.} As before, corresponding to a continuous curve 
$q: I \to Q$, where $I=\{t \in \mathbb{R} \mid 0 \le t \le T\}$ is the space of time interval, we consider 
the discrete path space by
$
\mathcal{C}_d(Q)=\left\{q_{d}: I_{d} \to Q \right\}, 
$
where $I_{d}:=\{ t_{k}=kh \in \mathbb{R} \mid k=0,...,N \in \mathbb{N} \}$ denotes the increasing sequence of discrete points  with a constant time step $h=t_{k+1}-t_{k} \in \mathbb{R}^{+}$.

\begin{definition}\rm\label{def:DiscMap_Cot}
Consider the $(\pm)$-finite difference maps $\Psi^{\pm}_{(T^{\ast}Q)_{i}}:T^{\ast}Q \times T^{\ast}Q \to T(T^{\ast}Q),\;i=1,2$ as 
\begin{equation*}
\begin{split}
\Psi_{(T^{\ast}Q)_{2}}^{+}&: T^{\ast}Q \times T^{\ast}Q \to TT^{\ast}Q; (z_{0},z_{1}) \mapsto \hat{z}^+_{1}=\Psi^{+}_{(T^{\ast}Q)_{2}}(z_{0},z_{1}) \in T_{z_{1}}T^{\ast}Q,\\[3mm]
\Psi_{(T^{\ast}Q)_{1}}^{-}&: T^{\ast}Q \times T^{\ast}Q \to TT^{\ast}Q; (z_{0},z_{1}) \mapsto  \hat{z}^{-}_{0}= \Psi^{-}_{(T^{\ast}Q)_{1}}(z_{0},z_{1}) \in T_{z_{0}}T^{\ast}Q,
\end{split}
\end{equation*}
each of which is denoted, in local coordinates $(z_{0},z_{1})=(q_{0},p_{0}, q_1, p_1) \in T^{\ast}Q \times T^{\ast}Q$, as
\begin{equation*}
\begin{split}
(z_{0},z_{1})&=(q_{0},p_{0}, q_1, p_1) \in T^{\ast}Q \times T^{\ast}Q \mapsto \\
&\hat z_{1}^{+}=\Psi^{+}_{(T^{\ast}Q)_{2}}(z_{0},z_{1})=(q_1, p_1,\hat q_{1}^{-}, \hat p_{1}^{+}) =
\left(q_{1}, p_{1},  \frac{q_{1}-q_{0}}{h}, \frac{p_2-p_{1}}{h} \right)
\in T_{(q_1,p_1)}T^{\ast}Q,
\end{split}
\end{equation*}
and 
\begin{equation*}
\begin{split}
(z_{0},z_{1})&=(q_0,p_0, q_1, p_1) \in T^{\ast}Q \times T^{\ast}Q \mapsto \\
&\hat z_{0}^{-}=\Psi^{-}_{(T^{\ast}Q)_{1}}(z_{0},z_{1})=(q_0, p_0,\hat q_{0}^{+}, \hat p_{0}^{-}) 
=\left( q_0,p_0, \frac{q_{1}-q_{0}}{h}, \frac{p_0-p_{-1}}{h} \right)\in T_{(q_0,p_0)}T^{\ast}Q.
\end{split}
\end{equation*}

\end{definition}
\begin{remark}\rm\label{rem:DiscMap_Cot}
In the construction of the $(\pm)$-dicretization maps $\Psi_{(T^{\ast}Q)_{i}}^{\pm}: T^{\ast}Q \times T^{\ast}Q \to T(T^{\ast}Q)$  ($i=1,2$) on the cotangent bundle $T^{\ast}Q$, one might question why these maps have twisted structures.  Initially, one might think that the following  $(\pm)$-finite difference maps $\Psi_{(T^{\ast}Q)_{i}}^{\pm}: T^{\ast}Q \times T^{\ast}Q \to T(T^{\ast}Q)$  ($i=1,2$)  are natural choices: 
\begin{equation*}
\begin{split}
\hat z_{1}^{+}=\Psi^{+}_{(T^{\ast}Q)_{2}}(z_{0},z_{1})=(q_1, p_1,\hat q_{1}^{+}, \hat p_{1}^{+}),\;\;
\hat z_{0}^{-}=\Psi^{-}_{(T^{\ast}Q)_{1}}(z_{0},z_{1})=(q_0, p_0,\hat q_{0}^{-}, \hat p_{0}^{-}).
\end{split}
\end{equation*}
However, these choices are to be inappropriate, because discrete flat maps $({\Omega}^{d\pm}_{T^{\ast}Q})^{\flat}$ that are defined later by using $\Psi^{\pm}_{(T^{\ast}Q)_{i}}$ do not preserve symplectic properties in the unconstrained cases. The correct choices are:
\begin{equation*}
\begin{split}
\hat z_{1}^{+}=\Psi^{+}_{(T^{\ast}Q)_{2}}(z_{0},z_{1})=(q_1, p_1,\hat q_{1}^{-}, \hat p_{1}^{+}),\;\;
\hat z_{0}^{-}=\Psi^{-}_{(T^{\ast}Q)_{1}}(z_{0},z_{1})=(q_0, p_0,\hat q_{0}^{+}, \hat p_{0}^{-}),
\end{split}
\end{equation*}
which may result in preserving the discrete canonical symplectic structures. In fact,  the choice of $\hat z_{1}^{+}=\Psi^{+}_{(T^{\ast}Q)_{2}}(z_{0},z_{1})=(q_1, p_1,\hat q_{1}^{-}, \hat p_{1}^{+})$ leads to the symplectic Euler-A method and the choice of $\hat z_{0}^{-}=\Psi^{-}_{(T^{\ast}Q)_{1}}(z_{0},z_{1})=(q_0, p_0,\hat q_{0}^{+}, \hat p_{0}^{-})$ the symplectic Euler-B method. 

Regarding the symplectic Euler-A and Euler-B methods, for instance, refer to \cite{LeRe2004} and we will address the details on these issues within the context of discrete Hamiltonian systems in a future study.
\end{remark}

\paragraph{Discrete dual pairings.}  Define the $(+)$-{\it discrete dual paring} between $(z_{0},z_{1})=(q_{0}, p_{0},q_1,p_1) \in T^{\ast}Q \times T^{\ast}Q$ and $(z_{1},\zeta)=(q_1,p_1,\eta,\xi)\in T^{\ast}_{z_{1}}T^{\ast}Q$ as
\begin{equation*}
\begin{split}\hspace{1cm}
\left<(z_{1},\zeta), (z_{0},z_{1}) \right>_{d+} &=\left<(z_{1},\zeta), \Psi_{(T^{\ast}Q)_2}^{+}(z_{0},z_{1})) \right>\\
&=\left<(z_{1},\zeta),(z_{1},\hat{z}^{+}_{1}) \right>\\
&=\left<(q_1,p_1,\eta,\xi), \Psi_{(T^{\ast}Q)_2}^{+}(q_{0},  p_{0},q_1,p_1) \right>\\
&=\left<(q_1,p_1,\eta,\xi), (q_1,p_1, \hat q_{1}^{-}, \hat p_{1}^{+}) \right>\\
&=\left<\eta, \hat q_{1}^{-}\right>+ \left<\hat p_{1}^{+},\xi\right>,
 \end{split}
\end{equation*}
where $(z_{1},\hat{z}^{+}_{1})=(q_1,p_1,\hat q_{1}^{-}, \hat p_{1}^{+}) \in T_{z_{1}}T^{\ast}Q$.
\medskip

Similarly, define the $(-)$-{\it discrete dual paring} between $(z_{0},z_{1})=(q_0,p_0, q_1, p_1) \in T^{\ast}Q \times T^{\ast}Q$ and $(z_{0}, \zeta)=(q_0,p_0,\eta,\xi) \in T^{\ast}_{z_{0}}T^{\ast}Q$ as
\begin{equation*}
\begin{split}\hspace{1cm}
\left<(z_{0}, \zeta), (z_{0},z_{1}) \right>_{d-} &=\left<(z_{0}, \zeta), \Psi_{(T^{\ast}Q)_1}^{-}(z_{0},z_{1}) \right>\\
&=\left<(z_{0},\zeta),(z_{0},\hat{z}^{-}_{0}) \right>\\
&=\left<(q_0,p_0,\eta,\xi), \Psi_{(T^{\ast}Q)_1}^{-}(q_0,p_0, q_1, p_1) \right>\\
&=\left<(q_0,p_0,\eta,\xi), (q_0,p_0, \hat q_0^{+}, \hat p_0^{-}) \right>\\
&=\left<\eta, \hat q_0^{+}\right>+ \left<\hat p_0^{-},\xi\right>,
 \end{split}
\end{equation*}
where $(z_{0},\hat{z}^{-}_{0})=(q_0,p_0,\hat q_0^{+}, \hat p_0^{-}) \in T_{z_{0}}T^{\ast}Q$.
\paragraph{Discretization of the canonical symplectic structure on $T^{\ast}Q$.} Associated with the canonical symplectic forms on the cotangent bundle $T^{\ast}Q$, we shall introduce the corresponding $(\pm)$-discrete canonical symplectic forms as follows.
\medskip

First, recall from \cite{MaRa1999} that the cotangent bundle $T^{\ast}Q$  is naturally equipped with the symplectic one-form $\Theta_{T^{\ast}Q}$ as in equation \eqref{CanOneForm_Cot}. Namely, $\Theta_{T^{\ast}Q}$ is the horizontal one-form  on $T^{\ast}Q$ such that, for $z \in T^{\ast}Q$,
\begin{equation*}
\begin{split}
\Theta_{T^{\ast}Q}(z)(v_{z})=\left<z, T\pi_{Q}(v_{z})\right>,
\end{split}
\end{equation*}
where $v_{z}\in T_{z}T^{\ast}Q$ and $T\pi_{Q}: TT^{\ast}Q \to TQ$ is the tangent map of the cotangent bundle projection $\pi_{Q}: T^{\ast}Q \to Q$.
\begin{definition}\label{Def:discSymOneForm}\rm
Associated with the canonical one-form $\Theta_{T^{\ast}Q}$ on $T^{\ast}Q$, define 
the $(+)$-\textit{discrete canonical  one-form} $\Theta_{T^{\ast}Q}^{d+}$ by, for each $(z_{0},z_{1}) \in T^{\ast}Q \times T^{\ast}Q$,
\begin{equation*}
\begin{split}
\Theta^{d+}_{T^{\ast}Q}(z_{0}, z_1)&= \Theta_{T^{\ast}Q}(z_1)\left(\Psi^{+}_{(T^{\ast}Q)_2}(z_{0}, z_1)\right)\\
&=\Theta_{T^{\ast}Q}(z_1)(\hat{z}^{+}_{1}),
\end{split}
\end{equation*}
which is denoted in coordinates by, for $(z_{0},z_{1})=(q_{0}, p_{0}, q_1, p_1) \in T^{\ast}Q \times T^{\ast}Q$,
\begin{equation*}
\begin{split}
\Theta^{d+}_{T^{\ast}Q}(z_{0}, z_1)=\Theta_{T^{\ast}Q}(z_1)(\hat{z}^{+}_{1})=\left( p_{1}dq_{1}\right)\left(\hat{q}^{-}_{1} \frac{\partial }{\partial q_{1}}+ \hat{p}^{+}_{1} \frac{\partial }{\partial p_{1}}\right)=\langle p_{1}, \hat{q}^{-}_{1}\rangle.
\end{split}
\end{equation*}

Similarly, define the $(-)$-\textit{discrete canonical one-form} $\Theta_{T^{\ast}Q}^{d+}$ by, for each $(z_{0},z_{1}) \in T^{\ast}Q \times T^{\ast}Q$,
\begin{equation*}
\begin{split}
\Theta^{d-}_{T^{\ast}Q}(z_0, z_1)&= \Theta_{T^{\ast}Q}(z_0)\left(\Psi^{-}_{(T^{\ast}Q)_1}(z_0, z_1)\right)\\
&=\Theta_{T^{\ast}Q}(z_0)(\hat{z}^{-}_{0}),
\end{split}
\end{equation*}
which is denoted, in local coordinates $(z_{0},z_{1})=(q_{0}, p_{0}, q_1, p_1) \in T^{\ast}Q \times T^{\ast}Q$, by
\begin{equation*}
\begin{split}
\Theta^{d-}_{T^{\ast}Q}(z_0, z_1)=\left(p_{0}dq_{0}\right)\left(\hat{q}^{+}_{0} \frac{\partial }{\partial q_{0}}+ \hat{p}^{-}_{0} \frac{\partial }{\partial p_{0}}\right)=\langle p_{0}, \hat{q}^{+}_{0}\rangle.
\end{split}
\end{equation*}
\end{definition}

\begin{definition}\rm
Associated with the symplectic structure $\Omega_{T^{\ast}Q}$ on $T^{\ast}Q$, the $(+)$-\textit{discrete canonical two-form (discrete canonical symplectic structure)} ${\Omega}_{T^{\ast}Q}^{d+}$ is defined by, for each $(z_{0},z_{1}) \in T^{\ast}Q \times T^{\ast}Q$,
\begin{equation*}
\begin{split}
\Omega^{d+}_{T^{\ast}Q}((z_{0},z_1),{v})&= \Omega_{T^{\ast}Q}(z_{1})\left(\Psi^{+}_{(T^{\ast}Q)_2}(z_{0},z_1),{v}\right)\\
&=\Omega_{T^{\ast}Q}(z_{1})\left(\hat z^{+}_{1},{v}\right), \;\; \textrm{for all} \; {v} \in T_{z_1}T^{\ast}Q,
\end{split}
\end{equation*}
which is denoted, in local coordinates $(z_{0},z_{1})=(q_{0}, p_{0},q_{1}, p_{1}) \in T^{\ast}Q \times T^{\ast}Q$, by
\[
\Omega^{d+}_{T^{\ast}Q}((q_{0}, p_{0},q_{1}, p_{1}),{v})= \Omega_{T^{\ast}Q}(q_1,p_1)\left( (\hat q_{1}^{-}, \hat p_{1}^{+}),{v}\right), \;\;\textrm{for all} \; {v} \in T_{(q_1,p_1)}T^{\ast}Q.
\]

Similarly, the $(-)$-\textit{discrete canonical two-form (discrete canonical symplectic structure)} ${\Omega}_{T^{\ast}Q}^{d-}$ is defined by, for each $(z_{0},z_{1}) \in T^{\ast}Q \times T^{\ast}Q$,
\begin{equation*}
\begin{split}
\Omega^{d-}_{T^{\ast}Q}((z_0, z_1),{v})&= \Omega_{T^{\ast}Q}(z_0)\left(\Psi^{-}_{(T^{\ast}Q)_1}(z_0, z_1),{v}\right)\\
&= \Omega_{T^{\ast}Q}(z_0)\left( \hat z^{-}_{0},{v}\right), \;\;\textrm{for all ${v} \in T_{z_0}T^{\ast}Q$},
\end{split}
\end{equation*}
which is represented, in local coordinates $(z_{0},z_{1})=(q_{0}, p_{0}, q_1, p_1) \in T^{\ast}Q \times T^{\ast}Q$, by
\[
\Omega^{d-}_{T^{\ast}Q}((q_0,p_0, q_1, p_1),{v})= \Omega_{T^{\ast}Q}(q_0, p_0)\left( (\hat q_0^{+}, \hat p_0^{-}),{v}\right), \;\;\textrm{for all ${v} \in T_{(q_0, p_0)}T^{\ast}Q$}.
\]
Note that the relations $\Omega^{d\pm}_{T^{\ast}Q}=-\mathbf{d}\Theta^{d\pm}_{T^{\ast}Q}$ hold, since $\Omega_{T^{\ast}Q}=-\mathbf{d}\Theta_{T^{\ast}Q}$.

\end{definition}

\begin{definition}\rm\label{def:DisCanSymFlatMap}
Associated with the $(+)$-discrete symplectic two-form $\Omega^{d+}_{T^{\ast}Q}$ on $T^{\ast}Q$,
the $(+)$-{\it discrete canonical symplectic flat map} 
\[
({\Omega}^{d+}_{T^{\ast}Q})^{\flat}: T^{\ast}Q \times T^{\ast}Q \to T^{\ast}T^{\ast}Q
\]
is defined such that, for each $(z_{0},z_{1}) \in T^{\ast}Q \times T^{\ast}Q$,
\begin{equation*}
\begin{split}
(\Omega^{d+}_{T^{\ast}Q})^\flat(z_{0},z_{1})(v)=\Omega^{d+}_{T^{\ast}Q}((z_{0},z_1),{v}),\quad \textrm{for any $v \in T_{z_1}T^{\ast}Q$}.
\end{split}
\end{equation*}
Consequently,
\begin{equation*}
\begin{split}
(\Omega^{d+}_{T^{\ast}Q})^\flat(z_{0},z_{1})=\mathbf{i}_{(z_{0},z_{1})}\Omega^{d+}_{T^{\ast}Q}=\mathbf{i}_{\hat z_{1}^{+}}\Omega_{T^{\ast}Q}(z_1)=(\Omega^{\flat}_{T^{\ast}Q})(z_{1})({\hat z_{1}^{+}}).
\end{split}
\end{equation*}

\medskip

Similarly, associated with the $(-)$-discrete symplectic two-form $\Omega^{d-}_{T^{\ast}Q}$ on $T^{\ast}Q$, the $(-)$-{\it discrete canonical symplectic flat map} 
\[
({\Omega}^{d-}_{T^{\ast}Q})^{\flat}: T^{\ast}Q \times T^{\ast}Q \to T^{\ast}T^{\ast}Q
\]
is defined such that, for each $(z_{0},z_{1}) \in T^{\ast}Q \times T^{\ast}Q$,
\begin{equation*}
\begin{split}
(\Omega^{d-}_{T^{\ast}Q})^{\flat}(z_0, z_1)(v)=\Omega^{d-}_{T^{\ast}Q}((z_{0},z_1),{v}),\quad\textrm{for any $v \in T_{z_0}T^{\ast}Q$}.
\end{split}
\end{equation*}
This can be equivalently written as
\begin{equation*}
\begin{split}
(\Omega^{d-}_{T^{\ast}Q})^\flat(z_0, z_1)&=\mathbf{i}_{(z_0, z_1)}\Omega^{d-}_{T^{\ast}Q}=\mathbf{i}_{\hat z_{0}^{-}}\Omega_{T^{\ast}Q}(z_0)=\Omega^{\flat}_{T^{\ast}Q}(z_0)(\hat z_{0}^{-}).
\end{split}
\end{equation*}
\end{definition}

\paragraph{Local expressions for the $(\pm)$-discrete symplectic structures.} Here we shall derive local expressions for the $(\pm)$-discrete canonical symplectic structures. 
\medskip

According to Def. \ref{def:DisCanSymFlatMap}, the $(+)$-discrete flat map is described in local coordinates $(z_{0},z_{1})=(q_{0}, p_{0},q_1,p_1) \in T^{\ast}Q \times T^{\ast}Q$ as
\begin{equation}\label{disceteOmegaFlat_plus}
(\Omega^{d+}_{T^{\ast}Q})^{\flat}: T^{\ast}Q \times T^{\ast}Q \to T^{\ast}T^{\ast}Q;\quad
(q_{0}, p_{0},q_1,p_1)\mapsto (q_1,p_1,-\hat p_{1}^{+}, \hat q_{1}^{-}),
\end{equation}
which reads
\begin{equation*}
\begin{split}
(\Omega^{d+}_{T^{\ast}Q})^\flat(q_{0}, p_{0},q_1,p_1)&=\mathbf{i}_{(q_{0}, p_{0},q_1,p_1)}\Omega^{d+}_{T^{\ast}Q}
=\mathbf{i}_{(\hat q_{1}^{-}, \hat p_{1}^{+})}\Omega_{T^{\ast}Q}(q_1,p_1)\\
&=\Omega^{\flat}_{T^{\ast}Q}(q_1,p_1)(\hat q_{1}^{-}, \hat p_{1}^{+})
=(q_1,p_1,-\hat p_{1}^{+}, \hat q_{1}^{-}),
\end{split}
\end{equation*}
where $(\hat q_{1}^{-}, \hat p_{1}^{+})=\Psi^{+}_{(T^{\ast}Q)_2}(q_{0}, p_{0},q_1,p_1)\in T_{(q_1,p_1)}T^{\ast}Q$.
\medskip

 Similarly, the $(-)$-discrete flat map is described in local coordinates $(z_{0},z_{1})=(q_0,p_0, q_1, p_1) \in T^{\ast}Q \times T^{\ast}Q$ as
\begin{equation}\label{disceteOmegaFlat_minus}
(\Omega^{d-}_{T^{\ast}Q})^{\flat}: T^{\ast}Q \times T^{\ast}Q \to T^{\ast}T^{\ast}Q;\quad
(q_0,p_0, q_1, p_1)\mapsto (q_0,p_0,-\hat p_0^{-}, \hat q_0^{+}),
\end{equation}
which reads
\begin{equation*}
\begin{split}
(\Omega^{d-}_{T^{\ast}Q})^\flat(q_0,p_0, q_1, p_1)&=\mathbf{i}_{(q_0,p_0, q_1, p_1)}\Omega^{d-}_{T^{\ast}Q}=\mathbf{i}_{(\hat q_0^{+}, \hat p_0^{-})}\Omega_{T^{\ast}Q}(q_0,p_0)\\
&=\Omega^{\flat}_{T^{\ast}Q}(q_0,p_0)(\hat q_0^{+}, \hat p_0^{-})=(q_0,p_0,-\hat p_0^{-}, \hat q_0^{+}),
\end{split}
\end{equation*}
where $(\hat q_0^{+}, \hat p_0^{-})= \Psi^{-}_{(T^{\ast}Q \times T^{\ast}Q)_{1}}(q_0,p_0, q_1, p_1)\in T_{(q_0,p_0)}T^{\ast}Q$.
\paragraph{Discrete constraint spaces.}
Next, we shall consider the $(\pm)$-discrete constraint spaces on $Q \times Q$ associated with the given distribution $\Delta_{Q} \subset TQ$ and also define the $(\pm)$-induced discrete subspaces on the cotangent bundle $T^{\ast}Q$.

\begin{definition}\rm\label{DiscConstSpaces_Q}
Let $\Delta_Q^{d\pm} \subset Q \times Q$ be $(\pm)$-{\it discrete constraint spaces}   defined as
\begin{equation*}
\begin{split}
\Delta_{Q}^{d+}&:=\left\{ (q_{0},q_{1}) \in Q \times Q \mid \hat{q}^{-}_{1}= \Psi_{Q_{2}}^-(q_{0},q_{1})=\left(q_{1},\frac{q_{1}-q_{0}}{h}\right) \in \Delta_Q(q_{1}) 
\right\},\\[2mm]
\Delta_{Q}^{d-}&:=\left\{ (q_{0},q_{1}) \in Q \times Q \mid \hat{q}^{+}_{0}=\Psi_{Q_{1}}^+(q_{0},q_{1})=\left(q_{0},\frac{q_{1}-q_{0}}{h}\right)\in \Delta_Q(q_{0})\right\}.
\end{split}
\end{equation*}
\end{definition}


\begin{definition}\rm
Define the {\it $(\pm)$-discrete constraint spaces} $\Delta_{T^{\ast}Q}^{d\pm} \subset T^{\ast}Q \times T^{\ast}Q$ by lifting $\Delta_{Q}^{d\pm} \subset Q \times Q$ by
\[
\Delta_{T^{\ast}Q}^{d\pm}=(\pi_Q\times \pi_Q)^{-1}\Delta_{Q}^{d\pm},
\]
which are  given in local coordinates $(z_{0},z_{1})=(q_0,p_0,q_{1},p_{1})\in T^{\ast}Q \times T^{\ast}Q$ by
\[
\begin{split}
\Delta_{T^{\ast}Q}^{d\pm}&=\left\{(q_0,p_0,q_{1},p_{1}) \mid (\pi_Q\times \pi_Q) \circ (q_0,p_0,q_{1},p_{1}) =(q_0,q_{1})\in \Delta^{d\pm}_Q \right\}.
\end{split}
\]
\end{definition}
\begin{definition}\rm
The {\it annihilator of} $\Delta_{T^{\ast}Q}^{d+}$ is given by a subspace of $T^{\ast}T^{\ast}Q$ that is defined by
\begin{equation*}
\begin{split}
(\Delta_{T^{\ast}Q}^{d+})^{\circ}&=\left\{(z_1, \zeta)\in T^{\ast}_{z_{1}}T^{\ast}Q \mid \left<(z_1, \zeta), (z_{0},z_{1})\right>_{d+}=0,
 \; \forall (z_{0},z_{1})\in \Delta^{d+}_{T^{\ast}Q},\right\},
\end{split}
\end{equation*}
which is  denoted in local coordinates $(z_{1}, \zeta)=(q_{1},p_{1},\eta, \xi) \in T^{\ast}T^{\ast}Q$ and $(z_{0},z_{1})=(q_{0},p_{0},q_1,p_1)\in T^{\ast}Q \times T^{\ast}Q$ by
\begin{equation*}
\begin{split}
(\Delta_{T^{\ast}Q}^{d+})^{\circ}&=\left\{(q_{1},p_{1},\eta, \xi)  \mid
\left<(q_{1},p_{1},\eta, \xi), (q_{0}, p_{0},q_1,p_1)\right>_{d+}=0,\;\; \forall (q_{0}, p_{0},q_1,p_1) \in \Delta_{T^{\ast}Q}^{d+}
\right\},
\end{split}
\end{equation*}
where $(q_{0},q_1) \in \Delta_{Q}^{d+}$.
Recall that
\[
\left<(q_{1},p_{1},\eta, \xi), (q_{0},p_{0},q_1,p_1)\right>_{d+}= \left<(q_{1},p_{1},\eta, \xi), (q_1,p_1, \hat q_{1}^{-},\hat p_{1}^{+})\right>=\left<\eta, \hat q_{1}^{-}\right>+ \left< \hat p_{1}^{+},\xi\right>=0,
\]
for all $\hat q_{1}^{-}=\Psi^{-}_{Q_{2}}(q_{0},q_{1}) \in \Delta_{Q}(q_{1})$ and for all $\hat p_{1}^{+} \in T^{\ast}_{q_{1}}Q$. Therefore, we obtain $\eta \in \Delta_Q^{\circ}(q_{1})$ and $\xi=0$. Hence, the annihilator of $\Delta_{T^{\ast}Q}^{d+}$ is equivalently given by
\begin{equation}\label{discreteAnnihilator_plus}
\begin{split}
(\Delta_{T^{\ast}Q}^{d+})^{\circ}&=\left\{(q_{1},p_{1},\eta, \xi)  \mid
\eta \in \Delta_Q^\circ(q_1), \;\; \xi=0
\right\}.
\end{split}
\end{equation}
\quad Similarly, the {\bf annihilator of} $\Delta_{T^{\ast}Q}^{d-}$ is given by a subspace of $T^{\ast}T^{\ast}Q$ that is defined by
\begin{equation*}
\begin{split}
(\Delta_{T^{\ast}Q}^{d-})^{\circ}&=\left\{(z_0, \zeta)\in T^{\ast}T^{\ast}Q \mid 
\left<(z_0, \zeta), (z_0,z_1)\right>_{d-}=0,  \; \forall (z_0,z_1)\in \Delta^{d-}_{T^{\ast}Q}\right\},
\end{split}
\end{equation*}
which is  denoted in local coordinates $(z_0, \zeta)=(q_{0},p_{0},\eta, \xi)$ for $T^{\ast}T^{\ast}Q$ and $(z_{0},z_{1})=(q_0,p_0,q_1,p_1)$ for $T^{\ast}Q \times T^{\ast}Q$ by
\begin{equation*}
\begin{split}
(\Delta_{T^{\ast}Q}^{d-})^{\circ}&=\left\{(q_{0},p_{0},\eta, \xi)  \mid \left<(q_{0},p_{0},\eta, \xi), (q_0,p_0,q_1,p_1)\right>_{d-} =0, \; \forall (q_0,p_0,q_1,p_1) \in \Delta_{T^{\ast}Q}^{d-}
\right\},
\end{split}
\end{equation*}
where $(q_0,q_1) \in \Delta_{Q}^{d-}$.
Recall that
\begin{equation*}
\begin{split}
&\left<(q_{0},p_{0},\eta, \xi), (q_0,p_0,q_1,p_1)\right>_{d-}
=\left<(q_{0},p_{0},\eta, \xi), (q_0,p_0,\hat q_0^{+}, \hat p_0^{-})\right>=\left<\eta, \hat q_0^{+}\right>+\left<\hat p_0^{-},\xi \right>
=0,
\end{split}
\end{equation*}
for all $\hat q_0^{+}=\Psi_{Q_{1}}^{
+}(q_{0},q_{1}) \in \Delta_{Q}(q_{0})$ and for all $\hat p_0^{-} \in T^{\ast}_{q_0}Q$. Therefore, we get $\eta \in \Delta_Q^\circ(q_0)$ and $\xi=0$. Thus, the annihilator of $\Delta_{T^{\ast}Q}^{d-}$ is equivalently given by
\begin{equation}\label{discreteAnnihilator_minus}
\begin{split}
(\Delta_{T^{\ast}Q}^{d-})^{\circ}&=\left\{(q_{0},p_{0},\eta, \xi)  \mid
\eta \in \Delta_Q^\circ(q_0), \;\; \xi=0
\right\}.
\end{split}
\end{equation}

\end{definition}
\paragraph{Discrete induced Dirac structures on $T^{\ast}Q$.} Now we have the following theorem regarding the $(\pm)$-discrete induced Dirac structures over $T^{\ast}Q$, each of which is a discrete analogue of the induced Dirac structure $D_{\Delta_{Q}}$ on $T^{\ast}Q$.
\begin{framed}
\begin{theorem}\rm\label{DiscreteDiracStructure_Cot}
Let $({\Omega}_{T^*Q}^{d \pm})^{\flat}: T^{\ast}Q \times T^{\ast}Q \to T^{\ast}T^{\ast}Q$ be the $(\pm)$-discrete canonical symplectic flat maps defined in Def. \ref{def:DisCanSymFlatMap} and let $\Delta_Q^{d\pm} \subset Q \times Q$ be the $(\pm)$-discrete constraint spaces given in Def. \ref{DiscConstSpaces_Q}. Then, the $(+)$-discrete structure that is given by, for each $z_{1} \in T^{\ast}Q$,
\begin{equation*}
\begin{split}
D^{d+}_{\Delta_Q}(z_{1})&:=\left\{((z_{0},z_1),\alpha_{z_1})\in (T^{\ast}Q\times T^{\ast}Q)\times T^{\ast}_{z_1}T^{\ast}Q\mid\right.\\
&\hspace{3cm}\left. (z_{0},z_1)\in \Delta_{T^{\ast}Q}^{d+} \text{ and } \alpha_{z_{1}}-(\Omega^{d+}_{T^{\ast}Q})^{\flat}(z_{0},z_1)\in (\Delta_{T^{\ast}Q}^{d+})^{\circ}(z_1)
\right\}
\end{split}
\end{equation*}
is a discrete induced Dirac structure on $T^{\ast}Q$. Furthermore, the $(-)$-discrete structure that is defined by, for $z_{0} \in T^{\ast}Q$,
\begin{equation*}
\begin{split}
D^{d-}_{\Delta_Q}(z_{0})&:=\left\{((z_{0},z_{1}),\alpha_{z_0})\in (T^{\ast}Q\times T^{\ast}Q)\times T^{\ast}_{z_0}T^{\ast}Q\mid\right.\\
&\hspace{3cm}\left. (z_{0},z_{1})\in \Delta_{T^{\ast}Q}^{d-} \text{ and } \alpha_{z_{0}}-(\Omega^{d-}_{T^{\ast}Q})^{\flat}(z_{0},z_{1})\in (\Delta_{T^{\ast}Q}^{d-})^{\circ}(z_0)
\right\}
\end{split}
\end{equation*}
is a discrete induced Dirac structure on $T^{\ast}Q$.
\end{theorem}
\end{framed}
\begin{proof}
Recall Theorem \ref{thm:discreteDirac_M} for the discrete Dirac structure on $M$ and by applying Theorem \ref{thm:discreteDirac_M} to the case in which $M=T^{\ast}Q$, we can easily check that the discrete structures in Theorem \ref{DiscreteDiracStructure_Cot} are $(\pm)$-discrete induced Dirac structures on $T^{\ast}Q$ that are defined by the $(\pm)$-discrete canonical symplectic two-forms $\Omega^{d+}_{T^{\ast}Q}$ and the $(\pm)$-discrete constraint spaces $\Delta_{T^{\ast}Q}^{d\pm}=(\pi_Q\times \pi_Q)^{-1}\Delta_{Q}^{d\pm}$.
\end{proof}
The discrete structures $D^{d\pm}_{\Delta_Q} \subset (T^{\ast}Q\times T^{\ast}Q)\times T^{\ast}T^{\ast}Q$ are called the {\bf $(\pm)$-discrete induced Dirac structures}
 over $T^{\ast}Q$.

\paragraph{Local expressions for the $(\pm)$-discrete induced Dirac structures.}
By analogy with the local expression for the continuous induced Dirac structure given in \eqref{LocIndDiracStr}, we shall illustrate the local expressions for the $(\pm)$-discrete Dirac structures in Theorem \ref{DiscreteDiracStructure_Cot}.

\begin{proposition}\rm
Associated with the $(\pm)$-discrete Dirac structures, we have the $(\pm)$-formulas of local expressions  as follows.
\begin{itemize}
\item[(i)]
For the $(+)$-discrete induced Dirac structure $D^{d+}_{\Delta_{Q}}$, if $X=(q_{0},p_{0},q_{1},p_{1}) \in T^{\ast}Q \times T^{\ast}Q$ and $\alpha=(q_{1},p_{1},\beta,\gamma) \in T^{\ast}_{(q_{1},p_{1})}T^{\ast}Q$ satisfy the condition 
\begin{equation}\label{CondDisDiracStr_plus}
(X,\alpha) \in D^{d+}_{\Delta_{Q}}(q_{1},p_{1}),
\end{equation}
then the following equations are satisfied:
\begin{equation}\label{Loc_discrete_Dirac_plus}
(q_{0},q_1)\in \Delta^{d+}_Q,  \qquad \hat{q}^{-}_{1}=\gamma, 
\qquad \beta +\hat p^{+}_{1} \in \Delta_{Q}^{\circ}(q_{1}).
\end{equation}
\item[(ii)]
For the $(-)$-discrete induced Dirac structure $D^{d-}_{\Delta_{Q}}$, if $X=(q_{0},p_{0},q_{1},p_{1}) \in T^{\ast}Q \times T^{\ast}Q$ and $\alpha=(q_{0},p_{0},\beta,\gamma) \in T^{\ast}_{(q_{0},p_{0})}T^{\ast}Q$ satisfy the condition 
\begin{equation}\label{CondDisDiracStr_minus}
(X,\alpha) \in D^{d-}_{\Delta_{Q}}(q_{0},p_{0}),
\end{equation}
then the following equations are satisfied:
\begin{equation}\label{Loc_discrete_Dirac_minus}
(q_0,q_1)\in \Delta^{d-}_Q, \qquad \hat{q}^{+}_{0}=\gamma, \qquad \beta +\hat p^{-}_{0} \in \Delta_{Q}^{\circ}(q_{0}).
\end{equation}
\end{itemize}
\end{proposition}
\begin{proof}
Let us check these by direct computations.
\begin{itemize}
\item[(i)]
From the condition \eqref{CondDisDiracStr_plus}, we can get
\begin{equation*}
 (q_{0},p_{0},q_{1},p_{1})\in \Delta_{T^{\ast}Q}^{d+},\qquad  (q_{1},p_{1},\beta,\gamma)-({\Omega}^{d+}_{T^{\ast}Q})^{\flat}(q_{0},p_{0},q_{1},p_{1})\in (\Delta_{T^{\ast}Q}^{d+})^{\circ}(q_{1},p_{1}),
\end{equation*}
 The first equation yields $(q_0,q_1)\in \Delta_{Q}^{d+}$, while the second equation reads, in view of \eqref{disceteOmegaFlat_plus} and \eqref{discreteAnnihilator_plus},
\[
\beta 
+\hat p^{+}_{1} \in \Delta_{Q}^{\circ}(q_{1}),\quad  \gamma-\hat{q}^{-}_{1}=0.
\]
\item[(ii)]
From the condition \eqref{CondDisDiracStr_minus}, we can get
\begin{equation*}
 (q_0,p_0,q_1,p_1)\in \Delta_{T^{\ast}Q}^{d-},\qquad (q_{0},p_{0},\beta,\gamma)-({\Omega}^{d-}_{T^{\ast}Q})^{\flat}(q_0,p_0,q_1,p_1)\in (\Delta_{T^{\ast}Q}^{d-})^{\circ}(q_0,p_0).
\end{equation*}
The first equation gives $(q_0,q_1)\in \Delta_{Q}^{d-}$, while the second equation reads, in view of \eqref{disceteOmegaFlat_minus} and \eqref{discreteAnnihilator_minus},
\[
\beta +\hat p^{-}_{0} \in \Delta_{Q}^{\circ}(q_{0}), \quad \gamma-\hat{q}^{+}_{0}=0.
\]
\end{itemize}
\end{proof}
\section{Discrete Lagrange--Dirac dynamical systems}
In this section, we consider the $(\pm)$-discretizations of Lagrange--Dirac dynamical systems within the context of $(\pm)$-discrete induced Dirac structures. To do this, we first present the  $(\pm)$-discretizations of the iterated tangent and cotangent bundles. We then illustrate the  $(\pm)$-discretizations of the Dirac differential of a discrete Lagrangian and the $(\pm)$-discrete evolution maps over $T^{\ast}Q$.  
  
\subsection{Discretization of iterated bundles}
As in \S\ref{Sect:LagDiracDySys}, we have the geometric structure between the iterated tangent and cotangent bundles $TT^{\ast}Q, T^{\ast}TQ$ and $T^{\ast}T^{\ast}Q$ for the Lagrange--Dirac dynamical system  illustrated in Fig. \ref{fig:BundPic}.
Here we shall consider the discretization of $TT^{\ast}Q, T^{\ast}TQ$ and $T^{\ast}T^{\ast}Q$.
\paragraph{Tulczyjew's triple.}
Recall the following diffeomorphisms \eqref{eq:candiff} between these iterated bundles:
\begin{equation*}
\begin{split}
&\Omega_{T^{\ast}Q}^{\flat}: TT^{\ast}Q \to T^{\ast}T^{\ast}Q; \hspace{1.6cm}\, (q,p,\delta{q},\delta{p}) \mapsto (q,p,-\delta{p}, \delta{q}), \\
&\kappa_{Q}: TT^{\ast}Q \to T^{\ast}TQ;\hspace{2.2cm} (q, p, \delta q, \delta p) \mapsto (q, \delta q, \delta p, p),\\
&\gamma_{Q}:=\Omega^{\flat} \circ \kappa_{Q}^{-1}: T^{\ast}TQ \to T^{\ast}T^{\ast}Q;\; \quad (q,\delta{q},\delta{p},p) \mapsto (q,p,-\delta{p}, \delta{q}).
\end{split}
\end{equation*}
The first diffeomorphism is the flat map associated with the canonical symplectic structure $\Omega_{T^{\ast}Q}$. The second diffeomorphism was introduced by \cite{Tu1977} in the context of the {\it generalized Legendre transform}; see also \cite{YoMa2006b}.

Recall also from \cite{Tu1977} that the manifold $TT^{\ast}Q$ is the symplectic manifold with a particular symplectic form $\Omega_{TT^{\ast}Q}$ 
that can be defined from the two distinct one-forms as
\begin{equation*}
\begin{split}
\lambda&=(\kappa_{Q})^{\ast}\Theta_{T^{\ast}TQ}=\delta{p}\,dq+p\,d\delta{q}, \\
\chi&=(\Omega^{\flat})^{\ast}\Theta_{T^{\ast}T^{\ast}Q}=-\delta{p}\,dq+\delta{q}\,dp, 
\end{split}
\end{equation*}
where $\Theta_{T^{\ast}TQ}$ and $\Theta_{T^{\ast}T^{\ast}Q}$ are the canonical one-forms on $T^{\ast}TQ$ and $T^{\ast}T^{\ast}Q$ respectively. Hence the particular symplectic form $\Omega_{TT^{\ast}Q}$ is defined by
\[
\Omega_{TT^{\ast}Q}=-\mathbf{d}\lambda=\mathbf{d}\chi=dq \wedge d\delta{p}+d\delta{q} \wedge dp.
\]
For an unconstrained Lagrangian system with Lagrangian $L: TQ \to \mathbb{R}$, define a subset of the symplectic manifold $(TT^{\ast}Q,\Omega_{TT^{\ast}Q}=-\mathbf{d}\lambda)$ by
\begin{align}\label{LagSub_Lag}
N=\{ x \in TT^{\ast}Q \, \mid \, 
\lambda_{x}(w)=\left< \mathbf{d}L(T\pi_{Q}(x)), T_{x}T\pi_{Q}(w) \right>, \; \forall w \in T_{x}(TT^{\ast}Q) \},
\end{align}
which is a Lagrangian submanifold whose dimension is  $\frac{1}{2}\mathrm{dim\,}TT^{\ast}Q$. Hence, the Lagrangian $L$ is a {\it generating function} of $N$, since $N \subset TT^{\ast}Q$ is the graph of $(\kappa_{Q})^{-1}(\mathbf{d}L)$. In fact, it follows from \eqref{LagSub_Lag} that, for $x=(q,p,\dot{q},\dot{p}) \in N \subset TT^{\ast}Q$, 
\begin{align}\label{Variational_Relation_LagSub_Lag}
\dot{p}\delta q+p\delta\dot{q}=\frac{\partial L}{\partial q}\delta q+\frac{\partial L}{\partial \dot{q}}\delta\dot{q},
\end{align}
which satisfies for arbitrary $\delta{q}, \delta{\dot{q}}$ and hence we get
\[
p= \frac{\partial L}{\partial \dot{q}}, \quad \dot{p}= \frac{\partial L}{\partial q}.
\]
These equations are, of course, equivalent to the Euler--Lagrange equations. In other words, the diffeomorphism $\kappa_{Q}: (q, p, \dot q, \dot p) \mapsto (q, \dot q, \dot p, p)$ was introduced such that the relation \eqref{Variational_Relation_LagSub_Lag} holds.

 \paragraph{Canonical transformation and generating functions.}
 Consider a map $\tilde\varphi: T^{\ast}Q \to T^{\ast}Q$ on the cotangent bundle with  graph 
 \[
 \Gamma(\tilde\varphi) \subset T^{\ast}Q \times T^{\ast}Q
 \]  
 and introduce an inclusion map $\iota_{\tilde\varphi}: \Gamma(\tilde\varphi)  \to T^{\ast}Q \times T^{\ast}Q$. Define 
 a one-form on $T^{\ast}Q \times T^{\ast}Q$ as
 \[
 \hat\Theta=\tau_{(T^{\ast}Q)_{2}}^{\ast}\Theta_{T^{\ast}Q}-\tau_{(T^{\ast}Q)_{1}}^{\ast}\Theta_{T^{\ast}Q},
 \]
 where $\tau_{(T^{\ast}Q)_{i}}: T^{\ast}Q \times T^{\ast}Q \to T^{\ast}Q$ $(i=1,2)$ denote the projection of $T^{\ast}Q \times T^{\ast}Q$ onto the $i$-th manifold. Then we define
 \[
 \hat\Omega=-\mathbf{d}\hat\Theta=\tau_{(T^{\ast}Q)_{2}}^{\ast}\Omega_{T^{\ast}Q}-\tau_{(T^{\ast}Q)_{1}}^{\ast}\Omega_{T^{\ast}Q}.
 \]
 The map is canonical if and only if
 \begin{equation*}
\iota_{\tilde\varphi}^{\ast}\hat\Omega=0.
\end{equation*}
This is because
\begin{equation*}
\begin{split}
\iota_{\tilde\varphi}^{\ast}\hat\Omega&=\iota_{\tilde\varphi}^{\ast}(-\tau_{(T^{\ast}Q)_{1}}^{\ast}\Omega_{T^{\ast}Q}+\tau_{(T^{\ast}Q)_{2}}^{\ast}\Omega_{T^{\ast}Q})\\
&=-(\tau_{(T^{\ast}Q)_{1}}\circ \iota_{\tilde\varphi})^{\ast}\Omega_{T^{\ast}Q}+(\tau_{(T^{\ast}Q)_{2}}\circ \iota_{\tilde\varphi})^{\ast}\Omega_{T^{\ast}Q}\\
&=-(\tau_{(T^{\ast}Q)_{1}}|_{\Gamma(\tilde\varphi)})^{\ast}\Omega_{T^{\ast}Q}+(\tilde\varphi \circ (\tau_{(T^{\ast}Q)_{1}}|_{\Gamma(\tilde\varphi)}))^{\ast}\Omega_{T^{\ast}Q}\\
&=(\tau_{(T^{\ast}Q)_{1}}|_{\Gamma(\tilde\varphi)})^{\ast}(-\Omega_{T^{\ast}Q}+\tilde\varphi^{\ast}\Omega_{T^{\ast}Q}),
\end{split}
\end{equation*}
 where 
 \begin{equation*}
\tau_{(T^{\ast}Q)_{1}} \circ \iota_{\tilde\varphi}=\tau_{(T^{\ast}Q)_{1}}|_{\Gamma(\tilde\varphi)},\quad \textrm{and} \quad \tau_{(T^{\ast}Q)_{2}}\circ \iota_{\tilde\varphi}= \tilde\varphi \circ \tau_{(T^{\ast}Q)_{1}} \;\; \textrm{on}\;\; \Gamma(\tilde\varphi).
\end{equation*}
Hence, by the Poincar\'e Lemma, if the map is canonical, then there exists a local function $S$ on $\Gamma(\tilde\varphi) \subset T^{\ast}Q\times T^{\ast}Q$ such that
\begin{equation*}
\iota_{\tilde\varphi}^{\ast}\hat\Theta=\mathbf{d}S.
\end{equation*}

\paragraph{Construction of the discrete map $\kappa_{Q}^{d}: T^{\ast}Q \times T^{\ast}Q \to T^{\ast}(Q\times Q)$.} 
Here we consider a discrete analogue of the diffeomorphism $\kappa_{Q}: TT^{\ast}Q \to T^{\ast}TQ$, namely, a map  $\kappa_{Q}^{d}: T^{\ast}Q \times T^{\ast}Q \to T^{\ast}(Q \times Q)$ as shown below.

Let $(q_{0},p_{0},q_{1},p_{1})$ be the local coordinates of $T^{\ast}Q \times T^{\ast}Q$ and, we can appropriately choose any two of four  quantities $(q_{0},p_{0},q_{1},p_{1})$ for the local coordinates of the submanifold $\Gamma(\tilde\varphi) \subset T^{\ast}Q \times T^{\ast}Q$ if $Q$ is a vector space as in \cite{GoPoSa2000,LeOh2011}.

Since we are interested in the case where $Q$ is a manifold, the proper choice for the local coordinates is $(q_{0},q_{1}) \in Q \times Q$ for $\Gamma(\tilde\varphi)$. Then we shall show how the discrete Lagrangian $L_{d}$ on $Q\times Q$ can be the generating function $S$ for the canonical transformation; see also \cite{MaWe2001}.  

As was shown in \S\ref{Subsec:DisHamPrin}, associated with the {\it discrete Euler--Lagrange equations}  \eqref{DiscreteEulerLagEq}, i.e., 
\[
\frac{\partial L_d}{\partial q_{1}}(q_{0}, q_1)=- \frac{\partial L_d}{\partial q_{1}}(q_1, q_{2}),
\]
and under the regularity assumption of $L_{d}$, there exists the discrete Lagrange map 
\[
\varphi_{L_{d}}: Q \times Q \to Q \times Q; \quad (q_{0},q_{1}) \mapsto (q_{1}, q_{2}).
\]
This map induces the canonical transformation called the discrete Hamiltonian map $\tilde{\varphi}_{L_{d}}$
as
\[
\tilde{\varphi}_{L_{d}}=\mathbb{F}^{+}L_{d}\circ (\mathbb{F}^{-}L_{d})^{-1}: T^{\ast}Q \to T^{\ast}Q,
\]
which is given in local coordinates by 
\begin{equation}\label{CanonicalTransformation}
\begin{split}
\left(q_{0},p_{0}= -\frac{\partial L}{\partial q_{0}}(q_{0},q_{1})\right)\mapsto 
\left(q_{1}, p_{1}=\frac{\partial L}{\partial q_{1}}(q_{0},q_{1})\right).
\end{split}
\end{equation}
In the below, we shall illustrate how the canonical transformation 
(discrete Hamiltonian map) $\tilde{\varphi}_{L_{d}}: T^{\ast}Q \to T^{\ast}Q$ can be generated by the discrete Lagrangian $L_{d}$ on $Q\times Q$. 
\medskip

Recall that the one-form $\hat\Theta$ on $T^{\ast}Q \times T^{\ast}Q$  is naturally defined by
\begin{equation*}
\hat\Theta=\tau_{(T^{\ast}Q)_{2}}^{\ast}\Theta_{T^{\ast}Q}-\tau_{(T^{\ast}Q)_{1}}^{\ast}\Theta_{T^{\ast}Q}=p_{1}dq_{1}-p_{0}dq_{0}.
\end{equation*}
Let $\iota_{\tilde\varphi_{L_{d}}}: Q\times Q \subset \Gamma(\tilde\varphi_{L_{d}})  \to T^{\ast}Q \times T^{\ast}Q$ be the natural inclusion map given by
\[
(q_{0},q_{1}) \mapsto \left( (q_{0},p_{0}), (q_{1},p_{1}) \right) = \left( (q_{0},p_{0}), \tilde{\varphi}_{L_{d}}(q_{0},p_{0}) \right).
\]
Then the pullback of $\hat\Omega=-\mathbf{d}\hat\Theta$ by $\iota_{\tilde\varphi_{L_{d}}}$ vanishes, namely,
\begin{equation}\label{CondLagSubmanifold}
\iota_{\tilde\varphi_{L_{d}}}^{\ast}\hat\Omega=\iota_{\tilde\varphi_{L_{d}}}^{\ast}(-\mathbf{d}\hat\Theta)=0,
\end{equation}
which leads to the preservation of the canonical symplectic structure
\[
\tilde{\varphi}_{L_{d}}^{\ast}\Omega_{T^{\ast}Q}=\Omega_{T^{\ast}Q},
\]
where $\Omega_{T^{\ast}Q}(q_{0},p_{0})=dq_{0}\wedge dp_{0}$.
Therefore, it follows from \eqref{CondLagSubmanifold} that there exists a local function $L_{d}=L_{d}(q_{0},q_{1})$ on $Q \times Q \subset \Gamma(\tilde\varphi_{L_{d}})$ by the Poincar\'e Lemma, and the local function is exactly the discrete Lagrangian such that
\[
\iota_{\tilde\varphi_{L_{d}}}^{\ast}\hat\Theta=\mathbf{d}L_{d},
\]
where 
\begin{equation*}
\mathbf{d}L_{d}= \frac{\partial L_{d}}{\partial q_{0}}dq_{0} + \frac{\partial L_{d}}{\partial q_{1}}dq_{1}.
\end{equation*}
Thus the canonical transformation $\tilde{\varphi}_{L_{d}}: T^{\ast}Q \to T^{\ast}Q$  in \eqref{CanonicalTransformation} is generated by the discrete Lagrangian $L_{d}=L_{d}(q_{0},q_{1})$ on $Q \times Q$. From this construction, we can uniquely define the discrete map $\kappa_{Q}^{d}: T^{\ast}Q \times T^{\ast}Q \to T^{\ast}(Q \times Q)$ as follows. 

\begin{definition}\rm
Define the discrete map $\kappa_{Q}^{d}: T^{\ast}Q \times T^{\ast}Q \to T^{\ast}(Q \times Q)$ such that
\begin{equation*}
\begin{split}
\kappa_{Q}^{d} \circ \iota_{\tilde\varphi_{ L_{d}}}=\mathbf{d}L_{d}
\end{split}
\end{equation*}
and the map $\kappa_{Q}^{d}$ is locally denoted as
\begin{equation}\label{kappa_Q^d}
\kappa_{Q}^{d}: \left((q_{0},p_{0}),(q_{1},p_{1})\right) \mapsto \left((q_{0},q_{1}),(-p_{0}, p_{1})\right).
\end{equation}
\end{definition}

\begin{remark}\rm
In the construction of the discrete map $\kappa_{Q}^{d}$ in \eqref{kappa_Q^d}, we  considered the case in which there exist no constraints and therefore the canonical symplectic structure $\Omega_{T^{\ast}Q}$ is preserved under the discrete Hamiltonian map $\tilde{\varphi}_{L_{d}}$. The discrete map $\kappa_{Q}^{d}$ is uniquely determined so that it is consistently analogous with Tulczyjew's triple in the continuous case reviewed in \S\ref{Sect:LagDiracDySys}. Furthermore, we shall use the discrete map $\kappa_{Q}^{d}$ even for the case when  nonholonomic constraints exist and hence the discrete Lagrangian two-form $\Omega_{L_d}$ is no longer preserved. Note also that the same discrete map $\kappa_{Q}^{d}$ was used in \cite{LeOh2011}.
\end{remark}

\paragraph{Discretization of the structure between iterated  bundles.}
We have seen the construction of the discrete map $\kappa_{Q}^{d}: T^{\ast}Q \times T^{\ast}Q \to T^{\ast}(Q \times Q)$.
 As already shown in Def. \ref{def:DisCanSymFlatMap}, the bundle map $\Omega_{T^{\ast}Q}^{\flat}: TT^{\ast}Q \to T^{\ast}T^{\ast}Q$ is discretized, in view of $TT^{\ast}Q \cong T^{\ast}Q \times T^{\ast}Q$, to be the $(\pm)$-discrete canonical flat maps $(\Omega_{T^{\ast}Q}^{d\pm})^{\flat}: T^{\ast}Q \times T^{\ast}Q \to T^{\ast}T^{\ast}Q$. Now we introduce the $(\pm)$-discretizations of diffeomorphism $\gamma_{Q}: T^{\ast}TQ \to T^{\ast}T^{\ast}Q$ by using the discrete maps $\kappa_{Q}^{d}$ and $(\Omega_{T^{\ast}Q}^{d\pm})^{\flat}$.

\begin{definition}\rm
The $(\pm)$-discrete analogues of $\gamma_{Q}: T^{\ast}TQ \to T^{\ast}T^{\ast}Q$, i.e., the $(\pm)$-discrete maps $\gamma^{d\pm}_{Q}: T^{\ast}(Q \times Q) \to T^{\ast}T^{\ast}Q$, are defined by using the $(\pm)$-discrete canonical flat maps $({\Omega}^{d\pm}_{T^{\ast}Q})^{\flat}: T^{\ast}Q \times T^{\ast}Q \to T^{\ast}T^{\ast}Q$ given by \eqref{disceteOmegaFlat_plus} and \eqref{disceteOmegaFlat_minus} and the discrete map $\kappa_Q^d: T^{\ast}Q \times T^{\ast}Q \to T^{\ast}(Q \times Q)$ as
\begin{equation}\label{DiscreteGammaMap}
\begin{split}
\gamma_{Q}^{d+}:=({\Omega}^{d+}_{T^{\ast}Q})^{\flat} \circ (\kappa_{Q}^{d})^{-1}:  T^{\ast}(Q \times Q) &\to T^{\ast}T^{\ast}Q\\ 
((q_{0},q_1),(-p_{0},p_{1}))&\mapsto ((q_1,p_1),(-\hat{p}^{+}_{1},\hat{q}^{-}_{1})),\\
\gamma_{Q}^{d-}:=({\Omega}^{d-}_{T^{\ast}Q})^{\flat} \circ (\kappa_{Q}^{d})^{-1}:  T^{\ast}(Q \times Q) &\to T^{\ast}T^{\ast}Q\\
((q_0,q_1),(-p_0,p_1))&\mapsto ((q_0, p_0),(-\hat{p}^{-}_0,\hat{q}^{+}_0)).
\end{split}
\end{equation}
\end{definition}

\paragraph{Discrete projections.} In order to define the discrete bundle structures, we can naturally define the following  \textit{canonical discrete projections} by
\begin{equation*}
\begin{split}
&\pi_{Q} \times \pi_{Q}: T^{\ast}Q \times T^{\ast}Q \to Q \times Q; \quad ((q_{0},p_{0}),(q_{1},p_{1})) \mapsto (q_{0},q_{1}),\\
&\pi_{Q \times Q}: T^{\ast}(Q \times Q) \to Q \times Q; \quad ((q_{0},q_{1}),(-p_{0}, p_{1})) \mapsto (q_{0},q_{1}).\\
\end{split}
\end{equation*}

\begin{definition}\rm
We can further introduce canonical discrete projections as follows:
\begin{equation*}
\begin{split}
&\tau_{Q_{2}}: Q \times Q \to Q;\quad (q_{0},q_{1}) \mapsto q_{1},\\
&\tau_{Q_{1}}: Q \times Q \to Q;\quad (q_{0},q_{1}) \mapsto q_{0},\\
&\tau_{(T^{\ast}Q)_{2}}: T^{\ast}Q \times T^{\ast}Q \to T^{\ast}Q;\quad ((q_{0},p_{0}),(q_{1},p_{1})) \mapsto (q_{1},p_{1}),\\
&\tau_{(T^{\ast}Q)_{1}}: T^{\ast}Q \times T^{\ast}Q \to T^{\ast}Q;\quad ((q_{0},p_{0}),(q_{1},p_{1})) \mapsto (q_{0},p_{0}).\\
\end{split}
\end{equation*}
\end{definition}
\paragraph{Discrete bundle pictures.}
Before going into the details on the construction of $(\pm)$-discrete Lagrange--Dirac dynamical systems, we summarize the $(\pm)$-discrete bundle pictures as illustrated in Figs. \ref{plus_DisBundleGeometry} and \ref{minus_DisBundleGeometry}. 

\begin{figure}[htbp]
\centering
\includegraphics[scale=0.55]{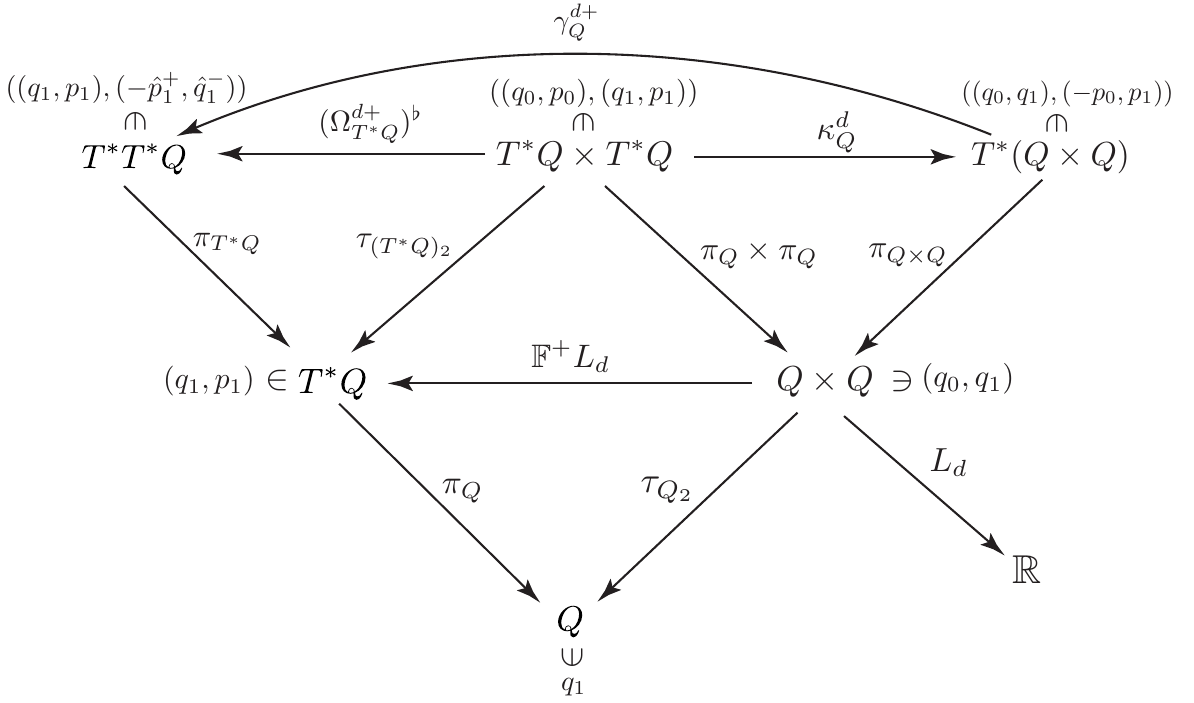}
\caption{$(+)$-discrete bundle picture.}
\label{plus_DisBundleGeometry}
\end{figure}
\vspace{5mm}
\begin{figure}[htbp]
\centering
\includegraphics[scale=0.55]{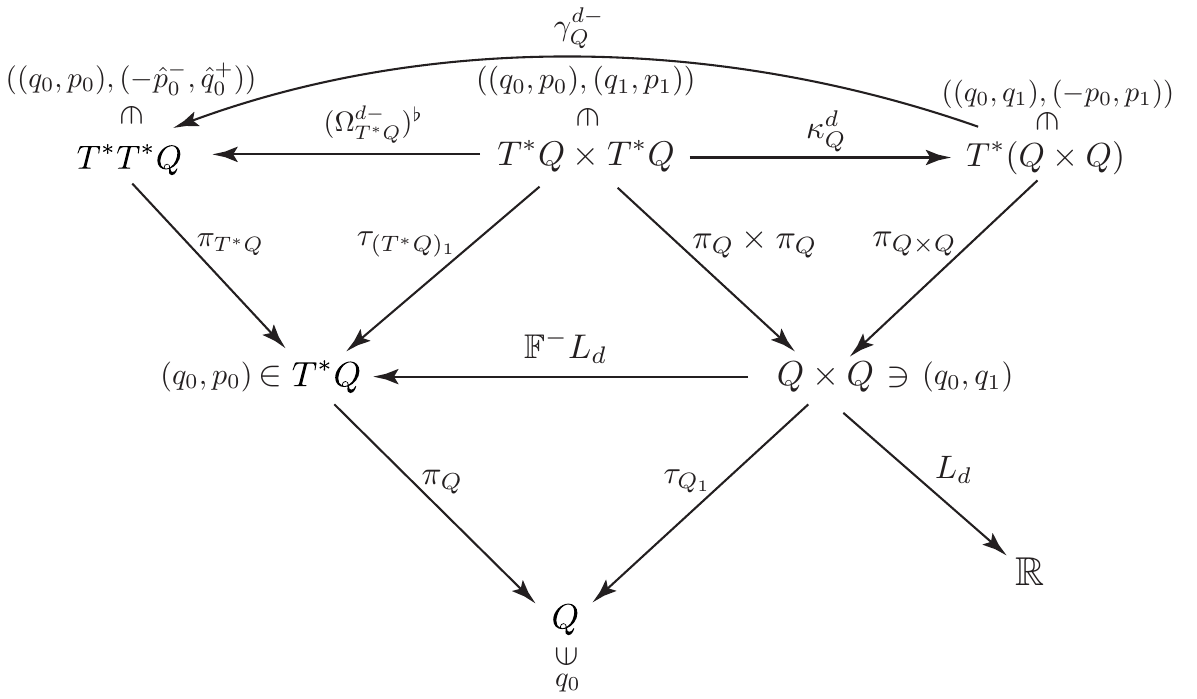}
\caption{$(-)$-discrete bundle picture.}
\label{minus_DisBundleGeometry}
\end{figure}

\subsection{Discrete Lagrange--Dirac dynamical systems}
Here, we shall develop $(\pm)$-discrete Lagrange--Dirac dynamical systems within the context of the $(\pm)$-discrete bundle pictures, associated with the continuous Lagrange--Dirac dynamical system described in \S\ref{Sect:LagDiracDySys}. In Fig. \ref{BundlePic1_discrete_coordinates_plus} and Fig. \ref{BundlePic1_discrete_coordinates_minus}, we illustrate that the $(\pm)$-discrete Lagrange--Dirac dynamical systems can be incorporated into the $(\pm)$-discrete bundle structures shown in Fig. \ref{plus_DisBundleGeometry} and Fig. \ref{minus_DisBundleGeometry}

\paragraph{Discrete Legendre transforms.}
For a given discrete Lagrangian $L_d:Q \times Q \to \mathbb{R}$, recalling from \eqref{eq:disLeg}, the $(\pm)$-discrete Legendre transforms $\mathbb{F}^{\pm}L_{d}: Q\times Q \to T^{\ast}Q$ are given by
\begin{equation*}
\begin{split} 
&\mathbb{F} ^+L_d:  (q_0 , v _{0}) \in Q \times Q \mapsto (v_0, p_{{1}})= (v_0, D_2L_d( q_0 , v _{0} )) \in T^*_{q _{1} }Q,\\
&\mathbb{F} ^-L_d: (v_1 , q _{1})  \in Q \times Q \mapsto (v_1, p_{0})= (v_1, -D_1L_d(v_1, q_{1} )) \in T^*_{q _0}Q,
\end{split}
\end{equation*} 
where we note that the local coordinates $(q_0 , v _{0}) \in Q \times Q$ are employed for the $(+)$-formula and $(v_1 , q _{1}) \in Q \times Q$ for the $(-)$-formula. Note that the base point of $\mathbb{F} ^+L_d(q_0 , v _{0})$ is given by $v_0=q_{1} \in Q$, while  the base point of $\mathbb{F} ^-L_d(v_1 , q _{1})$ is given by $v_1=q_{0} \in Q$.

\paragraph{Discrete Dirac differentials.}
Next we shall consider the discretization of the Dirac differential of the discrete Lagrangian. In the continuous setting, recall that the Dirac differential $\mathbf{d}_{D}L$ for a given Lagrangian $L:TQ \to \mathbb{R}$, namely, $\mathbf{d}_{D}L: TQ \to T^\ast T^\ast Q$ is the map defined by
\[
\begin{split}
\mathbf{d}_{D}L=\gamma_{Q} \circ \mathbf{d}L: TQ \to T^{\ast}T^{\ast}Q;\; (q,v) \mapsto \left(q,\frac{\partial L}{\partial v},-\frac{\partial L}{\partial q},v\right).
\end{split}
\]
In the discrete setting, corresponding to the $(\pm)$-discrete maps $\gamma_{Q}^{d\pm}$, there exist the $(\pm)$-discrete Dirac differentials of $L_d:Q \times Q \to \mathbb{R}$ as follows.

\begin{definition}\rm
The $(\pm)$-discrete Dirac differentials of the discrete Lagrangian $L_{d}:Q\times Q\to \mathbb{R}$ are the maps that are defined by the composition of the discrete maps $\gamma_{Q}^{d\pm}: T^{\ast}(Q\times Q) \to T^{\ast}T^{\ast}Q$ in \eqref{DiscreteGammaMap} and $\mathbf{d}L_{d}: Q \times Q \to T^{\ast}(Q \times Q)$ as
\begin{equation*}
\mathbf{d}^{\pm}_{D}L_d:=\gamma_{Q}^{d\pm}\circ \mathbf{d}L_d:Q\times Q\rightarrow T^{\ast}T^{\ast}Q.
\end{equation*}

The local expression of the {\it $(+)$-discrete Dirac differential} of $L_{d}$ is expressed in local  coordinates $(q_{0},v_{0}) \in Q \times Q$ as
\begin{equation*}
\begin{split}
&\mathbf{d}^{+}_{D}L_d: (q_0,v_{0})\mapsto \left(v_{0},D_2 L_d(q_0,v_{0}), -\widehat{D_{2}L}_{d}^{+}(q_0,v_{0}),\hat{v}^{-}_{0}\right),
\end{split}
\end{equation*}
where $\widehat{D_2L}_d^{+}(q_{0},v_{0})$ and $\hat{v}^-_0$ are respectively approximations of ${D_{2}L_{d}}(q_{0},v_{0})$ and $v_0$, which are given in association with the map $\gamma_{Q}^{d+}$ as
\begin{equation}\label{hat_rules_dL+}
\widehat{D_2L}_d^{+}(q_0,v_{0})=\frac{1}{h}\bigl(D_2L_d(q_{1},v_{1})+D_1L_d(q_{1},v_{1})\bigr), 
\quad
\hat{v}^-_0=\frac{v_0-q_{0}}{h}.
\end{equation}

On the other hand, the local expression of the {\it $(-)$-discrete Dirac differential} of $L_{d}$  is denoted  in local  coordinates $(v_{1},q_{1}) \in Q \times Q$ by
\begin{equation*}
\begin{split}
& \mathbf{d}^{-}_{D}L_d:(v_{1},q_{1})\mapsto \left(v_{1}, -D_1L_d(v_{1},q_{1}) , \widehat{D_{1}L}_{d}^{-}(v_{1},q_{1}),\hat{v}^{+}_{1}\right),
\end{split}
\end{equation*}
where $\widehat{D_{1}L}_{d}^{-}(v_{1},q_{1})$ and $\hat{v}^{+}_{1}$ are respectively approximations of ${D_{1}L_{d}}(v_{1},q_{1})$ and $v_1$, which are given in association with  the map $\gamma_{Q}^{d-}$ as
\begin{equation}\label{hat_rules_dL-}
\widehat{D_1L}_d^{-}(v_{1},q_{1})=-\frac{1}{h}\bigl(D_2L_d(v_{0},q_{0})+D_1L_d(v_{0},q_{0})\bigr), 
\quad
\hat{v}^+_1=\frac{q_1-v_{1}}{h}.
\end{equation}
\end{definition}

\paragraph{Discrete evolution operators.}
Next we will introduce $(\pm)$-discrete evolution operators, which serve analogous roles to continuous vector fields.

\begin{definition}$\;$\rm
Define the $(\pm)$-{\it discrete evolution operators} 
\begin{equation*}
{X}^{d\pm}: T^{\ast}Q \to T^{\ast}Q \times T^{\ast}Q
 \end{equation*}
such that
\[
\tau_{(T^{\ast}Q)_{2}}\circ X^{d+}=\mathrm{Id} \quad \textrm{and}\quad
\tau_{(T^{\ast}Q)_{1}}\circ X^{d-}=\mathrm{Id},
\]
which are locally given by
\begin{equation}\label{LocalDisEvoOpe}
\begin{split}
&{X}^{d +}: (q_{1},p_{1}) \mapsto (q_{0},p_{0},q_{1},p_{1}),\\
&{X}^{d -}: (q_{0},p_{0}) \mapsto (q_{0},p_{0}, q_{1},p_{1}).
\end{split}
\end{equation}
Corresponding to the $(\pm)$-discrete evolution operators $X^{d\pm}$, the $(\pm)$-{\it discrete flow maps} are defined by
\[
\begin{split}
&(\tilde{\varphi})^{-1}: =\tau_{(T^{\ast}Q)_1}\circ X^{d+}: T^{\ast}Q \to T^{\ast}Q;\quad (q_{1},p_{1}) \mapsto (q_{0},p_{0}),\\
&\tilde{\varphi}: =\tau_{(T^{\ast}Q)_2}\circ X^{d-}: T^{\ast}Q \to T^{\ast}Q;\quad  (q_{0},p_{0}) \mapsto (q_{1},p_{1}).
\end{split}
\]
Associated with the $(+)$-discrete evolution operator $X^{d+}$ in \eqref{LocalDisEvoOpe}, the approximation of $(\dot{q}_{1},\dot{p}_{1})$ is given by using $(+)$-discretization  map $\Psi_{(T^{\ast}Q)_{2}}^{+}: T^{\ast}Q \times T^{\ast}Q \to TT^{\ast}Q$ as
\begin{equation}\label{disc_plus_VectorField}
\begin{split}
\hat{X}^{d+}(q_1,p_1)&=\Psi_{(T^*Q)_2}^{+}(X^{d+}(q_1,p_1))\\
& =\Psi_{(T^*Q)_2}^{+}(q_{0},p_{0},q_{1},p_{1})\\
&=\left(q_1,p_1,\hat{q}^{-}_{1}=\frac{q_1-q_{0}}{h},\hat{p}^{+}_{1}=\frac{p_2-p_{1}}{h}\right)\in T_{(q_{1},p_{1})}T^{\ast}Q.
\end{split}
\end{equation}

Similarly, associated with the $(-)$-discrete evolution operator $X^{d-}$ in \eqref{LocalDisEvoOpe}, the approximation of $(\dot{q}_{0},\dot{p}_{0})$ is given by using $(-)$-discretization  map $\Psi_{(T^{\ast}Q)_{1}}^{-}: T^{\ast}Q \times T^{\ast}Q \to TT^{\ast}Q$ as
\begin{equation}\label{disc_minus_VectorField}
\begin{split}
\hat{X}^{d-}(q_0,p_0)&=\Psi_{(T^*Q)_1}^{-}(X^{d-}(q_0,p_0))\\
& =\Psi_{(T^*Q)_1}^{-}(q_{0},p_{0},q_{1},p_{1})\\
&=\left(q_0,p_0,\hat{q}^{+}_{0}=\frac{q_1-q_{0}}{h},\hat{p}^{-}_{0}=\frac{p_0-p_{-1}}{h}\right)\in T_{(q_{0},p_{0})}T^{\ast}Q.
\end{split}
\end{equation}

\end{definition}
%

%

\begin{figure}[htbp]
\centering
\includegraphics[scale=0.55]{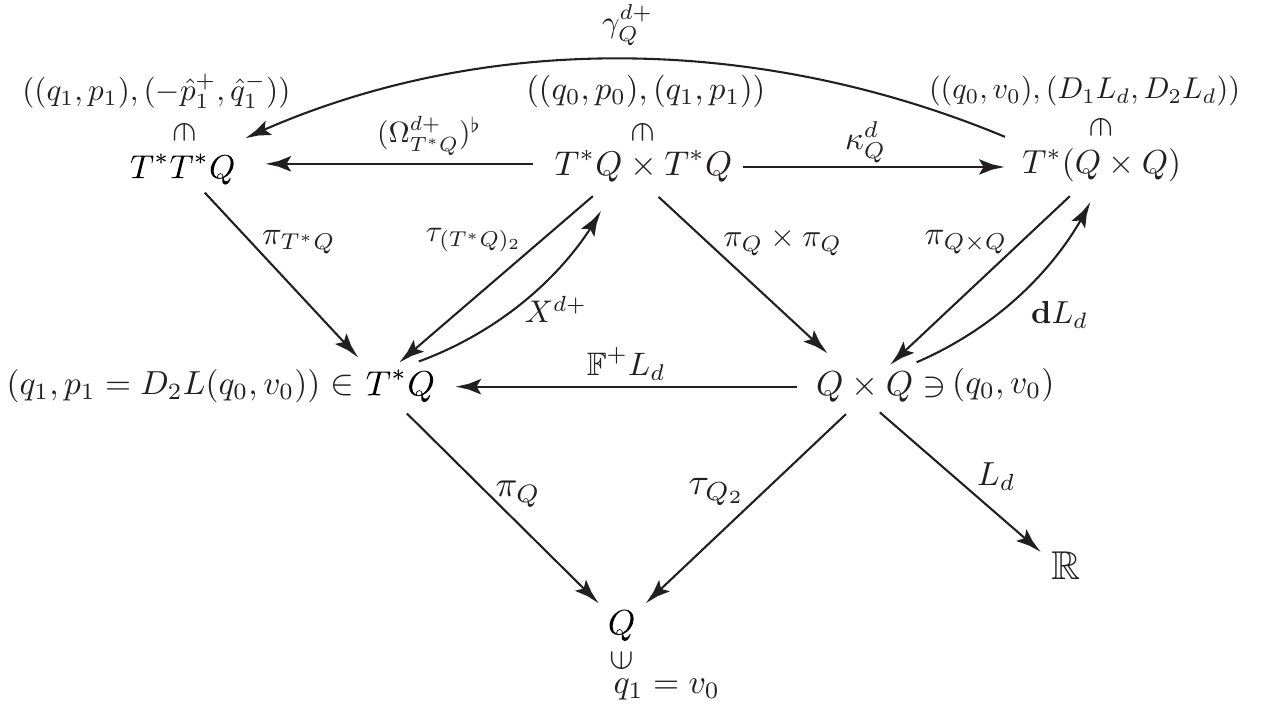}
\caption{Geometric structure of $(+)$-discrete Lagrange--Dirac dynamical systems.}
\label{BundlePic1_discrete_coordinates_plus}
\end{figure}
\vspace{5mm}
\begin{figure}[htbp]
\centering
\includegraphics[scale=0.55]{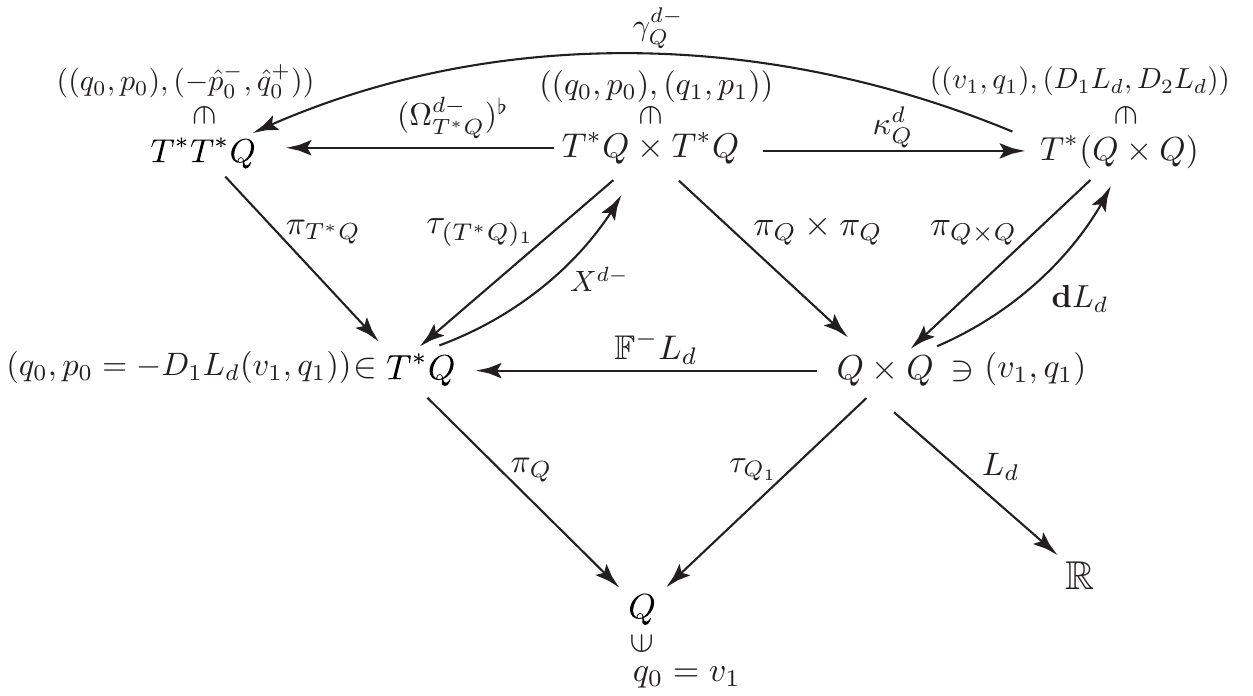}
\caption{Geometric structure of $(-)$-discrete Lagrange--Dirac dynamical systems.}
\label{BundlePic1_discrete_coordinates_minus}
\end{figure}

\paragraph{Discrete Lagrange--Dirac dynamical systems.}
By analogy with the continuous setting of Lagrange--Dirac dynamical systems given in Def.\,\ref{ContLagDiracDynSys},  we shall define $(\pm)$-discrete Lagrange--Dirac dynamical systems as follows.

\begin{definition}\rm Let $L_{d}: Q\times Q \to \mathbb{R}$ be a discrete Lagrangian, let $X^{d\pm}: T^{\ast}Q \to T^{\ast}Q \times T^{\ast}Q$ be  discrete evolution operators and let $D^{d\pm}_{\Delta_{Q}}$ be $(\pm)$-discrete induced Dirac structures.

The $(+)$-{\it discrete Lagrange--Dirac dynamical system} is defined by the triple $(X^{d+}, \mathbf{d}^{+}_{D}L_{d},$ $ D^{d+}_{\Delta_{Q}})$ that satisfies the condition, for each $((q_{0},v_{0}),(q_1,p_1)) \in (Q\times Q)\oplus T^{\ast}Q$,
\begin{equation}
\label{EqDiscLagDiracDynSys_plus}
\left(X^{d+}(q_{1},p_{1}),\mathbf{d}^{+}_{D}L_{d}(q_{0},v_{0})\right)\in D^{d+}_{\Delta_{Q}}(q_{1},p_{1}).
\end{equation}
Similarly, 
the $(-)$-{\it  discrete Lagrange--Dirac dynamical system} is defined by the triple $(X^{d-}, \mathbf{d}^{-}_{D}L_{d}, D^{d-}_{\Delta_{Q}})$ that satisfies the condition, for each $((v_{1},q_{1}),(q_0,p_{0})) \in (Q\times Q)\oplus T^{\ast}Q$,
\begin{equation}\label{EqDiscLagDiracDynSys_minus}
\left(X^{d-}(q_{0},p_{0}),\mathbf{d}^{-}_{D}L_{d}(v_{1},q_{1})\right)\in D^{d-}_{\Delta_{Q}}(q_{0},p_{0}).
\end{equation}
\end{definition}
Then, we have the following theorem for the discrete Lagrange--Dirac dynamical systems.

\begin{framed}
\item \begin{theorem}\rm\label{Thm:DiscLagDiracDynSys}
Consider a discrete path $(q_{d}, v_{d}, p_{d})$ on $(Q \times Q) \oplus T^{\ast}Q$, where $q_{d}\!=\!\{q_{k}\}_{k=0}^{N}$, \!
$v_{d}\!=\!\{v_{k}\}_{k=0}^{N}$\! and \!$p_{d}\!=\!\{p_{k}\}_{k=0}^{N}$ with $q_{k}\!=\!q_{d}(t_{k})$, $v_{k}\!=\!v_{d}(t_{k})$, and $p_{k}\!=\!p_{d}(t_{k})$.  
%
\begin{itemize}
\item[(i)]
If a discrete path $(q_{d}, v_{d}, p_{d}) \in \mathcal{C}((Q \times Q) \oplus T^{\ast}Q)$ is the solution of the \textit{$(+)$-discrete Lagrange--d'Alembert--Pontryagin equations}:
\begin{equation}\label{plus_DiscLagDiracSys}
\begin{split}
&p_{k+1}=\frac{\partial L_{d}}{\partial v_{k}}(q_{k},v_{k}),\;\; p_{k+1}+ \frac{\partial L_{d}}{\partial q_{k+1}}(q_{k+1},v_{k+1}) \in \Delta^{\circ}_Q(q_{k+1}),\\
&q_{k+1}=v_{k},\;\
(q_{k},v_{k})\in \Delta_Q^{d+},
\end{split}
\end{equation}
then the discrete path $(q_{d}, v_{d}, p_{d})$ satisfies the condition for the {\it $(+)$-discrete Lagrange--Dirac dynamical system} as
\[
(X^{d+}(q_{k+1},p_{k+1}) ,\mathbf{d}_{D}^{+}L_{d}(q_{k},v_{k}))\in D^{d+}_{\Delta_{Q}}(q_{k+1},p_{k+1}),\quad k=1,...,N-1.
\]
The $(+)$-discrete Lagrange--d'Alembert--Pontryagin equations recover the $(+)$-discrete Lagrange--d'Alembert equations:
\[
D_{2}L_{d}(q_{k},q_{k+1}) + D_{1}L_{d}(q_{k+1},q_{k+2}) \in \Delta^{\circ}_{Q}(q_{k+1}),\quad 
(q_{k},q_{k+1}) \in \Delta_{Q}^{d+}.
\]
\item[(ii)]
If a discrete path $(q_{d}, v_{d}, p_{d}) \in \mathcal{C}((Q \times Q) \oplus T^{\ast}Q)$ is the solution of the \textit{$(-)$-discrete Lagrange--d'Alembert--Pontryagin equations}:
\begin{equation}\label{minus_DiscLagDiracSys}
\begin{split}
&p_{k}=-\frac{\partial L_{d}}{\partial v_{k+1}}(v_{k+1},q_{k+1}),\;\; p_{k}- \frac{\partial L_{d}}{\partial q_{k}}(v_{k},q_{k}) \in \Delta^{\circ}_Q(q_{k}),\\
&q_k=v_{k+1},\;\
(v_{k+1},q_{k+1})\in \Delta_Q^{d-},
\end{split}
\end{equation}
then the discrete path $(q_{d}, v_{d}, p_{d})$ satisfies the condition for the {\it $(-)$-discrete Lagrange--Dirac dynamical system} as
\[
(X^{d-}(q_{k},p_{k}) ,\mathbf{d}^{-}_{D}L_{d}(v_{k+1},q_{k+1}))\in D^{d-}_{\Delta_{Q}}(q_{k},p_{k}),\quad k=1,...,N-1.
\]
The $(-)$-discrete Lagrange--d'Alembert--Pontryagin equations recover the $(-)$-discrete Lagrange--d'Alembert equations:
\[
D_{2}L_{d}(q_{k-1},q_{k}) + D_{1}L_{d}(q_{k},q_{k+1}) \in \Delta^{\circ}_{Q}(q_{k}),\quad 
(q_{k},q_{k+1}) \in \Delta_{Q}^{d-}.
\]
\end{itemize}
\end{theorem}
\end{framed}
\begin{proof}
\begin{itemize}
\item[(i)] For the $(+)$ part, setting  
\[
\begin{split}
&X^{d+}(q_{1},p_{1})=(q_{0},p_{0},q_{1},p_{1}),\quad \hat X^{d+}(q_{1},p_{1})=(q_{1},p_{1},\hat{q}^-_{1},\hat{p}^+_{1}),
\end{split}
\]
and
\[
\begin{split}
\mathbf{d}^{+}_{D}L_{d}(q_{0},v_{0})&=\left(v_{0},D_2 L_d(q_0,v_{0}), -\widehat{D_{2}L}_{d}^{+}(q_0,v_{0}),\hat{v}^{-}_{0}\right)\\
&=\left(\left(v_0,\frac{\partial L_d}{\partial v_0}(q_{0},v_{0})\right),\left(-\frac{\widehat{\partial L}^{+}_d}{\partial v_0}(q_{0},v_{0}), \hat{v}^{-}_0 \right) \right),
\end{split}
\]
it follows from \eqref{Loc_discrete_Dirac_plus} that the condition \eqref{EqDiscLagDiracDynSys_plus} becomes
\begin{equation*}
(q_{0},v_0)\in \Delta^{d+}_Q,\quad -\frac{\widehat{\partial L}^{+}_d}{\partial v_0}(q_{0},v_{0})
+\hat p^{+}_{1} \in \Delta^{\circ}_{Q}(q_{1}), \quad \hat{q}^{-}_{1}=\hat{v}^{-}_0
\end{equation*}
together with the $(+)$-Legendre transform $p_{1}=\frac{\partial L_{d}}{\partial v_{0}}(q_0,v_0)$.
Then it follows from \eqref{hat_rules_dL+} and \eqref{disc_plus_VectorField} that 
\[
\begin{split}
&\hat{q}^-_1=\frac{q_1-q_0}{h}, \quad \hat{v}^-_0=\frac{v_0-q_{0}}{h}, \quad \hat p^{+}_{1}=\frac{p_2-p_{1}}{h},\\
&\frac{\widehat{\partial L}^{+}_d}{\partial v_0}(q_{0},v_{0})=\frac{1}{h}\left( \frac{\partial L_d}{\partial q_1}(q_{1},v_{1})+\frac{\partial L_d}{\partial v_1}(q_{1},v_{1})\right)
\end{split}
\]
and we get the $(+)$-discrete Lagrange--d'Alembert--Pontryagin equations:
\begin{equation*}
p_{1}=\frac{\partial L_{d}}{\partial v_{0}}(q_0,v_0),\;\; p_{1}+ \frac{\partial L_{d}}{\partial q_{1}}(q_{1},v_{1}) \in \Delta^{\circ}_Q(q_{1}),\;\;v_0=q_{1},\;\
(q_0,v_{0})\in \Delta_Q^{d+},
\end{equation*}
which leads to the $(+)$-discrete Lagrange--d'Alembert equations:
\[
\frac{\partial L_{d}}{\partial q_{1}}(q_{0},q_{1})+\frac{\partial L_{d}}{\partial q_1}(q_1,q_{2}) \in \Delta^{\circ}_Q(q_1),\quad (q_{0},q_{1})\in \Delta_Q^{d+}.
\]
\smallskip
\item[(ii)] Similarly, for the $(-)$ part, setting  
\[
\begin{split}
&X^{d-}(q_{0},p_{0})=(q_{0},p_{0},q_{1},p_{1}),\quad \hat X^{d-}(q_{0},p_{0})=(q_{0},p_{0},\hat{q}^+_{0}, \hat{p}^-_{0}),
\end{split}
\]
and
\[
\begin{split}
\mathbf{d}^{-}_{D}L_{d}(v_{1},q_{1})
&= \left(v_{1}, -D_1L_d(v_{1},q_{1}) , \widehat{D_{1}L}_{d}^{-}(v_{1},q_{1}),\hat{v}^{+}_{1}\right).\\
&=\left(v_1,-\frac{\partial L_d}{\partial v_1}(v_{1},q_{1}),\; \frac{\widehat{\partial L}^{-}_d}{\partial v_1}(v_{1},q_{1}),\; \hat{v}^{+}_1 \right),
\end{split}
\]
it follows from \eqref{Loc_discrete_Dirac_minus} that the condition \eqref{EqDiscLagDiracDynSys_minus} becomes
\begin{equation*}
(q_0,v_{0})\in \Delta^{d-}_Q,\quad \frac{\widehat{\partial L}^{-}_d}{\partial v_1}(v_{1},q_{1})+\hat p^{-}_{0} \in \Delta^{\circ}_{Q}(q_{0}), \quad \hat{q}^+_{0}=\hat{v}^{+}_1,
\end{equation*}
together with the $(-)$-Legendre transform 
$
p_{0}=-\frac{\partial L_{d}}{\partial v_{1}}(v_1,q_1).
$
It follows from \eqref{hat_rules_dL-} and \eqref{disc_minus_VectorField} that 
\[
\begin{split}
&\hat{q}^+_0=\frac{q_1-q_0}{h}, \quad \hat{v}^+_1=\frac{q_1-v_1}{h}, \quad \hat p^{-}_{0}=\frac{p_0-p_{-1}}{h},\\
&\frac{\widehat{\partial L}^{-}_d}{\partial v_1}(v_{1},q_{1})=-\frac{1}{h}\left( \frac{\partial L_d}{\partial q_0}(v_{0},q_{0})+\frac{\partial L_d}{\partial v_0}(v_{0},q_{0})\right)
\end{split}
\]
and we get the $(-)$-discrete Lagrange--d'Alembert--Pontryagin equations:
\begin{equation*}
p_{0}=-\frac{\partial L_{d}}{\partial v_1}(v_1,q_1),\;\; p_0- \frac{\partial L_{d}}{\partial q_{0}}(v_{0},q_{0}) \in \Delta^{\circ}_Q(q_0),\;\;v_1=q_{0},\;\
(v_1,q_1)\in \Delta_Q^{d-}.
\end{equation*}
Finally, we can recover the $(-)$-discrete Lagrange--d'Alembert equations:
\[
\frac{\partial L_{d}}{\partial q_{0}}(q_{-1},q_{0}) +\frac{\partial L_{d}}{\partial q_0}(q_0,q_{1}) \in \Delta^{\circ}_Q(q_0)\quad (q_0,q_{1})\in \Delta_Q^{d-}.
\]
\end{itemize}
\end{proof}

\begin{remark}\rm
In Theorem \ref{Thm:DiscLagDiracDynSys}, we derived the $(\pm)$-discrete Lagrange--d'Alembert equations within the context of the $(\pm)$ discrete induced Dirac structures, in which the nonholonomic constraints are described in different ways by using $\Delta^{d\pm}_{Q}$. In particular, the $(-)$-discrete Lagrange--d'Alembert equations are identical to \eqref{eq: local_discrete_LDA}, which were originally derived in \cite{CoMa2001,McPe2006}.
\end{remark}

\begin{remark}\rm
Here we employ the $(\pm)$ signs   for the bundle structures in correspondence with those of the discrete bundle maps; namely, the $(+)$ sign of the formulas in Figs. \ref{plus_DisBundleGeometry} and \ref{BundlePic1_discrete_coordinates_plus} is allocated to be consistent with the $(+)$ signs of the discrete symplectic flat map $({\Omega}^{d+}_{T^{\ast}Q})^{\flat}$, the discrete map $\gamma_{Q}^{+}$, and the discrete Legendre transform $\mathbb{F}^{+}L_{d}$. On the other hand,  the $(-)$ sign of the formulas in Figs. \ref{minus_DisBundleGeometry} and \ref{BundlePic1_discrete_coordinates_minus} is allocated to be consistent with the $(-)$ signs of the discrete symplectic flat map $({\Omega}^{d-}_{T^{\ast}Q})^{\flat}$, the discrete map $\gamma_{Q}^{-}$, and the discrete Legendre transform $\mathbb{F}^{-}L_{d}$.
\end{remark}

\section{Variational discretization of Lagrange--Dirac dynamical systems.} 
We have shown the $(\pm)$-discrete induced Dirac structures and the associated $(\pm)$-discrete Lagrange--Dirac dynamical systems by analogy with the continuous setting proposed by \cite{YoMa2006a}. In the continuous setting, it was shown by \cite{YoMa2006b} that the variational structure of the Lagrange--Dirac dynamical system is given by the {\it Lagrange--d'Alembert--Pontryagin principle}. Therefore, here we will explore the variational structures associated with the $(\pm)$-discrete Lagrange--Dirac dynamical systems by discretizing the Lagrange--d'Alembert--Pontryagin principle.

\subsection{Variational structure of Lagrange--Dirac dynamical systems} 
\paragraph{The Lagrange--d'Alembert--Pontryagin principle.} 
There exists a variational structure behind the Lagrange--Dirac dynamical system, which is called the \textit{Lagrange--d'Alembert--Pontryagin principle}. Here following \cite{YoMa2006b}, we will make a brief review on this principle.
\medskip

Consider a path space on the Pontryagin bundle $TQ\oplus T^{\ast}Q$ over $Q$ as
\begin{equation*}
\begin{split}
\mathcal{C} (TQ\oplus T^{\ast}Q) &=\left\{ (q,v,p):I=[0,T] \to TQ \oplus T^{\ast}Q \mid (q,v,p) \;  \textrm{is a smooth curve} \right.\\
&\left. \textrm{on $TQ \oplus T^{\ast}Q$ such that}\;  q(0)=q_1\;\textrm{and}\; q(T)=q_2 \right\},
\end{split}
\end{equation*}
where $I=[0,T] \subset \mathbb{R}^{+}$ is the space of time.
\medskip

Define the action functional $\tilde{\mathfrak{S}}: \mathcal{C} (TQ\oplus T^{\ast}Q)  \to \mathbb{R}$ by
\[
\tilde{\mathfrak{S}}(q,v,p)= \int ^{T}_{0}  \left\{ L(q(t),v(t)) + \left\langle p(t),  \dot{q}(t)-v(t) \right\rangle \right\} \, dt.
\]
Let $(q,v,p)=(q(t),v(t),p(t)) \in \mathcal{C} ((TQ\oplus T^{\ast}Q)$ be a critical curve of $\tilde{\mathfrak{S}}: \mathcal{C} (Q)  \to \mathbb{R}$, namely,
\begin{equation}\label{LagDALP}
\begin{split}
\delta \tilde{\mathfrak{S}}(q,v,p)=\delta\int_{0}^{T}\left\{ L(q(t),v(t)) + \left\langle p(t),  \dot{q}(t)-v(t) \right\rangle \right\} \, dt=0,
\end{split}
\end{equation}
subject to $ \delta q(t)\in \Delta _Q (q(t) )$ with the fixed endpoint conditions $\delta{q}(0)=\delta{q}(T)=0$ and with the kinematic constraint $ \dot{q}(t) \in \Delta _Q (q(t))$. 

By direct computations,  the critical condition of \eqref{LagDALP} can be developed as
\begin{equation*}
\begin{split}
\delta\int_{0}^{T}&\left\{ L(q(t),v(t)) + \left\langle p(t),  \dot{q}(t)-v(t) \right\rangle \right\} \, dt\\
&=\int_{0}^{T}\left\{ \frac{\partial L}{\partial q} \delta q +\frac{\partial L}{\partial v} \delta v + \left\langle \delta p,  \dot{q}-v \right\rangle
+ \left\langle p,  \delta\dot{q}- \delta v \right\rangle\right\} \, dt\\
&=\int_{0}^{T}\left\{ \left<\frac{\partial L}{\partial q}-\dot{p},  \delta q\right> +\left<\frac{\partial L}{\partial v}-p, \delta v\right> + \left\langle \delta p,  \dot{q}-v \right\rangle
\right\} \, dt +\left< p, \delta{q}\right> \biggr{\arrowvert}_{0}^{T}=0.
\end{split}
\end{equation*}

Imposing the variational condition $ \delta q(t)\in \Delta _Q (q(t) )$ with the fixed endpoint conditions $\delta{q}(0)=\delta{q}(T)=0$ together with the kinematic constraints $ \dot{q}(t) \in \Delta _Q (q(t))$, we get the {\it Lagrange--d'Alembert--Pontryagin equations} as
\begin{equation*}
p =\frac{\partial L}{\partial v }, \quad \dot{q} =v \in \Delta_{Q}(q), \quad  \dot{p} - \frac{\partial L}{\partial q}
\in \Delta^{\circ}_{Q}(q),
\end{equation*}
which are equivalent with the equations of motion for the Lagrange--Dirac dynamical system given in \eqref{LDAPeqn}.

\paragraph{The variational formula based on the generalized energy.}
Now let us define the \textit{generalized energy} on $TQ \oplus T^{\ast}Q$ by
\begin{equation}\label{GeneralizedEnergy}
E(q,\dot{q},p):=\left<p,\dot{q}\right>-L(q,\dot{q}).
\end{equation}
Then, employing the generalized energy, the Lagrange--d'Alembert--Pontryagin principle in \eqref{LagDALP} can be equivalently restated by
\begin{equation*}
\delta\int_{0}^{T}\left\{ \left\langle p(t),  \dot{q}(t) \right\rangle -E(q(t), v(t),p(t))\right\} \, dt
 =0
\end{equation*}
with respect to $ \delta q(t)\in \Delta _Q (q(t))$, together with the kinematic constraint $ \dot{q}(t) \in \Delta _Q (q(t))$; see \cite{YoMa2006b}.
\medskip

Setting $u(t)=(q(t),v(t)) \in TQ$ and $z(t)=(q(t),p(t)) \in T^{\ast}Q$, where $q(t)=\tau_{Q}(u(t))=\pi_{Q}(z(t)) \in Q$ is the base curve on $Q$, the criticality condition of the Lagrange--d'Alembert--Pontryagin principle can be intrinsically restated by,  for each $(u(t),z(t))=(q(t),v(t),p(t))\in TQ \oplus T^{\ast}Q$,
\begin{equation}\label{IntLagDAPPrin}
\delta\int_{0}^{T}\left\{ \Theta_{T^{\ast}Q}(z(t))(\dot{z}(t))-E(u(t),z(t)) \right\}\, dt=0,
\end{equation}
subject to the variational constraint $\delta{q}(t)=T\tau_{Q}(\delta{u})=T\pi_{Q}(\delta{z}(t))  \in \Delta_{Q}(q(t))$ with $\delta{q}(0)=\delta{q}(T)=0$, and also subject to the kinematic constraint $\dot{q}(t)=T\tau_{Q}(\dot{u}(t))=T\pi_{Q}(\dot{z}(t)) \in \Delta_{Q}(q(t))$.

\subsection{The discrete Lagrange--d'Alembert--Pontryagin principles}
Here, we  introduce the $(\pm)$ discretizations of the Lagrange--d'Alembert--Pontryagin principle as given in \eqref{LagDALP} or \eqref{IntLagDAPPrin}. We then demonstrate that the resulting $(\pm)$-discrete Lagrange--d'Alembert--Pontryagin equations are identical to the equations of motion for the $(\pm)$-discrete Lagrange--Dirac systems and, consequently, equivalent to the discrete Lagrange--d'Alembert equations in 
\cite{CoMa2001,McPe2006}, as presented in Theorem \ref{Thm:DiscLagDiracDynSys}. This implies that the resulting $(\pm)$-nonholonomic integrators preserve the $(\pm)$-induced discrete Dirac structures.

\paragraph{The discrete generalized energy.}
Recall from Def. \ref{Def:discSymOneForm} that the $(\pm)$-discrete canonical one-form $\Theta_{T^{\ast}Q}^{\pm}$ are given as
 by, for each $(z_{0},z_{1})=(q_{0},p_{0},q_{1},p_{1}) \in T^{\ast}Q \times T^{\ast}Q$,
\begin{equation*}
\begin{split}
\Theta^{d+}_{T^{\ast}Q}(z_{0}, z_1)&=\Theta_{T^{\ast}Q}(q_{1},p_{1})(\hat{q}^-_{1}, \hat{p}^+_1)=(p_{1}dq_{1})\left( \hat{q}^{-}_{1} \frac{\partial }{\partial q_{1}}+ \hat{p}^{+}_{0} \frac{\partial }{\partial p_{1}}\right)=\langle p_{1}, \hat{q}^{-}_{1}\rangle,\\
\Theta^{d-}_{T^{\ast}Q}(z_{0}, z_1)&=\Theta_{T^{\ast}Q}(q_{0},p_{0})(\hat{q}^+_{0}, \hat{p}^-_0)=(p_{0}dq_{0})\left( \hat{q}^{+}_{0} \frac{\partial }{\partial q_{0}}+ \hat{p}^{-}_{0} \frac{\partial }{\partial p_{0}}\right)=\langle p_{0}, \hat{q}^{+}_{0}\rangle.
\end{split}
\end{equation*}

\begin{definition}\rm
Associated with the generalized energy in \eqref{GeneralizedEnergy}, let us introduce the $(\pm)$-generalized energies $E^{\pm}_{d}: (Q \times Q) \oplus T^{\ast}Q \to \mathbb{R}$ as follows.

\begin{itemize}
\item[(i)] For a given discrete Lagrangian $L_d: Q \times Q \to \mathbb{R}$,
define the (+)-\textit{discrete generalized energy}  $E^+_{d}: (Q \times Q) \oplus T^{\ast}Q \to \mathbb{R}$ by using $(q_{k},v_{k}) \in Q \times Q$ and $(q_{k+1},p_{k+1})\in T^{\ast}Q$ as
\[
\begin{split}
E^+_{d}(q_{k},v_{k},p_{k+1})&=h\left<p_{k+1},\hat{v}^-_{k}\right>-L_d(q_{k} , v_{k})\\
&=\left<p_{k+1},v_k-q_k\right>-L_d( q_{k}, v_k).
\end{split}
\]
\item[(ii)] Define the (-)-\textit{discrete generalized energy}  $E^-_{d}: (Q \times Q) \oplus T^{\ast}Q \to \mathbb{R}$ by using $(v_{k+1}, q_{k+1}) \in Q \times Q$ and $(q_{k+1},p_{k+1})\in T^{\ast}Q$ as
\[
\begin{split}
E^-_{d}(v_{k+1},q_{k+1}, p_{k})&=h\left<p_{k}, \hat{v}^+_{k+1}\right>-L_d(v_{k+1}, q_{k+1})\\
&=\left<p_{k},q_{k+1}-v_{k+1}\right>-L_d(v_{k+1}, q_{k+1}).
\end{split}
\]

\end{itemize}
\end{definition}

\begin{remark}\rm
More intrinsically, the $(+)$-discrete generalized energy can be given, in coordinates $(q_{k},q_{k+1},p_{k+1}) \in (Q \times Q) \oplus T^{\ast}Q$, by
\begin{equation*}
\begin{split}
E^+_{d}(q_k, q_{k+1}, p_{k+1})&=h\left<(q_{k+1},p_{k+1}, p_{k+1}, 0),(q_k, p_k, q_{k+1}, p_{k+1})\right>_{d+}-L_d(q_{k} , q_{k+1})\\
&=h\left<(q_{k+1},p_{k+1}, p_{k+1}, 0),\Psi^+_{(T^{\ast}Q)_2}(q_k, p_k, q_{k+1}, p_{k+1})\right>-L_d(q_{k} , q_{k+1})\\
&=h\left<(q_{k+1},p_{k+1}, p_{k+1}, 0),(q_{k+1}, p_{k+1}, \hat{q}^-_{k+1}, \hat{p}^+_{k+1})\right>-L_d(q_{k} , q_{k+1})\\
&=h\left<(p_{k+1})_{q_{k+1}},\hat{q}^-_{k+1}\right>-L_d(q_{k} , q_{k+1})\\
&=h\left<(p_{k+1})_{q_{k+1}},\left(\frac{q_{k+1}-q_k}{h}\right)_{q_{k+1}}\right>-L_d(q_{k} , q_{k+1})\\
&=\left<p_{k+1}, q_{k+1}-q_k \right>-L_d(q_{k} , q_{k+1}).
\end{split}
\end{equation*} 
Then, by setting $v_k=q_{k+1}$, we get $E^+_{d}(q_k, v_{k}, p_{k+1})=\left<p_{k+1},v_k-q_k\right>-L_d( q_{k}, v_k)$. 
\medskip

Similarly, the $(-)$-discrete generalized energy can be given, in coordinates $(q_{k},q_{k+1},p_{k}) \in (Q \times Q) \oplus T^{\ast}Q$, by
\begin{equation*}
\begin{split}
E^-_{d}(q_k, q_{k+1}, p_{k})&=h\left<(q_{k},p_{k}, p_k, 0),(q_k, p_k, q_{k+1}, p_{k+1})\right>_{d-}-L_d(q_{k} , q_{k+1})\\
&=h\left<(q_{k},p_{k}, p_k, 0),\Psi^-_{(T^{\ast}Q)_1}(q_k, p_k, q_{k+1}, p_{k+1})\right>-L_d(q_{k} , q_{k+1})\\
&=h\left<(q_{k},p_{k}, p_k, 0),(q_{k}, p_{k}, \hat{q}^+_{k}, \hat{p}^-_{k})\right>-L_d(q_{k} , q_{k+1})\\
&=h\left<(p_{k})_{q_{k}}, \hat{q}^+_{k}\right>-L_d(q_{k} , q_{k+1})\\
&=h\left<(p_{k})_{q_{k}},\left(\frac{q_{k+1}-q_k}{h}\right)_{q_{k}}\right>-L_d(q_{k} , q_{k+1})\\
&=\left<p_{k}, q_{k+1}-q_k \right>-L_d(q_{k} , q_{k+1}).
\end{split}
\end{equation*} 
Then, by setting $v_{k+1}=q_{k}$, we get $E^+_{d}(v_{k+1}, q_{k+1}, p_{k})=\left<p_{k},q_{k+1}-v_{k+1}\right>-L_d(v_{k+1}, q_{k+1}).$

\end{remark}
%
\begin{definition}\rm
Consider a discrete path $(q_{d}, v_{d}, p_{d})$ on $(Q \times Q) \oplus T^{\ast}Q$, where $q_{d}\!=\!\{q_{k}\}_{k=0}^{N}$, \!
$v_{d}\!=\!\{v_{k}\}_{k=0}^{N}$\! and \!$p_{d}\!=\!\{p_{k}\}_{k=0}^{N}$ with $q_{k}\!=\!q_{d}(t_{k})$, $v_{k}\!=\!v_{d}(t_{k})$, and $p_{k}\!=\!p_{d}(t_{k})$.  
\medskip

\begin{itemize}
\item[(i)] Define the $(+)$-\textit{discrete action sum} $\tilde{\mathfrak{S}}^+_{d}: \mathcal{C}_d((Q \times Q) \oplus T^{\ast}Q) \to \mathbb{R}$ by
\begin{equation*}
\begin{split}
\tilde{\mathfrak{S}}^+_{d}(q_{d}, v_{d}, p_{d})
&=\sum_{k=0}^{N-1}
\left\{h\,\Theta^{d+}_{T^{\ast}Q}(q_k, p_k, q_{k+1}, p_{k+1})-E^+_{d}(q_{k},v_{k},p_{k+1})\right\}\\
&=\sum_{k=0}^{N-1}
\left\{
L_d( q_{k} , v_{k})+h\left< p_{k+1}, \hat{q}^{-}_{k+1}-\frac{v_{k}-q_{k}}{h}\right>
\right\}\\
&=\sum_{k=0}^{N-1}
\left\{
L_d( q_{k} , v_{k})+\left< p_{k+1}, q_{k+1}-v_{k}\right>
\right\}.\\
\end{split}
\end{equation*}

\item[(ii)] Define the $(-)$-\textit{discrete action sum} $\tilde{\mathfrak{S}}_{d}^{-}: \mathcal{C}_d((Q \times Q) \oplus T^{\ast}Q) \to \mathbb{R}$ by
\begin{equation*}
\begin{split}
\tilde{\mathfrak{S}}_{d}^{-}(q_{d}, v_{d}, p_{d})
&=\sum_{k=0}^{N-1}
\left\{h\,\Theta^{d-}_{T^{\ast}Q}(q_k, p_k, q_{k+1}, p_{k+1})-E^-_{d}(v_{k+1},q_{k+1}, p_{k})\right\}\\
&=\sum_{k=0}^{N-1}
\left\{
L_d(v_{k+1},q_{k+1})+h\left< p_{k}, \hat{q}^{+}_{k}-\frac{q_{k+1}-v_{k+1}}{h}\right>
\right\}\\
&=\sum_{k=0}^{N-1}
\left\{
L_d(v_{k+1},q_{k+1})+\left< p_{k}, v_{k+1}-q_{k}\right>
\right\}.
\end{split}
\end{equation*}
\end{itemize}
\end{definition}
\begin{framed}
\begin{theorem}\rm
The $(\pm)$-discretized Lagrange--d'Alembert--Pontryagin principle is given by the following statements.
\begin{itemize}
\item[(i)]
A discrete path $(q_{d}, v_{d}, p_{d})$ on $(Q \times Q) \oplus T^{\ast}Q$ joining $q_{0}$ and $q_{N}$ satisfies the $(+)$-discretized Lagrange--d'Alembert--Pontryagin equations
\begin{equation}\label{+disLagDAPeq}
\begin{split}
&p_{k+1}=\frac{\partial L_{d}}{\partial v_{k}}(q_k,v_k),\;p_{k}+ \frac{\partial L_{d}}{\partial q_{k}}(q_{k},v_{k}) \in \Delta_Q^{\circ}(q_{k}),\;v_k=q_{k+1},\;(q_k,v_{k})\in \Delta_Q^{d+},
\end{split}
\end{equation}
if and only if
\begin{equation*}
\begin{split}
\delta\tilde{\mathfrak{S}}_{d}^+(q_{d}, v_{d}, p_{d})&=\delta\sum_{k=0}^{N-1}
\left\{h\,\Theta^{d+}_{T^{\ast}Q}(q_k, p_k, q_{k+1}, p_{k+1})-E^+_{d}(q_{k}, v_{k},p_{k+1})\right\}\\
&=\delta\sum_{k=0}^{N-1}
\left\{
L_d( q_{k} , v_{k})+\left< p_{k+1}, q_{k+1}-v_{k}\right>
\right\}=0,
\end{split}
\end{equation*}
with respect to $\delta{q}_{k} \in \Delta_{Q}(q_{k})$, for all $\delta{v_k}$ and $\delta{p}_{k+1}$,  together with the discrete constraints $(q_k,v_{k})\in \Delta_{Q}^{d+}$.
\medskip

\quad From the $(+)$-discrete Lagrange--d'Alembert--Pontryagin equations in \eqref{+disLagDAPeq}, the $(+)$-discrete Lagrange--d'Alembert equations can be recovered as:
\[
D_{2}L_{d}(q_{k-1},q_{k})+ D_{1}L_{d}(q_{k},q_{k+1}) \in \Delta^{\circ}_{Q}(q_{k}),\quad 
(q_{k},q_{k+1}) \in \Delta_{Q}^{d+}.
\]
\item[(ii)]
A discrete path $(q_{d}, v_{d}, p_{d})$ on $(Q \times Q) \oplus T^{\ast}Q$ joining $q_{0}$ and $q_{N}$ satisfies the $(-)$-discretized Lagrange--d'Alembert--Pontryagin equations
\begin{equation}\label{-disLagDAPeq}
\begin{split}
&p_{k}=-\frac{\partial L_{d}}{\partial v_{k+1}}(v_{k+1},q_{k+1}),\;\; p_{k+1}- \frac{\partial L_{d}}{\partial q_{k+1}}(v_{k+1},q_{k+1}) \in \Delta^{\circ}_Q(q_{k+1}),\\\
&v_{k+1}=q_{k},\;\;(v_{k+1},q_{k+1})\in \Delta_Q^{d-},
\end{split}
\end{equation}
if and only if
\begin{equation*}
\begin{split}
\delta\tilde{\mathfrak{S}}_{d}^{-}(q_{d}, v_{d}, p_{d})&=\sum_{k=0}^{N-1}
\left\{h\,\Theta^{d-}_{T^{\ast}Q}(q_k, p_k, q_{k+1}, p_{k+1})-E^-_{d}(v_{k+1},q_{k+1}, p_{k})\right\}\\
&=\delta\sum_{k=0}^{N-1}
\left\{
L_d(v_{k+1},q_{k+1})+\left< p_{k}, v_{k+1}-q_{k}\right>
\right\}=0,
\end{split}
\end{equation*}
with respect to $\delta{q}_{k+1} \in \Delta_{Q}(q_{k+1})$, for all $\delta{v_{k+1}}$ and $\delta{p}_{k}$,  together with the discrete constraints $(v_{k+1}, q_{k+1})\in \Delta_{Q}^{d-}$.
\medskip

\quad From the $(-)$-discrete Lagrange--d'Alembert--Pontryagin equations in \eqref{-disLagDAPeq}, the $(-)$-discrete Lagrange--d'Alembert equations can be recovered as:
\[
D_{2}L_{d}(q_{k},q_{k+1})+D_{1}L_{d}(q_{k+1},q_{k+2}) \in \Delta^{\circ}_{Q}(q_{k+1}),\quad 
(q_{k},q_{k+1}) \in \Delta_{Q}^{d-}.
\]

\end{itemize}
\end{theorem}
\end{framed}
\begin{proof}
\begin{itemize}
\item[(i)]
By direct computations, the criticality condition reads
\begin{equation*}
\begin{split}
&\delta\tilde{\mathfrak{S}}^+_{d}(q_{d}, v_{d}, p_{d})=
\delta\sum_{k=0}^{N-1}
\left\{
L_d( q_{k} , v_{k})+\bigl< p_{k+1}, q_{k+1}-v_{k}\bigr>
\right\}\\
&=\sum_{k=0}^{N-1}
\biggl\{
\left<\frac{\partial L_d}{\partial q_{k}}( q_{k} , v_{k}), \delta q_{k}\right>+\left<\frac{\partial L_d}{\partial v_{k}}( q_{k} , v_{k}),\delta v_{k}\right>\biggr.\\
&\biggl.\hspace{5cm}+\bigl< \delta p_{k+1}, q_{k+1}-v_{k}\bigr>+\bigl< p_{k+1}, \delta q_{k+1}- \delta v_{k}\bigr>
\biggr\}\\
&=\sum_{k=0}^{N-1}\biggl\{\left<-p_{k+1}+\frac{\partial L_d}{\partial v_{k}}( q_{k} , v_{k}),\delta v_{k}\right>+\bigl< \delta p_{k+1}, q_{k+1}-v_{k}\bigr>\biggr\}\\
&\hspace{1cm}+\sum_{k=1}^{N-1}\left<p_{k}+\frac{\partial L_d}{\partial q_{k}}( q_{k} , v_{k}), \delta q_{k}\right>
+\left<p_{N},  \delta q_{N}\right>+\left<\frac{\partial L_d}{\partial q_{0}}( q_{0} , v_{0}), \delta q_{0}\right>=0,
\end{split}
\end{equation*}
for all $\delta{q}_{k} \in \Delta_{Q}(q_{k})$, for all $\delta{v_k}$ and $\delta{p}_{k+1}$, and  with the discrete constraint $(q_k,v_{k})\in \Delta_Q^{d+}$. Imposing the fixed endpoint conditions $\delta q_{0}=\delta q_{N}=0$, we get $(+)$-discretized Lagrange--d'Alembert--Pontryagin equations:
\[
p_{k+1}=\frac{\partial L_{d}}{\partial v_{k}}(q_k,v_k),\;\; p_{k}+ \frac{\partial L_{d}}{\partial q_{k}}(q_{k},v_{k}) \in \Delta_Q^{\circ}(q_{k}),\;\;v_k=q_{k+1},\;\
(q_k,v_{k})\in \Delta_Q^{d+}.
\]
Thus, we get the $(+)$-discrete Lagrange--d'Alembert equations can be recovered as:
\[
D_{2}L_{d}(q_{k-1},q_{k})+ D_{1}L_{d}(q_{k},q_{k+1}) \in \Delta^{\circ}_{Q}(q_{k}),\quad 
(q_{k},q_{k+1}) \in \Delta_{Q}^{d+}.
\]

\item[(ii)]
Similarly, 
the criticality condition reads
\begin{equation*}
\begin{split}
&\delta\tilde{\mathfrak{S}}^-_{d}(q_{d}, v_{d}, p_{d})=
\delta\sum_{k=0}^{N-1}
\left\{
L_d(v_{k+1},q_{k+1})+\left< p_{k}, v_{k+1}-q_{k}\right>
\right\}\\
&=\sum_{k=0}^{N-1}
\biggl\{
\left<\frac{\partial L_d}{\partial v_{k+1}}(v_{k+1},q_{k+1}), \delta v_{k+1}\right>+\left<\frac{\partial L_d}{\partial q_{k+1}}( v_{k+1} , q_{k+1}),\delta q_{k+1}\right>\biggr.\\
&\biggl.\hspace{7cm}+\bigl< \delta p_{k}, v_{k+1}-q_{k}\bigr>+\bigl< p_{k}, \delta v_{k+1}- \delta q_{k}\bigr>
\biggr\}\\
&=\sum_{k=0}^{N-1}\biggl\{\left<p_{k}-\frac{\partial L_d}{\partial v_{k+1}}(v_{k+1} , q_{k+1}),\delta v_{k+1}\right>+\bigl< \delta p_{k}, v_{k+1}-q_{k}\bigr>\biggr\}\\
&\hspace{0.5cm}
+\sum_{k=0}^{N-2}
\left<-p_{k+1}+\frac{\partial L_d}{\partial q_{k+1}}(v_{k+1} , q_{k+1}), \delta q_{k+1}\right>
-\left<p_{0},  \delta q_{0}\right>+\left<\frac{\partial L_d}{\partial q_{N}}(v_{N} , q_{N}), \delta q_{N}\right>=0,
\end{split}
\end{equation*}
for all $\delta{q}_{k+1} \in \Delta_{Q}(q_{k+1})$, for all $\delta{v_{k+1}}$ and $\delta{p}_{k}$, and  with the discrete constraint $(q_k,v_{k})\in \Delta_Q^{d-}$. Imposing the fixed endpoint conditions $\delta q_{0}=\delta q_{N}=0$, we get the $-$discretized Lagrange--d'Alembert--Pontryagin equations:
\[
\begin{split}
&p_{k}=-\frac{\partial L_{d}}{\partial v_{k+1}}(v_{k+1},q_{k+1}),\;\; p_{k+1}- \frac{\partial L_{d}}{\partial q_{k+1}}(v_{k+1},q_{k+1}) \in \Delta^{\circ}_Q(q_{k+1}),\\\
&v_{k+1}=q_{k},\;\;(v_{k+1},q_{k+1})\in \Delta_Q^{d-}.
\end{split}
\]
Thus we get the $(-)$-discrete Lagrange--d'Alembert equations can be recovered as:
\[
D_{2}L_{d}(q_{k},q_{k+1})+D_{1}L_{d}(q_{k+1},q_{k+2}) \in \Delta^{\circ}_{Q}(q_{k+1}),\quad 
(q_{k},q_{k+1}) \in \Delta_{Q}^{d-}.
\]

\end{itemize}

\end{proof}

\begin{remark}\rm
    In particular, for the case where the distribution is locally given by \eqref{ConstraintDistribution}, namely,
\begin{equation*}
\Delta_Q(q)=\left\{ \dot{q} \in T_{q}Q\,\mid \left< \omega^{r}(q),\dot{q}\right>=0,\,r=1,...,m <n\right\},
\end{equation*}
where $\omega^{r}=\omega_{i}^{r}(q)dq^{i},\; r=1,...,m<n$ are some given $m$ independent one-forms on $Q$,
the $(+)$-discrete Lagrange--d'Alembert--Pontryagin equations can be rewritten by using Lagrange multipliers as  
\begin{equation*}\label{plus_DiscLagDiracSysv1}
\begin{split}
&p_{k+1}=\frac{\partial L_{d}}{\partial v_{k}}(q_{k},v_{k}),\;\; p_{k}+ \frac{\partial L_{d}}{\partial q_{k}}(q_{k},v_{k}) =(\mu_r)_k\omega^r(q_k),\quad q_{k+1}=v_{k},\;\
\left\langle \omega^r(v_k),\frac{v_k-q_k}{h}\right\rangle=0.
\end{split}
\end{equation*}
Hence, the corresponding $(+)$-discrete Lagrangian--d'Alembert equations are
\[
D_{2}L_{d}(q_{k-1},q_{k}) + D_{1}L_{d}(q_{k},q_{k+1}) =(\mu_r)_k\omega^r(q_k),\quad 
\left\langle \omega^r(q_{k+1}),\frac{q_{k+1}-q_k}{h}\right\rangle=0.
\]

The $(-)$-discrete Lagrange--d'Alembert--Pontryagin equations can be rewritten using Lagrange multiplies as  
\begin{equation*}\label{minus_DiscLagDiracSysv1}
\begin{split}
&p_{k}=-\frac{\partial L_{d}}{\partial v_{k+1}}(v_{k+1},q_{k+1}),\;\; p_{k}- \frac{\partial L_{d}}{\partial q_{k}}(v_{k},q_{k}) =(\mu_r)_k\omega^r(q_k),\\
&q_k=v_{k+1},\;\
\left\langle \omega^r(v_{k+1}),\frac{q_{k+1}-v_{k+1}}{h}\right\rangle=0,
\end{split}
\end{equation*}
Hence, the corresponding $(-)$-discrete Lagrangian--d'Alembert equations are
\[
D_{2}L_{d}(q_{k-1},q_{k}) + D_{1}L_{d}(q_{k},q_{k+1}) =-(\mu_r)_k\omega^r(q_k),\quad 
\left\langle \omega^r(q_{k}),\frac{q_{k+1}-q_k}{h}\right\rangle=0.
\]
\end{remark}

\section{Examples of the discrete Lagrange--Dirac dynamical systems}\label{exam}
In this section, we illustrate how nonholonomic mechanical systems can be described by the $(\pm)$-discrete Lagrange--Dirac dynamical systems by two examples, i.e., the vertical rolling disk on a plane and the classical Heisenberg system; as to these examples, refer to \cite{Bl2003}.

\subsection{The vertical rolling disk on the plane}
\paragraph{Continuous setting.} In this example, we consider a homogenous disk rolling without slipping on a horizontal plane or tilting away from the vertical. Let $S^1$ denote the unit circle in the plane; it is parameterized by an angular variable. The configuration space of the rolling disk is then $Q=\mathbb{R}^2\times S^1\times S^1$, whose arbitrary point $q$ is denoted by the local coordinates $(x,y, \theta,\phi)$. Here $(x,y)$ indicates the position of the contact point on the $xy$-plane, $\theta$ the rotation angle of the disk, and $\phi$ the orientation of the disk (Fig. \ref{fig:rollingdisk}). The kinetic energy of the rolling disk is given by
\begin{equation*}
K(q,\dot{q})=\frac{1}{2}m\left(\dot{x}^2+\dot{y}^2\right)+\frac{1}{2}I\dot{\theta}^2+\frac{1}{2}J\dot{\phi}^2,
\end{equation*}
where $m$ is the mass of the disk, $I$ is the moment of inertia of the disk about the axis perpendicular to the plan of the disk, and $J$ is the moment of inertia about an axis in the plan of the disk. As illustrated in \cite{CoMa2001, Co2002}, we add an artificial potential $V=-10 \sin \theta$ in order to force the system. Then the Lagrangian of the rolling disk to be
\begin{equation*}
L(q,\dot{q})=K(q,\dot{q})-V(q)=\frac{1}{2}m\left(\dot{x}^2+\dot{y}^2\right)+\frac{1}{2}I\dot{\theta}^2+\frac{1}{2}J\dot{\phi}^2-10 \sin \theta.
\end{equation*}

As the point fixed on the rim of the disk has zero velocity at the point of contact with the $xy$-plane, we have the following nonholonomic constraints of rolling without slipping
\begin{equation}\label{rdc}
\begin{aligned}
&\dot{x}=R (\cos \phi) \dot{\theta},\\
&\dot{y}=R(\sin \phi) \dot{\theta},
\end{aligned}
\end{equation}
where $R$ is the radius of the disk.

\begin{figure}[htbp]
  \begin{center}
    \includegraphics[scale=0.25]{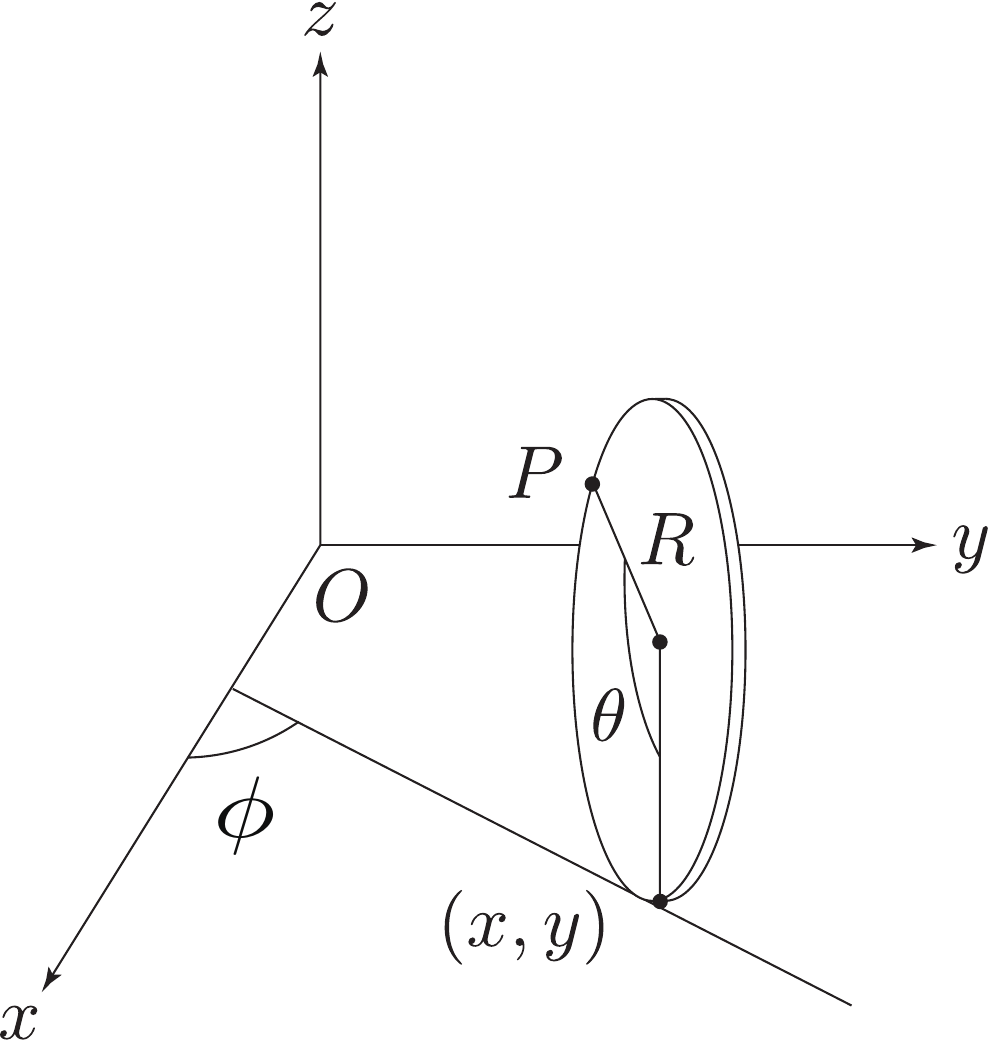}
  \end{center}
  \caption{A rolling disk on the plane.}
  \label{fig:rollingdisk}
\end{figure}

 Noting $q^{1}=x, q^{2}=y, q^{3}=\theta$, and $q^{4}=\phi$, the nonholonomic constraints are given by the distribution 
\[ 
\Delta_Q(q)=\left\{ (q,\dot{q}) \mid \left< \omega^{r}(q), \dot{q} \right>=0, \;\; r=1, 2 \right\},
\]
where $\omega^r(q)=\omega^r_i(q)dq^i,\;r=1,2$ are the differential forms on $Q$ given by
\begin{equation*}
\omega^1(q)=dx-R(\cos\phi)d\theta,\quad\omega^2(q)=dy-R(\sin\phi)d\theta,
\end{equation*}
and where  the non-zero coefficients are
\begin{equation*}
\omega_1^1(q)=1,\quad\omega_3^1(q)=-R\cos\phi,\quad\omega_2^2(q)=1,\quad\omega_3^2(q)=-R\sin\phi.
\end{equation*}
Using Lagrange multipliers $\mu_{r},\;r=1,2$, the equations of motion of the continuous Lagrange--Dirac system given in \eqref{LDAPeqn} are
\begin{equation*}
\begin{split}
&p_x=mv_{x},\quad p_y=m v_{y},\quad p_{\theta}=I v_{\theta},\quad p_{\phi}=Jv_{\phi},\\ 
&\dot{x}=v_{x},\quad \dot{y}=v_{y},\quad \dot{\theta}=v_{\theta},\quad \dot{\phi}=v_{\phi},\\
&\dot{p}_{x}=\mu_1\omega_1^1,\quad \dot{p}_{y}=\mu_2\omega_2^2,\quad \dot{p}_{\theta}+10 \cos \theta=\mu_1\omega_3^1+\mu_2\omega_3^2,\quad \dot p_{\phi}=0,
\end{split}
\end{equation*}
together with the nonholonomic constraints $\omega_i^r(q)\dot{q}^i=0$ in \eqref{rdc}.

\paragraph{The (+)-discrete Lagrange--Dirac dynamical system.}
For the Lagrangian
\begin{equation*}
L(q,v)=\frac{1}{2}m\left(v_{x}^2+v_{y}^2\right)+\frac{1}{2}Iv_{\theta}^2+\frac{1}{2}Jv_{\phi}^2-10 \sin \theta_{k},
\end{equation*}
the discrete Lagrangian is given in local coordinates $(q_{k},v_{k}) \in Q \times Q$ by
\[
\begin{split}
L_{d}(q_{k},v_{k})&=h L\left(q_{k}, \frac{v_{k}-q_k}{h}\right)\\
&=h\left[\frac{1}{2}m\left\{\left(\frac{(v_{x})_{k}-x_k}{h}\right)^{2}+
\left(\frac{(v_{y})_{k}-y_k}{h}\right)^{2}\right\}+
\frac{1}{2}I\left(\frac{(v_{\theta})_{k}-\theta_k}{h}\right)^{2} \right.\\
&\hspace{6cm}\left.+\frac{1}{2}J\left(\frac{(v_{\phi})_{k}-\phi_k}{h}\right)^{2}
-10 \sin \theta_{k}
\right],
\end{split}
\]
where $h=t_{k+1}-t_{k} \in \mathbb{R}^{+}$ is a constant time step.
Recall from \eqref{plus_DiscLagDiracSys} that the equations of motion for the (+)-discrete Lagrange--Dirac dynamical system are
\begin{equation*}
\begin{split}
&p_{k+1}=\frac{\partial L_{d}}{\partial v_{k}}(q_{k},v_{k}),\;\; p_{k}+ \frac{\partial L_{d}}{\partial q_{k}}(q_{k},v_{k}) \in \Delta^{\circ}_Q(q_{k}),\;\;
q_{k+1}=v_{k},\;\
(q_{k},v_{k})\in \Delta_Q^{d+},
\end{split}
\end{equation*}
and then we get
\[
\begin{split}
&(p_{x})_{k+1}=\frac{m}{h}\left\{(v_{x})_{k}-x_{k}\right\},\;\; (p_{y})_{k+1}=\frac{m}{h}\left\{(v_{y})_{k}-y_{k}\right\},\;\; (p_{\theta})_{k+1}=\frac{I}{h}\left\{(v_{\theta})_{k}-\theta_{k}\right\},\\
&(p_{\phi})_{k+1}=\frac{J}{h}\left\{(v_{\phi})_{k}-\phi_{k}\right\},\;\;(v_{x})_{k}=x_{k+1},\;\; (v_{y})_{k}=y_{k+1},\;\; (v_{\theta})_{k}=\theta_{k+1},\;\; (v_{\phi})_{k}=\phi_{k+1},\\
&(p_{x})_{k}-\frac{m}{h}\left\{(v_{x})_{k}-x_{k}\right\}= (\mu_{1})_{k}\,\omega^{1}_{1}(q_{k}),\;\;
(p_{y})_{k}-\frac{m}{h}\left\{(v_{y})_{k}-y_{k}\right\}= (\mu_{1})_{k}\,\omega^{2}_{2}(q_{k}),\\
&(p_{\theta})_{k}-\frac{I}{h}\left\{(v_{\theta})_{k}-\theta_{k}\right\}-10 h \cos \theta_{k}= (\mu_1)_{k}\,\omega_3^1(q_k)+(\mu_2)_{k}\,\omega_3^2(q_k),\\
&(p_{\phi})_{k}-\frac{J}{h}\left\{(v_{\phi})_{k}-\phi_{k}\right\}= 0,
\end{split}
\]
together with $(q_{k},v_{k})\in \Delta_Q^{d+}$, namely, 
\[
\begin{split}
\frac{(v_{x})_{k}-x_k}{h}-R(\cos\phi_{k+1})\left(\frac{(v_{\theta})_{k}-\theta_k}{h}\right)=0,\;\; 
\frac{(v_{y})_{k}-y_k}{h}-R(\sin\phi_{k+1})\left(\frac{(v_{\theta})_{k}-\theta_k}{h}\right)=0.
\end{split}
\]
Therefore, the $(+)$-discrete Lagrange--Dirac dynamical system \eqref{minus_DiscLagDiracSys} is simplified by eliminating 
the variables $v_k$ and $p_k$, and becomes 
\begin{equation*}
\begin{aligned}
&\frac{m}{h}(2x_k-x_{k-1}-x_{k+1})=(\mu_1)_{k},\quad \frac{m}{h}(2y_k-y_{k-1}-y_{k+1})=(\mu_2)_{k},\\
& \frac{I}{h}(2\theta_k-\theta_{k-1}-\theta_{k+1})-10 h \cos \theta_{k}=-(\mu_1)_{k} R \cos \phi_k-(\mu_2)_{k} R \sin \phi_k ,\\
& \frac{J}{h}(2\phi_k-\phi_{k-1}-\phi_{k+1})=0,
\end{aligned}
\end{equation*}
subject to constraints 
\begin{equation*}
\begin{aligned}
&\frac{x_{k+1}-x_k}{h}-R\cos\phi_{k+1} \frac{\theta_{k+1}-\theta_k}{h}=0,\\
& \frac{y_{k+1}-y_k}{h}-R\sin\phi_{k+1} \frac{\theta_{k+1}-\theta_k}{h}=0.
\end{aligned}
\end{equation*}

\paragraph{The ($-$)-discrete Lagrange--Dirac dynamical system.}
The discrete Lagrangian is given in local coordinates $(v_{k+1},q_{k+1}) \in Q \times Q$ by
\[
\begin{split}
L_{d}(v_{k+1},q_{k+1})&=h L\left(v_{k+1},\frac{q_{k+1}-v_{k+1}}{h}\right)\\
&=h\left[\frac{1}{2}m\left\{\left(\frac{x_{k+1}-(v_{x})_{k+1}}{h}\right)^{2}+\left(\frac{y_{k+1}-(v_{y})_{k+1}}{h}\right)^{2}\right\}\right.\\
&\left.\hspace{1cm}
+\frac{1}{2}I\left(\frac{\theta_{k+1}-(v_{\theta})_{k+1}}{h}\right)^{2}
+\frac{1}{2}J\left(\frac{\phi_{k+1}-(v_{\phi})_{k+1}}{h}\right)^{2}
-10 \sin (v_{\theta})_{k+1}
\right].
\end{split}
\]
Recall from \eqref{minus_DiscLagDiracSys} that the equations of motion for the $(-)$-discrete Lagrange--Dirac dynamical system are
\begin{equation*}
\begin{split}
&p_{k}=-\frac{\partial L_{d}}{\partial v_{k+1}}(v_{k+1},q_{k+1}),\;\; p_{k}- \frac{\partial L_{d}}{\partial q_{k}}(v_{k},q_{k}) \in \Delta^{\circ}_Q(q_{k}),\;\;q_k=v_{k+1},\;\
(v_{k+1},q_{k+1})\in \Delta_Q^{d-},
\end{split}
\end{equation*}
and then we get
\[
\begin{split}
&(p_{x})_{k}=\frac{m}{h}\left\{x_{k+1}-(v_{x})_{k+1}\right\},\;\; (p_{y})_{k}=\frac{m}{h}\left\{y_{k+1}-(v_{y})_{k+1}\right\},\\ 
&(p_{\theta})_{k}=\frac{I}{h}\left\{\theta_{k+1}-(v_{\theta})_{k+1}\right\} +10 \cos (v_{\theta})_{k+1},\\
&(p_{\phi})_{k}=\frac{J}{h}\left\{\phi_{k+1}-(v_{\phi})_{k+1}\right\},\;\;(v_{x})_{k+1}=x_{k},\;\; (v_{y})_{k+1}=y_{k},\;\; (v_{\theta})_{k+1}=\theta_{k},\;\; (v_{\phi})_{k+1}=\phi_{k},\\
&(p_{x})_{k}-\frac{m}{h}\left\{x_{k}-(v_{x})_{k}\right\}= (\mu_{1})_{k}\,\omega^{1}_{1}(q_{k}),\;\;
(p_{y})_{k}-\frac{m}{h}\left\{y_{k}-(v_{y})_{k}\right\}= (\mu_{1})_{k}\,\omega^{2}_{2}(q_{k}),\\
&(p_{\theta})_{k}-\frac{I}{h}\left\{\theta_{k}-(v_{\theta})_{k}\right\}= (\mu_1)_{k}\,\omega_3^1(q_k)+(\mu_2)_{k}\,\omega_3^2(q_k),\;\;
(p_{\phi})_{k}-\frac{J}{h}\left\{\phi_{k}-(v_{\phi})_{k}\right\}= 0,
\end{split}
\]
together with $(v_{k+1},q_{k+1})\in \Delta_Q^{d-}$, namely, 
\[
\begin{split}
\frac{x_{k+1}-(v_{x})_{k+1}}{h}-R\cos\phi_{k}\frac{\theta_{k+1}-(v_{\theta})_{k+1}}{h}=0,\;\; 
\frac{y_{k+1}-(v_{y})_{k+1}}{h}-R\sin\phi_{k}\frac{\theta_{k+1}-(v_{\theta})_{k+1}}{h}=0.
\end{split}
\]
Therefore, the $(-)$-discrete Lagrange--Dirac dynamical system \eqref{minus_DiscLagDiracSys} is simplified by eliminating 
the variables $v_k$ and $p_k$, and becomes 
\begin{equation*}
\begin{aligned}
&\frac{m}{h}(x_{k+1}-2x_k+x_{k-1})=(\mu_1)_{k},\quad \frac{m}{h}(y_{k+1}-2y_k+y_{k-1})=(\mu_2)_{k},\\
& \frac{I}{h}(\theta_{k+1}-2\theta_k+\theta_{k-1})+h 10 \cos \theta_{k}=-(\mu_1)_{k}\,R\cos\phi_k-(\mu_2)_{k}\,R\sin\phi_k,\\
& \frac{J}{h}(\phi_{k+1}-2\phi_k+\phi_{k-1})=0,
\end{aligned}
\end{equation*}
subject to the constraints 
\begin{equation*}
\begin{aligned}
\frac{x_{k+1}-x_k}{h}-R(\cos\phi_k)\left(\frac{\theta_{k+1}-\theta_k}{h}\right)=0,\;
\frac{y_{k+1}-y_k}{h}-R(\sin\phi_k)\left(\frac{\theta_{k+1}-\theta_k}{h}\right)=0.
\end{aligned}
\end{equation*}

\paragraph{Numerical tests.}
In the following numerical results, the step is set to be $h=0.001$ while the time interval is chosen as $[0,50]$.
We use the following parameters $m=1.0$, $r=1.0$, $I=0.25$, $J=0.5$.  The initial conditions are set to be
\begin{equation*}
 (x_{0},y_{0},\phi_{0}, \theta_{0})=\left(0.0,0.0,\frac{\pi}{3},0\right),
\end{equation*}
and $(x_{1},y_{1},\phi_{1}, \theta_{1})$ are determined such that $\theta_1 =\theta_0 + h (v_{\theta})_{0}=0.01$, 
$\phi_1 =\phi_0 + h  (v_{\phi})_{0}=1.0481976$, $x_1 = x_0 + r  (\theta_k-\theta_{k+1}) \cos\phi_0=0.005$ and 
$y_1 = z_0 + r  (\theta_k-\theta_{k+1}) \sin\phi_0=0.0086603$.

We illustrate the trajectory of the vertical coin that is obtained by the $(+)$-discrete Lagrange--Dirac dynamical system in Fig. \ref{fig:vcoin_trajectory+} and also show the energy behavior of $E^+_d$ in 
Fig. \ref{fig:vcoin_energy+} and we also show the discrete path of the position $x$ and the momentum $p_{\phi}$ 
respectively in Fig. \ref{fig:vcoin_xt+} and Fig.  \ref{fig:vcoin_pphit+}.

\begin{figure}[htbp]
\centering
\;\;
\begin{minipage}{7cm}
\includegraphics[width=7cm]{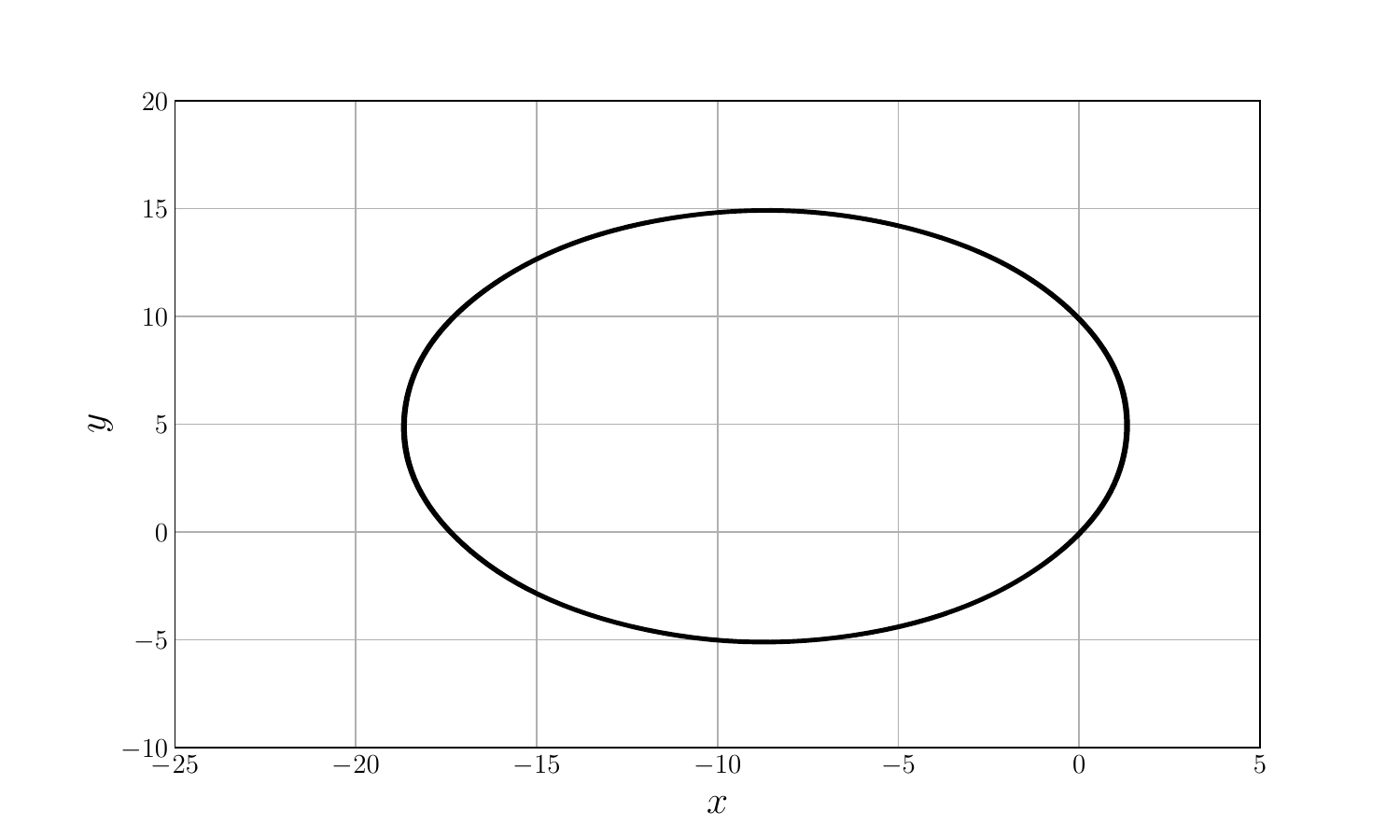}
\caption{Trajectory $(x(t), y(t))$ on a plane.}
\label{fig:vcoin_trajectory+}
\end{minipage}
\begin{minipage}{7cm}
\includegraphics[width=7cm]{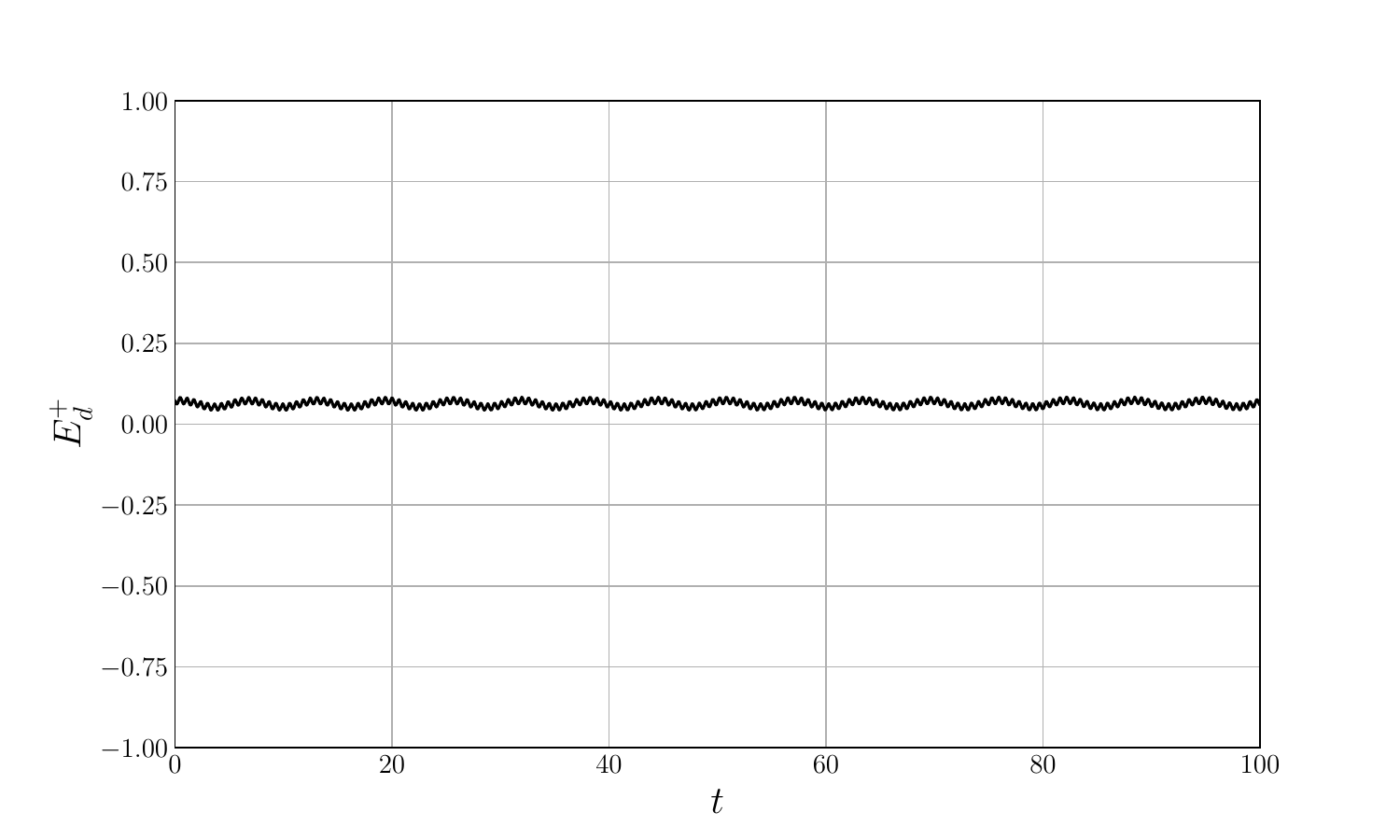}
\caption{Energy behavior $E_d^{+}(t)$.}
\label{fig:vcoin_energy+}
\end{minipage}
\end{figure}
\begin{figure}[htbp]
\centering
\;\;
\parbox{7cm}{
\includegraphics[width=7cm]{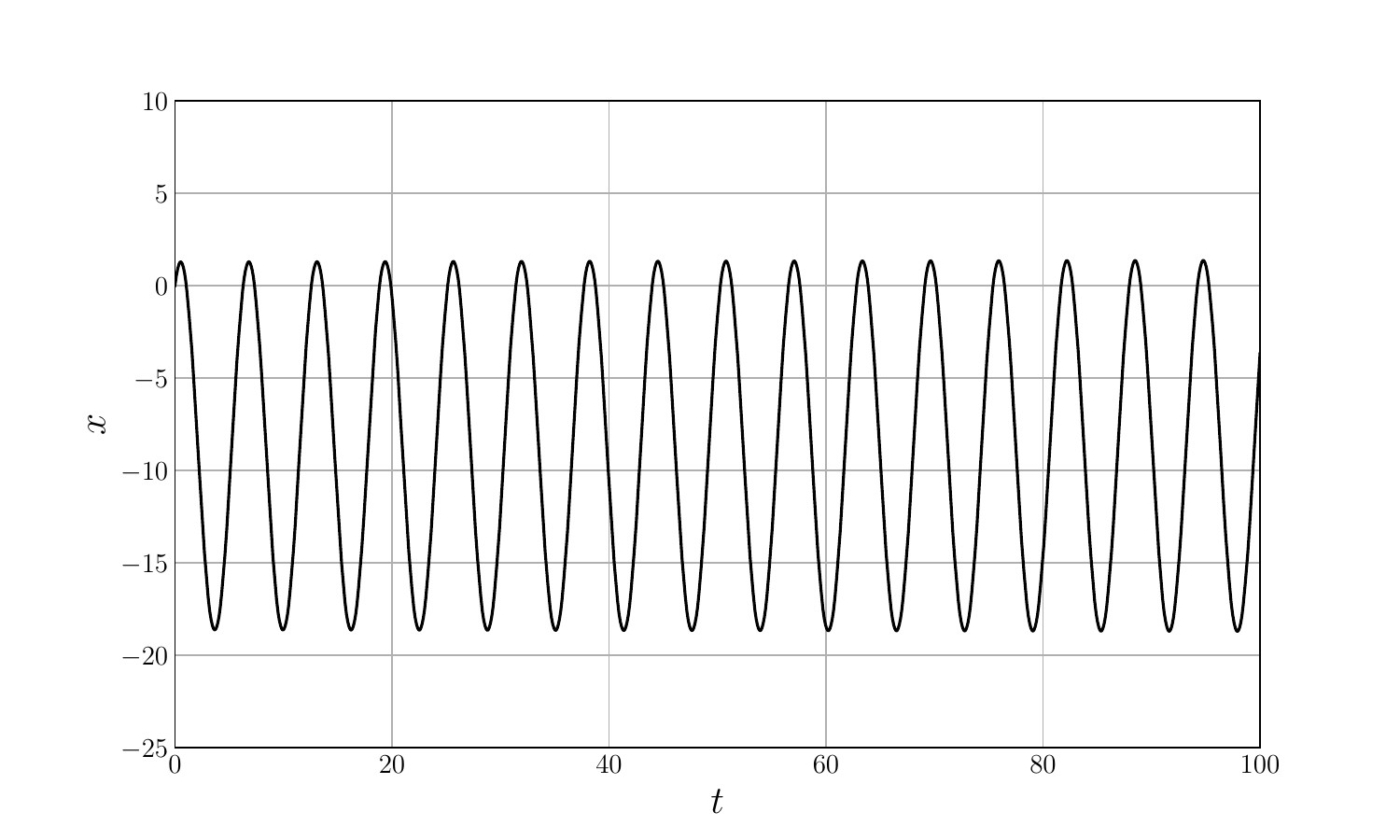}
\caption{Discrete path $x(t)$.}
\label{fig:vcoin_xt+}
}
\begin{minipage}{7cm}
\includegraphics[width=7cm]{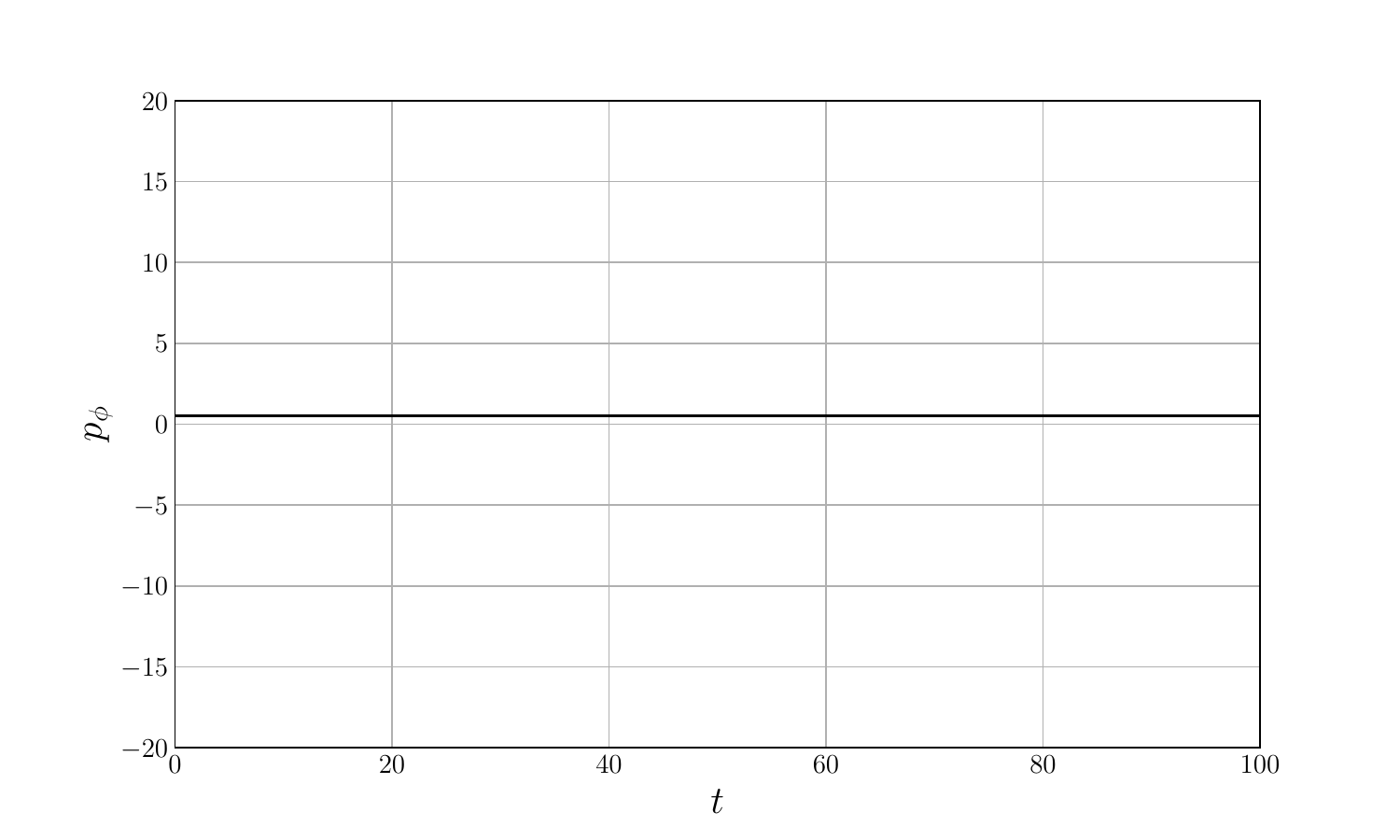}
\caption{Momentum $p_{\phi}(t)$.}
\label{fig:vcoin_pphit+}
\end{minipage}
\end{figure}

We also illustrate the trajectory of the vertical coin that is obtained by the $(-)$-discrete Lagrange--Dirac dynamical system in Fig. \ref{fig:vcoin_trajectory-}, and  show the energy behavior of $E^-_d$ in Fig. \ref{fig:vcoin_energy-} and the discrete paths of the position $x$ and the momentum $p_{\phi}$ 
respectively in Fig. \ref{fig:vcoin_xt-} and Fig. \ref{fig:vcoin_pphit-}. 

\begin{figure}[htbp]
\centering
\;\;
\begin{minipage}{7cm}
\includegraphics[width=7cm]{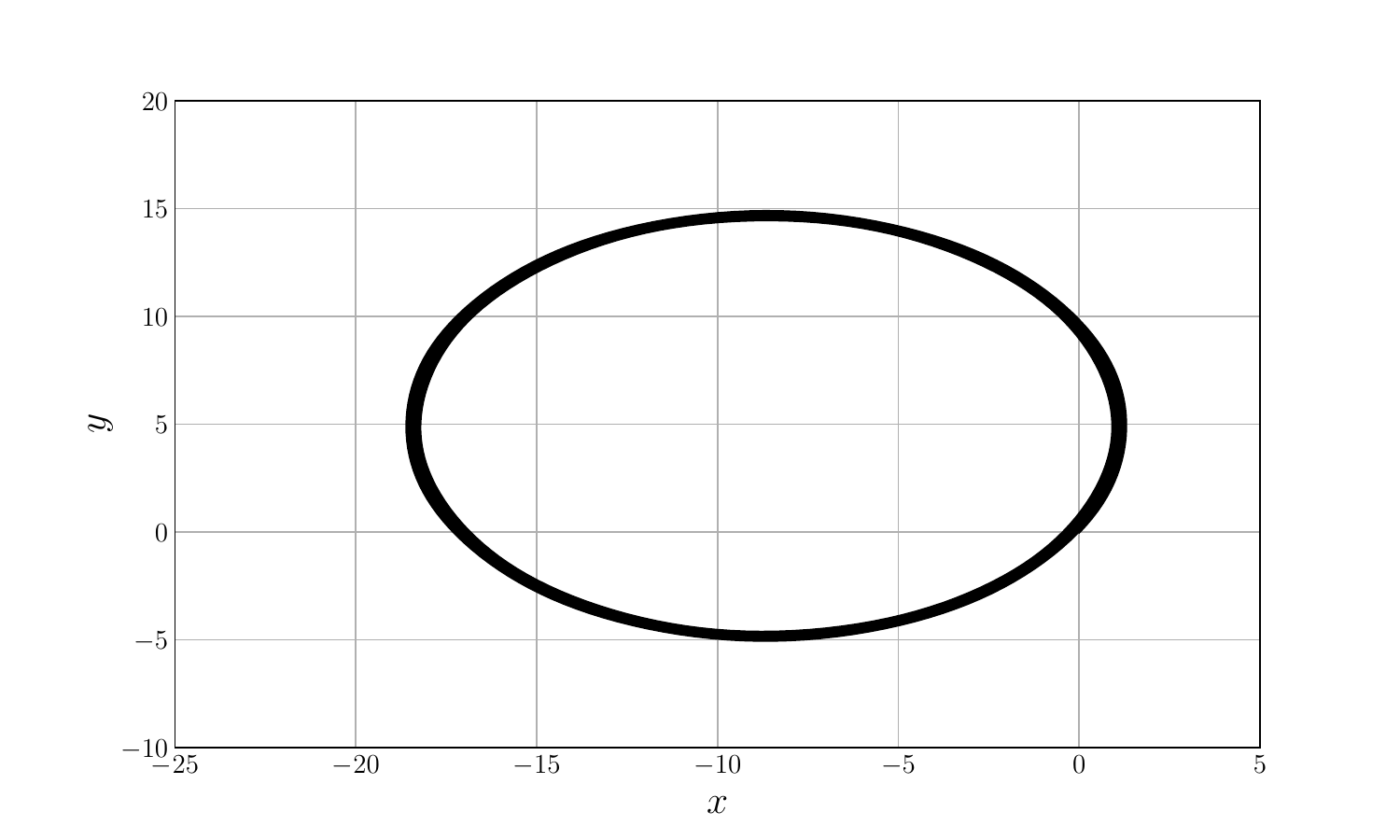}
\caption{Trajectory $(x(t), y(t))$ on a plane.}
\label{fig:vcoin_trajectory-}
\end{minipage}
\begin{minipage}{7cm}
\includegraphics[width=7cm]{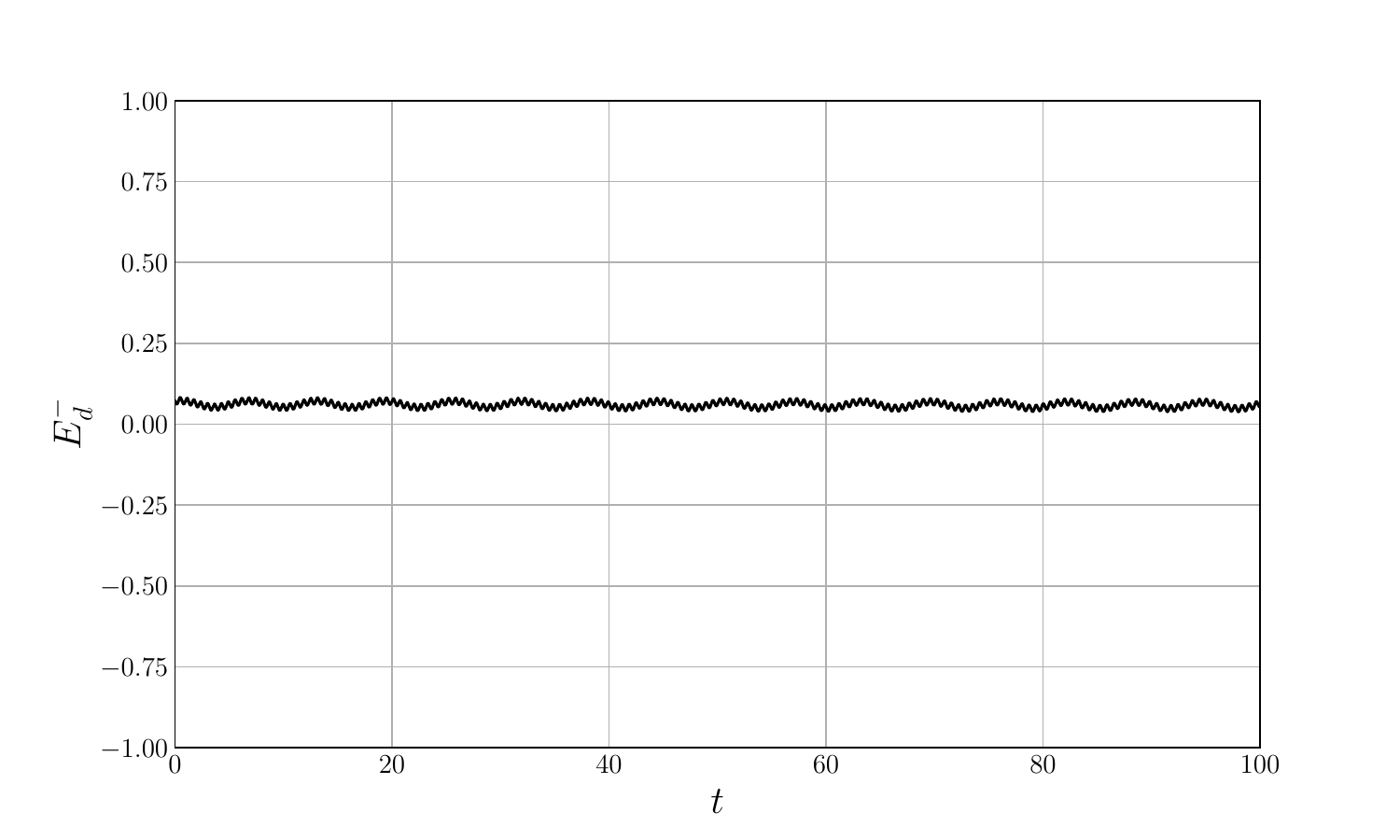}
\caption{Energy behavior $E_d^{-}(t)$.}
\label{fig:vcoin_energy-}
\end{minipage}
\end{figure}
\begin{figure}[htbp]
\centering
\;\;
\begin{minipage}{7cm}
\includegraphics[width=7cm]{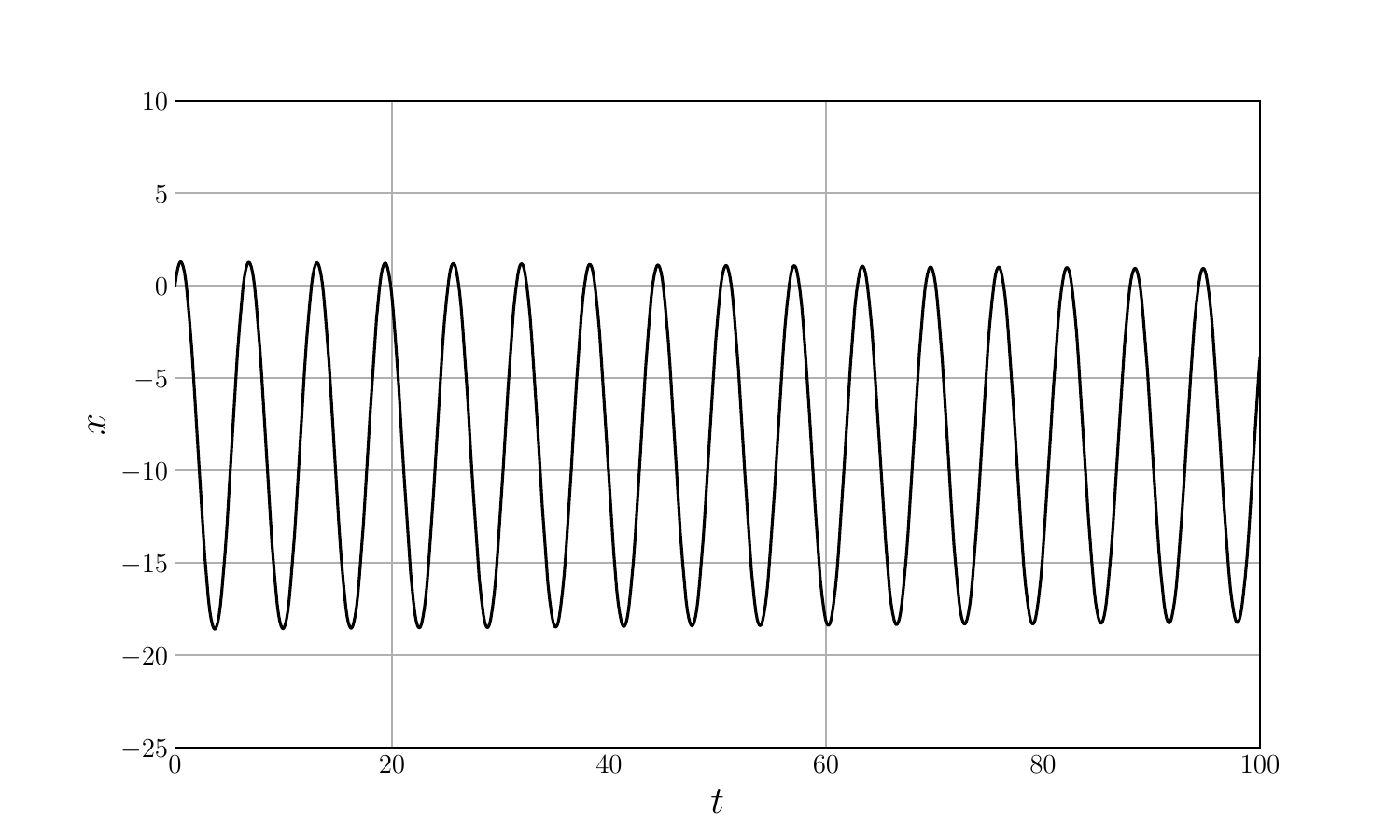}
\caption{Discrete path $x(t)$.}
\label{fig:vcoin_xt-}
\end{minipage}
\begin{minipage}{7cm}
\includegraphics[width=7cm]{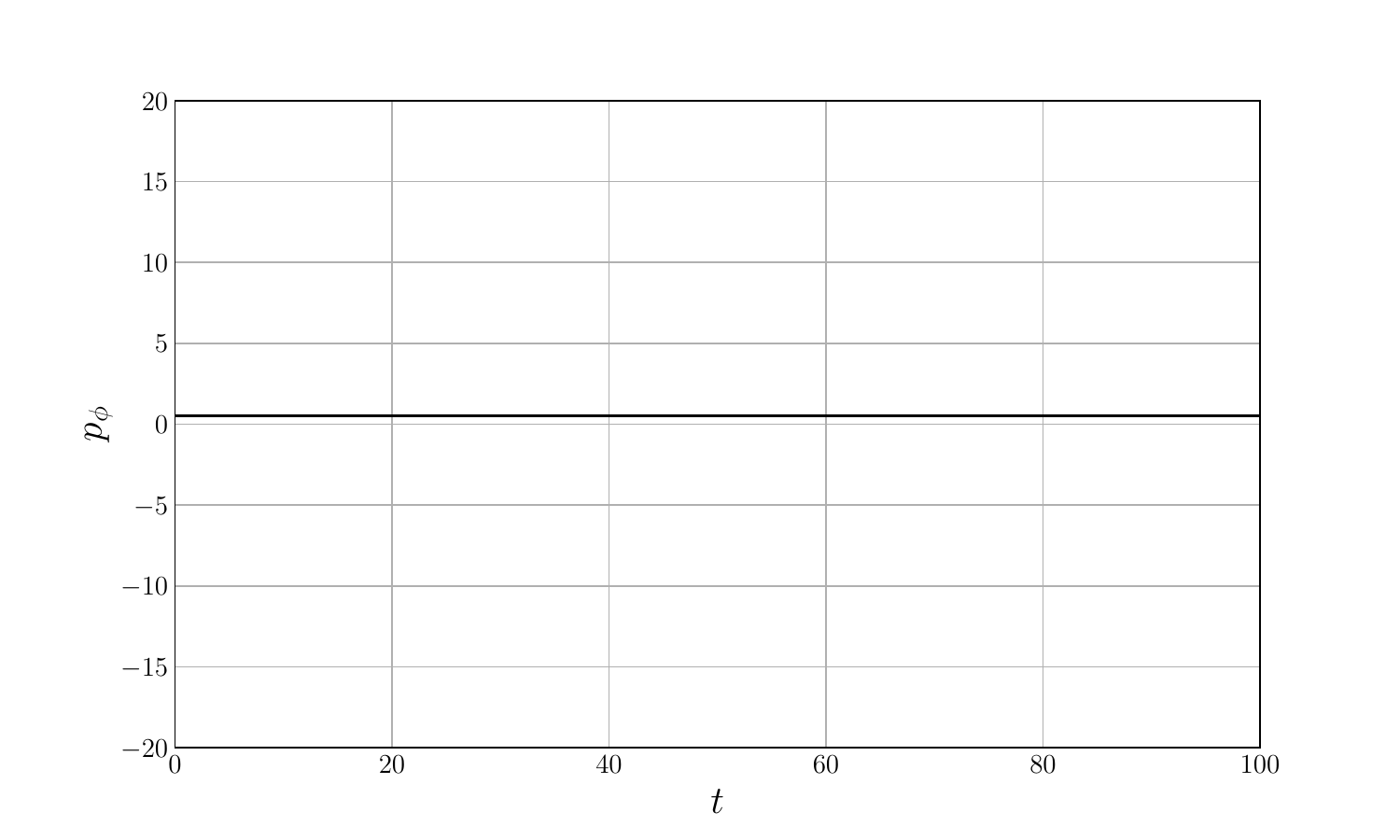}
\caption{Momentum $p_{\phi}(t)$.}
\label{fig:vcoin_pphit-}
\end{minipage}
\end{figure}

Figs. \ref{fig:vcoin_energy+} and \ref{fig:vcoin_energy-} exhibit good energy behavior, while Figs. \ref{fig:vcoin_pphit+} and \ref{fig:vcoin_pphit-} demonstrate the conservation of momentum $p_\phi$ well.

\subsection{The classical  Heisenberg system}
\paragraph{Continuous setting.} As another example, let us consider a classical Heisenberg system whose Lie algebra is similar to the Heisenberg algebra in quantum mechanics; see, e.g., \cite{Bl2003}. The configuration space is given by $Q=\mathbb{R}^3$, with local coordinates $q=(x,y,z) \in Q$, and the system has a Lagrangian on $TQ$ given by
\begin{equation}
L(q(t),\dot{q}(t))=\frac{1}{2}\left(\dot{x}^2(t)+\dot{y}^2(t)+\dot{z}^2(t)\right).
\end{equation}
The motion $q(t) \in Q, \; t \in [0, T] \subset \mathbb{R}^{+}$ of the system is subject to the nonholonomic constraint as 
$
\dot{q}=(\dot{x}, \dot{y}, \dot{z}) \in \Delta_{Q}(q) \subset T_{q}Q,
$
in which $\Delta_{Q}\subset TQ$ is the constraint distribution given by
\[
\Delta_{Q}(q)=\left\{\dot{q} \in T_{q}Q \mid \left< \omega(q), \dot{q} \right>=0 \right\},
\]
where $\omega(q)=dz-ydx+x dy$ is a differential form on $Q$. The nonholonomic constraint can be rewritten in local coordinates as
\begin{equation*}
\dot{z}(t)=y(t)\dot{x}(t)-x(t)\dot{y}(t),
\end{equation*}
whose annihilator $\Delta_{Q}^{\circ} \subset T^{\ast}Q$ is 
\[
\Delta_{Q}^{\circ}(q):=\left\{ \alpha \in T^{\ast}_{q}Q \mid \left<\alpha, \dot{q} \right>=0, \;\; \forall \dot{q} \in \Delta_{Q}(q)\right\},
\]
$\alpha \in \Delta_{Q}^{\circ}(q)$ is locally represented by using a Lagrange multiplier $\mu$ as $\alpha=\mu(dz-ydx+x dy)$.

Hence, the Lagrange--Dirac dynamical system in \eqref{LDAPeqn} is given by
\begin{equation*}
\begin{split}
&\dot{x}=v_{x},\quad \dot{y}=v_{y},\quad \dot{z}=v_{z},\quad p_{x}= v_{x}, \quad p_{y}= v_{y}, \quad p_{z}= v_{z}, \\
&\dot{p}_{x}=-\mu y,\quad \dot{p}_{y}=\mu x,\quad \dot{p}_{z}=\mu,\quad
\dot{z}=y\dot{x}-x\dot{y}.
\end{split}
\end{equation*}
Direct calculation shows that the energy $\frac{1}{2}\left(\dot{x}^2(t)+\dot{y}^2(t)+\dot{z}^2(t)\right)$ is a constant of motion. 
 
\paragraph{The $(+)$-discrete Lagrange--Dirac dynamical system.}
The discrete Lagrangian is given, in coordinates $(q_{k},v_{k}) \in Q \times Q$, by
\[
\begin{split}
L_{d}(q_{k},v_{k})&=h L\left(q_{k}, \frac{v_{k}-q_k}{h}\right)\\
&=\frac{h}{2}\left\{\left(\frac{(v_{x})_{k}-x_k}{h}\right)^{2}+
\left(\frac{(v_{y})_{k}-y_k}{h}\right)^{2} +\left(\frac{(v_{z})_{k}-z_k}{h}\right)^{2}\right\}.
\end{split}
\]
By using \eqref{plus_DiscLagDiracSys}, we get the equations of motion for the $(+)$-discrete Lagrange--Dirac dynamical system as
\[
\begin{split}
&(p_{x})_{k+1}=\frac{1}{h}\left\{(v_{x})_{k}-x_{k}\right\},\;\; (p_{y})_{k+1}=\frac{1}{h}\left\{(v_{y})_{k}-y_{k}\right\},\;\; (p_{z})_{k+1}=\frac{1}{h}\left\{(v_{z})_{k}-z_{k}\right\},\\
&(p_{x})_{k}-\frac{1}{h}\left\{(v_{x})_{k}-x_{k}\right\}= -\mu_{k}\,y_{k},\;\;
(p_{y})_{k}-\frac{1}{h}\left\{(v_{y})_{k}-y_{k}\right\}= \mu_{k}\, x_{k},\\
&(p_{z})_{k}-\frac{1}{h}\left\{(v_{z})_{k}-z_{k}\right\}= \mu_{k},\;\; x_{k+1}=(v_{x})_{k},\;\; y_{k+1}=(v_{y})_{k},\;\; z_{k+1}=(v_{z})_{k},\\
\end{split}
\]
together with $(q_{k},v_{k})\in \Delta_Q^{d+}$, namely, 
\[
\begin{split}
\frac{z_{k+1}-z_k}{h}=y_{k+1}\frac{x_{k+1}-x_k}{h}-x_{k+1}\frac{y_{k+1}-y_k}{h}.
\end{split}
\]
Thus, the $(+)$-discrete Lagrange--Dirac system are immediately obtained as follows:
 \begin{equation*}
 \begin{aligned}
& \frac{2x_k-x_{k-1}-x_{k+1}}{h}=-\mu_{k} y_k,\quad \frac{2y_{k}-y_{k-1}-y_{k+1}}{h}=\mu_{k} x_k,\quad \frac{2z_{k}-z_{k-1}-z_{k+1}}{h}=\mu_{k},\\
&\frac{z_{k+1}-z_k}{h}=y_{k+1}\frac{x_{k+1}-x_k}{h}-x_{k+1}\frac{y_{k+1}-y_k}{h}.
 \end{aligned}
 \end{equation*}
 
 \paragraph{The ($-$)-discrete Lagrange--Dirac dynamical system.}
The discrete Lagrangian is given in local coordinates $(v_{k+1},q_{k+1}) \in Q \times Q$ by
\[
\begin{split}
L_{d}&(v_{k+1},q_{k+1})=h L\left(v_{k+1},\frac{q_{k+1}-v_{k+1}}{h}\right)\\
&=\frac{h}{2}\left\{\left(\frac{x_{k+1}-(v_{x})_{k+1}}{h}\right)^{2}+\left(\frac{y_{k+1}-(v_{y})_{k+1}}{h}\right)^{2}+\left(\frac{z_{k+1}-(v_{z})_{k+1}}{h}\right)^{2}\right\}.
\end{split}
\]
By using \eqref{minus_DiscLagDiracSys}, the equations of motion for the $(-)$-discrete Lagrange--Dirac dynamical system are obtained as
\[
\begin{split}
&(p_{x})_{k}=\frac{1}{h}\left\{x_{k+1}-(v_{x})_{k+1}\right\},\;\; (p_{y})_{k}=\frac{1}{h}\left\{y_{k+1}-(v_{y})_{k+1}\right\},\;\; (p_{z})_{k}=\frac{1}{h}\left\{z_{k+1}-(v_{z})_{k+1}\right\},\\
&(p_{x})_{k}-\frac{1}{h}\left\{x_{k}-(v_{x})_{k}\right\}= -\mu_{k}\,y_{k},\;\;
(p_{y})_{k}-\frac{1}{h}\left\{y_{k}-(v_{y})_{k}\right\}= \mu_{k}\, x_{k},\\
&(p_{z})_{k}-\frac{1}{h}\left\{z_{k}-(v_{z})_{k}\right\}= \mu_{k},\;\;(v_{x})_{k+1}=x_{k},\;\; (v_{y})_{k+1}=y_{k},\;\; (v_{z})_{k+1}=z_{k},
\end{split}
\]
together with $(q_{k},v_{k})\in \Delta_Q^{d-}$, namely, 
\[
\begin{split}
\frac{z_{k+1}-z_k}{h}=y_{k}\frac{x_{k+1}-x_k}{h}-x_{k}\frac{y_{k+1}-y_k}{h}.
\end{split}
\]
Thus, the $(-)$-discrete Lagrange--Dirac system are immediately obtained as follows:
 \begin{equation*}
 \begin{aligned}
& \frac{2x_k-x_{k-1}-x_{k+1}}{h}=-\mu_{k} y_k,\quad \frac{2y_{k}-y_{k-1}-y_{k+1}}{h}=\mu_{k} x_k,\quad \frac{2z_{k}-z_{k-1}-z_{k+1}}{h}=\mu_{k},\\
&\frac{z_{k+1}-z_k}{h}=y_k\frac{x_{k+1}-x_k}{h}-x_k\frac{y_{k+1}-y_k}{h}.
 \end{aligned}
 \end{equation*}

Note that the obtained equations of motion for the $(\pm)$-discrete Lagrange--Dirac dynamical systems are identical to each other. It can be shown by direct calculation that the corresponding discrete  energy, the one obtained by applying the Legendre transforms to the discrete generalized energies $E_d^{\pm}$, 
\begin{equation*}
E_d(x_k,y_k,z_k,x_{k+1},y_{k+1},z_{k+1})=\frac{1}{h}\left((x_{k+1}-x_k)^2+(y_{k+1}-y_k)^2+(z_{k+1}-z_k)^2\right),
\end{equation*}
is  a constant of motion.

 \paragraph{Numerical tests.}
In the following numerical results, the step is set to be $h=0.01$ while the time interval is chosen as $[0,1000]$.
We use the same parameters as before except the initial condition. The initial condition is chosen to be
\begin{equation}\label{ic001}
 (x_{0},y_{0},z_{0})=(1.0,0.0,0.1),\quad (x_{1},y_{1},z_{1})=(1.05,0.1,0.0).
\end{equation}
We show the discrete path of the position  $(x(t),y(t),z(t))$ of the particle in Fig. \ref{fig:Heisenberg_position} and the  velocity $(v_x(t),v_y(t),v_z(t))$ in Fig. \ref{fig:Heisenberg_velocity}. We also present the  momentum $(p_x(t),p_y(t),p_z(t))$ in Fig. \ref{fig:Heisenberg_momentum} and the energy behaviors in Fig. \ref{fig:Heisenberg_enegy}, each of which illustrates that the momentum and the energy are well preserved, respectively.

\begin{figure}[htbp]
\centering
\;\;
\begin{minipage}{7cm}
\includegraphics[width=7cm]{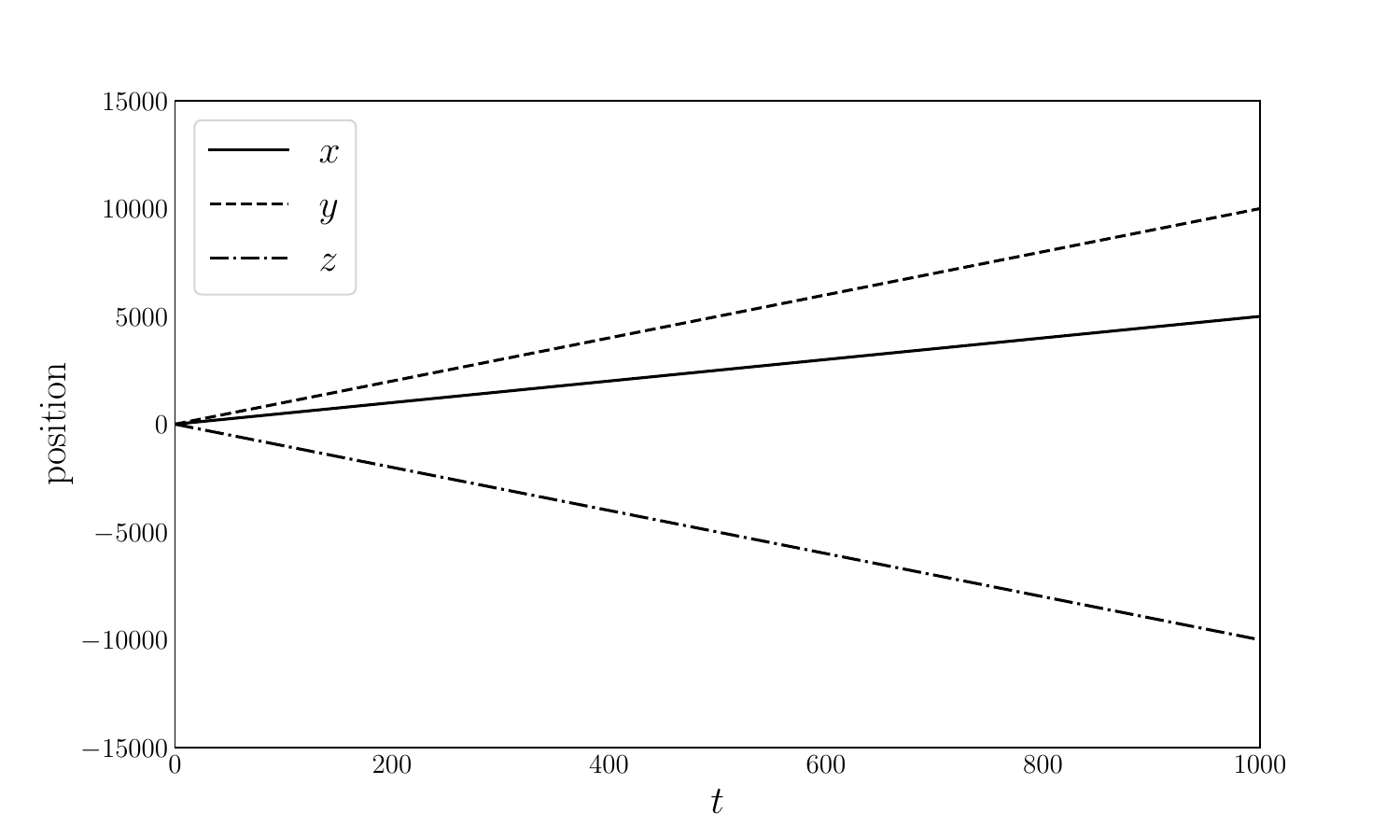}
\caption{Position $(x(t),y(t),z(t))$.}
\label{fig:Heisenberg_position}
\end{minipage}
\begin{minipage}{7cm}
\includegraphics[width=7cm]{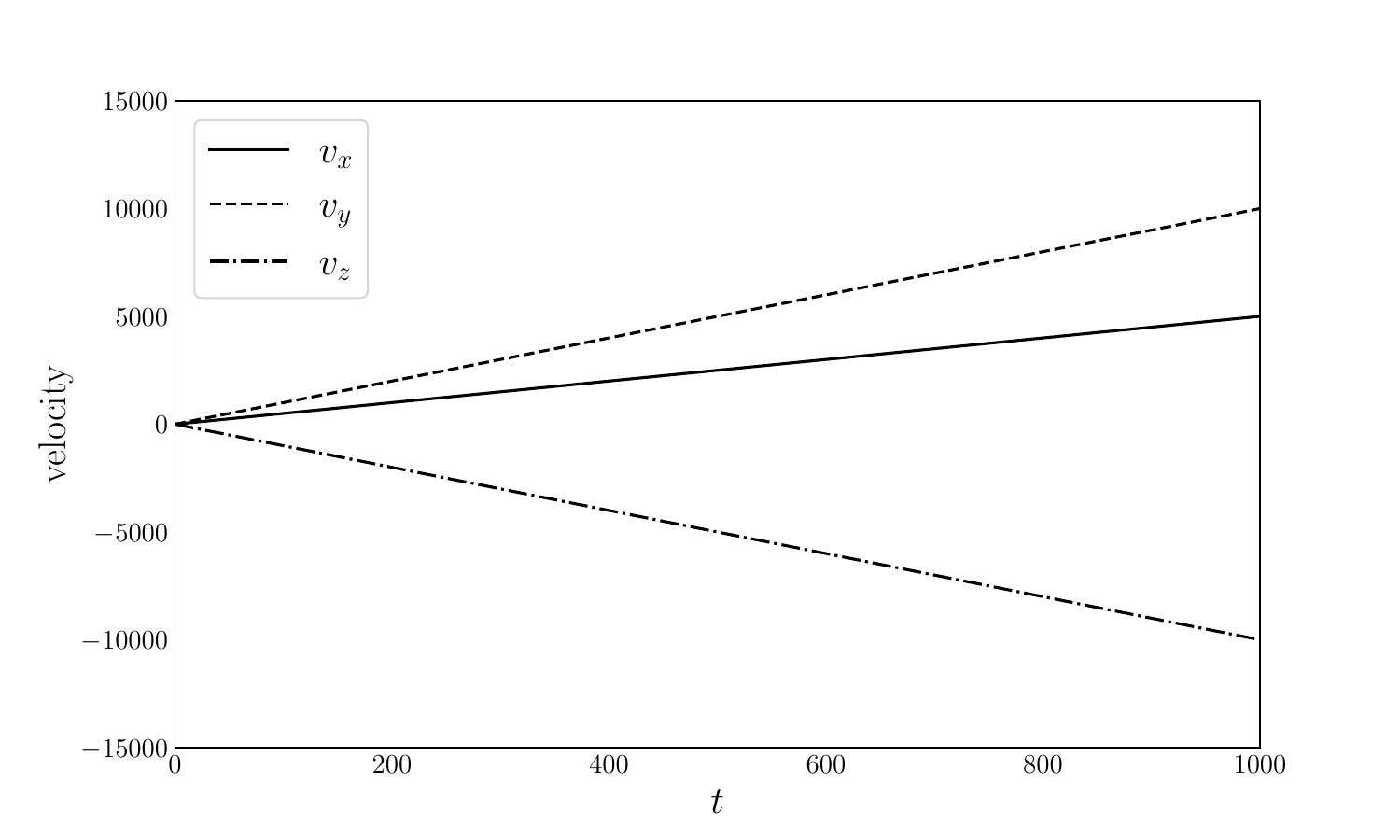}
\caption{Velocity $(v_x(t),v_y(t),v_z(t))$.}
\label{fig:Heisenberg_velocity}
\end{minipage}
%
\;\;
\begin{minipage}{7cm}
\includegraphics[width=7cm]{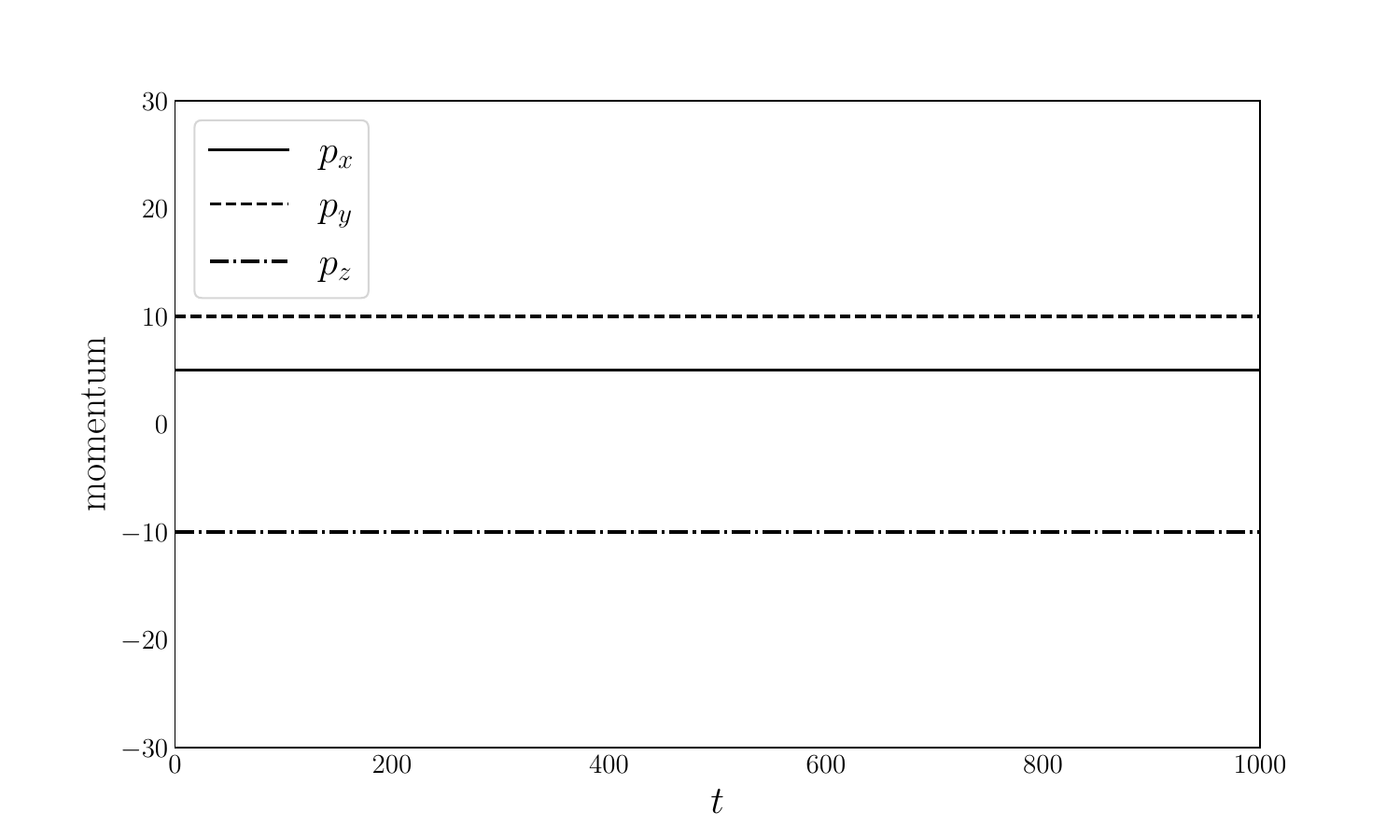}
\caption{Momentum $(p_x(t),p_y(t),p_z(t))$.}
\label{fig:Heisenberg_momentum}
\end{minipage}
\begin{minipage}{7cm}
\includegraphics[width=7cm]{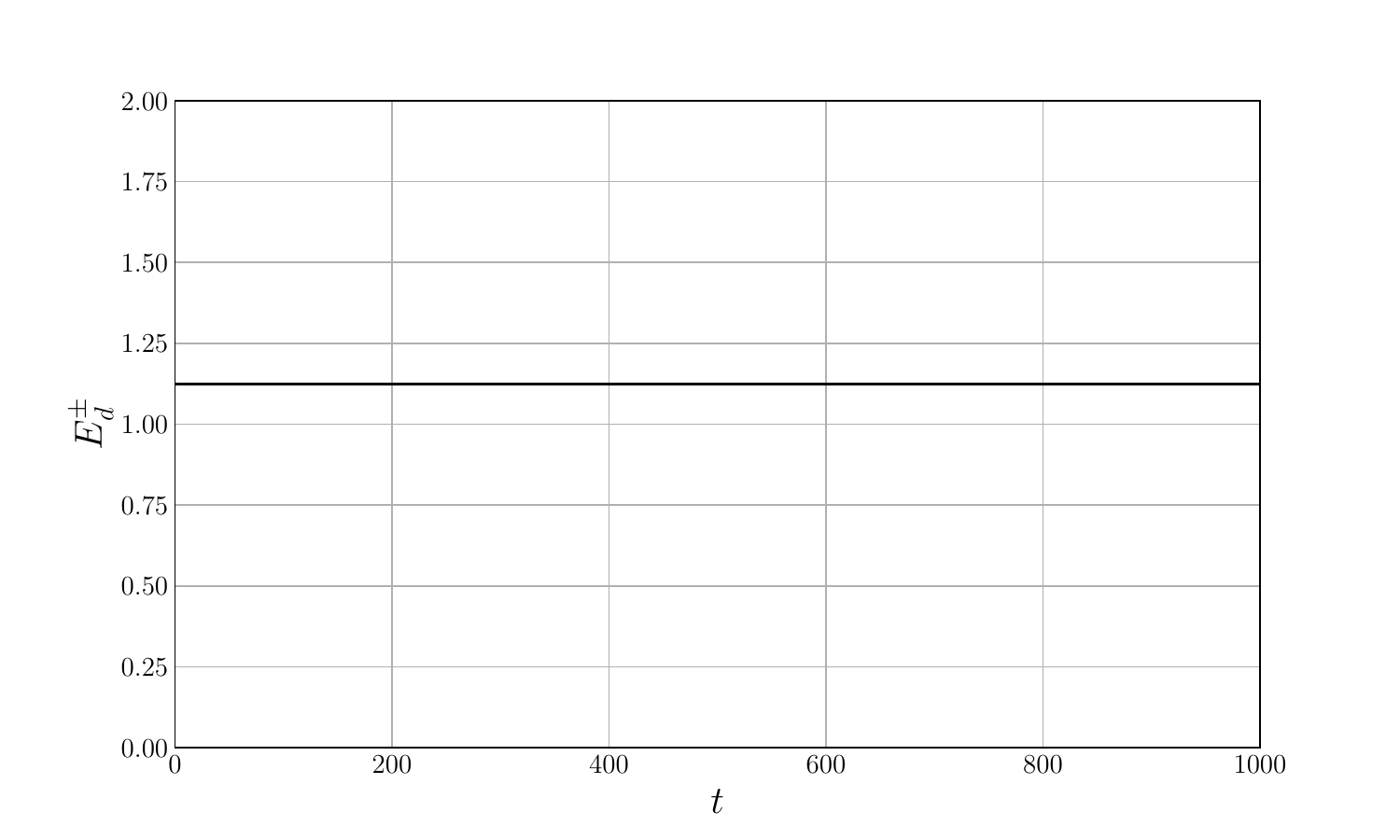}
\caption{Energy behaviors $E_d^{\pm}(t)$.}
\label{fig:Heisenberg_enegy}
\end{minipage}
\end{figure}

\section{Conclusions}
In this paper, we have proposed $(\pm)$-discrete Dirac structures $D_M^{d\pm}$ on a manifold $M$, which are induced from the discrete two-forms $\Omega^{d\pm}_{M}$ and the discrete constraint spaces $\Delta^{d\pm}_{M}$. These structures serve as discrete analogues of the Dirac structure $D_M$ on $M$. Specifically, we  have developed the $(\pm)$-discrete induced Dirac structure $D_{\Delta_Q}^{d\pm} \subset (T^\ast Q \times T^\ast Q) \oplus T^\ast T^\ast Q$ over the cotangent bundle $T^\ast Q$, corresponding to the continuous Dirac structure $D_{\Delta_Q} \subset TT^\ast Q \oplus T^\ast T^\ast Q$ described in \cite{YoMa2006a}. Here, the $(\pm)$-discrete canonical symplectic one-forms $\Theta^{d\pm}_{T^\ast Q}$ are introduced, and the $(\pm)$-discrete canonical symplectic structures $\Omega^{d\pm}_{T^\ast Q}$ are defined as $\Omega^{d\pm}_{T^\ast Q} = -\mathbf{d}\Theta^{d\pm}_{T^\ast Q}$.

The key idea involves using $(\pm)$-finite difference maps to obtain forward and backward finite difference approximations for $Q \times Q$, $T^\ast (Q \times Q)$, and $T^\ast Q \times T^\ast Q$, corresponding to $TQ$, $T^\ast TQ$, and $TT^\ast Q$, respectively. The base manifold $Q$ is consistently assigned from $Q \times Q$ using the $(\pm)$-Legendre transform $\mathbb{F}^\pm L_d: Q \times Q \to T^\ast Q$.

For the discrete Lagrange--Dirac dynamical systems, we have demonstrated the discrete analogue of the bundle structure between $T^\ast T^\ast Q$, $TT^\ast Q$, and $T^\ast TQ$, establishing a discrete diffeomorphism between $T^\ast T^\ast Q$, $T^\ast Q \times T^\ast Q$, and $T^\ast (Q \times Q)$. Then, we have derived the $(\pm)$-discrete Dirac differential of the discrete Lagrangian and have illustrated the $(\pm)$-discrete Lagrange--Dirac dynamical system using $(\pm)$-discrete evolution maps. Furthermore, we have shown that the discrete Lagrange--d'Alembert equations developed by \cite{CoMa2001, McPe2006} can be recovered from the $(\pm)$-discrete Lagrange--d'Alembert--Pontryagin equations. Finally, we have validated our theory of discrete Lagrange--Dirac dynamical systems through illustrative examples, including the vertical rolling disk on a plane and the classical Heisenberg system.
\bigskip

There are some interesting topics for future work that may be relevant with this paper as follows:
\begin{itemize}
\item Development of discrete Hamilton--Dirac dynamical systems: We do not know how the discrete Hamilton systems can be understood in the context of the $(\pm)$-discrete Dirac structures. In particular, we will explore how the symplectic Euler-A and Euler-B methods can be understood in the framework of discrete Hamilton--Dirac dynamical systems. 
\item Extension to the infinite dimensional Dirac dynamical systems: The exploration of this paper has been restricted to the finite dimensional cases, but there may be natural extensions to the infinite dimensional case of the $(\pm)$-discrete Lagrange--Dirac dynamical systems.   
\item Discretization of the Lie--Dirac reduction and the associated reduced discrete Lagrange--Dirac systems: We will develop the discrete analogue of the Euler--Poincar\'e--Suslov reduction and the Lie--Poisson--Dirac reduction (\cite{YoMa2009}).  
\item Applications to nonequilibrium thermodynamic systems: The variational discretization of simple closed nonequilibrium thermodynamic systems was explored in \cite{GBYo2018b}, while we do not understand how such a simple closed nonequilibrium thermodynamic system can be understood in the context of the $(\pm)$-discrete Lagrange--Dirac dynamical systems.  

\end{itemize}

\noindent \paragraph{Acknowledgements.} 
We are very grateful to Kenshin Okuwaki at Waseda University for his helpful support in numerical simulations. The first author's research is partially supported by JSPS Grant-in-Aid for Scientific Research (20K14365, 24K06852) and by Keio University (Fukuzawa Memorial Fund, KLL). The second author's research is partially supported by JST CREST (JPMJCR1914), JSPS Grant-in-Aid for Scientific Research (22K03443) and Waseda University Grants for Special Research Projects(SR 2022C-423, 2023C-089, SR 2024C-102). 

\end{document}